\documentclass[11pt, twoside]{article} % 'twoside' enables distinct odd/even pages

\usepackage[a4paper, margin=.7in]{geometry} 

\usepackage{fancyhdr}

\newcommand{\papercitation}{T\'{o}th et al. (2026)}
\newcommand{\shorttitle}{WFOMC over Ordered Domains}

\makeatletter
\let\@fnsymbol\@arabic
\makeatother

\newcommand{\singleauthor}[3]{%
    \textsc{#1}, \textit{#2}\thanks{\texttt{#3}} \\[1ex]
}

\title{Weighted First-Order Model Counting over Ordered Domains}
\author{
    \singleauthor{Jan T\'oth}{Czech Technical University in Prague}{tothjan2@fel.cvut.cz}
    \singleauthor{Qipeng Kuang}{The University of Hong Kong}{kuangqipeng@connect.hku.hk}
    \singleauthor{Kuncheng Zou}{Jilin University}{zoukc25@mails.jlu.edu.cn}
    \singleauthor{V\'aclav K\r{u}la}{Czech Technical University in Prague}{kulavacl@fel.cvut.cz}
    \singleauthor{Yuyi Wang}{CRRC Zhuzhou Institute \& Tengen Intelligence Institute}{yuyiwang920@gmail.com}
    \singleauthor{Yuanhong Wang}{Jilin University}{lucienwang@jlu.edu.cn}
    \singleauthor{Ond\v{r}ej Ku\v{z}elka}{Czech Technical University in Prague}{ondrej.kuzelka@fel.cvut.cz}
}   
\date{}

\usepackage{natbib}
\setcitestyle{authoryear,round}
\usepackage[T1]{fontenc}
\usepackage{lmodern}

\usepackage{amssymb}

\providecommand{\Description}[1]{}
\providecommand{\rev}[1]{#1}

\usepackage{amsmath}
\usepackage{amsthm}
\usepackage{bm}
\usepackage{braket}
\usepackage[basic]{complexity}
\usepackage{nicefrac}               % compact symbols for 1/2, etc.
\usepackage{thm-restate}
\usepackage{xspace}

\usepackage[linesnumbered,ruled,vlined]{algorithm2e}

\usepackage{booktabs}               % professional-quality tables
\usepackage{braket}
\usepackage{multirow}
\usepackage{subcaption}
\usepackage{tabulary}

\usepackage{graphicx}
\usepackage{tikz}
\usetikzlibrary{calc, decorations.pathreplacing, arrows.meta, positioning, patterns}
\usepackage{pgffor}

\usepackage{hyperref}
\usepackage[capitalise]{cleveref}
\usepackage{soul}
\usepackage[normalem]{ulem}
\usepackage[dvipsnames]{xcolor}

\usepackage{lipsum} % for dummy text

\newtheorem{theorem}{Theorem}
\newtheorem{lemma}{Lemma}
\newtheorem{corollary}{Corollary}

\newtheorem{definition}{Definition}
\newtheorem{remark}{Remark}
\newtheorem{example}{Example}

\newcommand{\fomc}{FOMC\xspace}
\newcommand{\symbfomc}{\ensuremath{\mathsf{FOMC}}}
\newcommand{\wfomc}{WFOMC\xspace}
\newcommand{\symbwfomc}{\ensuremath{\mathsf{WFOMC}}}
\newcommand{\wmc}{WMC\xspace}
\newcommand{\symbwmc}{\ensuremath{\mathsf{WMC}}}

\newcommand{\fomodels}[2]{\ensuremath{\mathcal{M}_{#1, #2}}}
\newcommand{\dom}{\Delta}
\newcommand{\negw}{\overline{w}}
\newcommand{\weights}{w, \negw}
\newcommand{\pred}[1]{\mathsf{pred}\left( #1 \right)}

\newcommand{\fo}{\ensuremath{\mathbf{FO}}}
\newcommand{\fotwo}{\ensuremath{\mathbf{FO}^2}\xspace}
\newcommand{\ctwo}{\ensuremath{\mathbf{C}^2}\xspace}

\newcommand{\fotwoformula}{\psi}

\newcommand{\nat}{\mathbb{N}}
\newcommand{\real}{\mathbb{R}}

\newcommand{\vecdelta}{\bm{\delta}}

\newcommand{\hb}{HB\space}

\newcommand{\preds}[1]{\ensuremath{\mathcal{P}_{#1}}}

\newcommand{\loaxiom}{\mathcal{L}}
\newcommand{\lonat}{\mathcal{L}_\nat}
\newcommand{\gridaxiom}{\mathcal{G}}
\newcommand{\succaxiom}{\mathcal{S}}
\newcommand{\lopred}{\ensuremath{\leq}}
\newcommand{\acyclicityaxiom}{\mathcal{A}}

\newcommand{\rlo}{\widehat{r}}
\newcommand{\wlo}{\widehat{w}}
\newcommand{\rpred}{\widetilde{r}}
\newcommand{\evidence}[1]{\ensuremath{\textup{Evidence}[#1]}}
\newcommand\tiling[1]{\mathsf{11NMCT}_{(#1)}}
\newcommand{\segment}[2]{{#1}\rightsquigarrow{#2}}

\begin{document}

\maketitle

\begin{abstract}
    The Weighted First-Order Model Counting Problem (WFOMC) asks to compute the weighted sum of models of a given first-order logical sentence over a given domain.
    WFOMC is a fundamental problem in the area of statistical relational learning, as many inference problems in that field are reducible to it.
    Its applications extend further, including a general framework for analyzing and automatically solving problems in enumerative combinatorics, as well as contributions to graph polynomials.
    A long stream of research has investigated the tractability boundary of WFOMC with respect to the domain size.
    It is known that computing WFOMC for the three-variable fragment of first-order logic is \class{\#P_1}-hard, while polynomial-time algorithms exist for computing WFOMC for the two-variable fragment and its extensions by cardinality constraints, counting quantifiers, and also certain axioms.
    In this work, we explore the possibility of computing WFOMC in polynomial time over linearly ordered domains.
    Having such a capability would allow us to model and reason in a tractable manner across various inference scenarios and combinatorial problems involving sequences.
    Since encoding the linear order using standard first-order logic would require three variables, which would immediately negate our polynomial-time aspirations, we add a linear order axiom to our first-order language instead.
    The linear order axiom forces one of the predicates to impose a linear (total) ordering on the domain elements in each model, and this is a modeling construct already studied in formal logic.
    We start with a positive result, proving that WFOMC with the linear order axiom can be solved in time polynomial in the domain size.
    We proceed by extending the result to ordered domains with access to the successor relations.
    While that is easy to show when we explicitly define the successors in terms of the linear order, we demonstrate an alternative approach with successors being defined implicitly (i.e., the successor relations are considered a part of the linear order axiom), which exhibits much better performance on all tested instances, sometimes even offering exponential runtime improvements.
    We conclude by analyzing the scenario with two distinct linear orders.
    Unfortunately, we show that WFOMC over the two-variable fragment with two linear order relations is \class{\#P_1}-hard.
    Nevertheless, we also develop a polynomial-time algorithm for WFOMC with one linear order relation and a successor relation of another linear order, pushing the intractability barrier further and leaving an open question of how close to the fully-fledged second linear order we can get.
\end{abstract}

\newpage
\section{Introduction}
We study the Weighted First-Order Model Counting Problem (WFOMC), which asks to compute the weighted sum of models of a given first-order logic sentence over a specified domain, along with a pair of weighting functions that assign a weight to each model.
WFOMC serves as a fundamental problem in statistical relational learning (SRL), where we often need to perform inference over models defined both in terms of first-order theories and probabilities \citep{getoor07:srl-intro}.
For computational reasons, we try to perform the inference without grounding the first-order sentences involved.
Techniques attempting to avoid the grounding procedure are commonly referred to as \emph{lifted inference} \citep{broecketal21:lifting-intro}.
The significance of WFOMC stems from the discovery that it can be used to perform lifted inference over models as diverse as Markov Logic Networks \citep{richardson06:mlns}, parfactor graphs \citep{poole03:fove}, probabilistic logic programs \citep{deraedt15:plp}, and probabilistic databases \citep{cavallo87:pdb}.
For example, a social network encoding the smoking habit transition over individuals could be represented by a sentence $\phi = \forall x\forall y: sm(x) \land fr(x,y) \to sm(y)$, and the WFOMC of $\phi\land sm(c)$ for any domain element $c$ would give us the (unnormalized) probability of $c$ being a smoker.
Recent works also reveal the potential of WFOMC to contribute to enumerative combinatorics by providing a general framework for encoding counting problems or integer sequences, and computing graph polynomials on specific graphs \citep{svatos23:fluffy,kuangetal24:wfomc-polynomials}.

However, it is unlikely to find an algorithm computing WFOMC in time polynomial in the size of the sentence since WFOMC encodes the well-known \class{\#P}-complete problem \#SAT, even if the symbols of relations are fixed \citep{beame15:wfomc-fo3}.
Consequently, most studies focus on the complexity in terms of the domain size, which is analogous to the concept of \emph{data complexity} in database theory \citep{vardi82:datacomplexity}.
The logical fragments (subsets of full first-order logic) enjoying polynomial time complexity in the domain size are called \emph{domain-liftable} \citep{broeck11:wfomc-ufo2}.

Previous results have shown that the two-variable fragment (\fotwo) is domain-liftable \citep{broeck11:wfomc-ufo2,broecketal14:wfomc-skolem}, while the three-variable fragment ($\fo^3$) is not, unless \class{\#P_1 \subseteq FP} \citep{beame15:wfomc-fo3}.
Here, \class{\#P_1} defined by \citet{valiant79:counting-complexity} is the class of counting problems whose input is in unary.
The work of \citet{kuzelka21:wfomc-c2} extended the domain-liftability to the fragment of \ctwo, i.e., the fragment of \fotwo extended by counting quantifiers and possibly cardinality constraints.

To further advance domain-liftability, recent studies have focused on augmenting \ctwo with \emph{additional axioms} which may not necessarily be finitely expressible in first-order logic.
Starting from the result in \citet{kuusisto18:function-axiom} showing that \fotwo remains domain-liftable with the addition of one functionality axiom,%
\footnote{Such a logical fragment is actually a subset of \ctwo; however, domain-liftability of the entire \ctwo was not yet established at the time.}
 subsequent results have identified various tractable axioms such as the tree axiom \citep{bremen23:tree-axiom}, the connectedness axiom \citep{malhotra23:wfomc-axioms}, and the acyclicity axiom \citep{malhotra23:wfomc-axioms}.
\citet{kuangetal24:wfomc-polynomials} later proposed a general approach to prove domain-liftability of axioms expressible by graph polynomials (e.g., the bipartite axiom and the strong connectedness axiom) and even their combinations (e.g., the combination of the acyclicity and the connectedness axioms on a single relation).

Note that in many applications of lifted inference, the domain of interest naturally follows a total order, such as time series data or sequences of events~\citep{le20:probabilistic-time-series,vlasselaer14:dynamic-relational-models}, which leads to a need for tractable inference over ordered domains.
For instance, consider the sentence $\forall x\forall y: sm(x) \land older(x,y) \to sm(y)$ which aims to capture that the smoking habits of an older person may influence a younger one.
The relation $older$ essentially orders the domain of people we work with by age.

In this work, we continue the study of domain-liftability, focusing on scenarios with ordered domains.
We start by extending \ctwo with the \emph{linear order axiom} (i.e., we force one distinguished binary relation in our language to define a total ordering of the domain).
Note that studying such a logical fragment is not beneficial only to the domain of lifted inference, since it has also received attention from logicians \citep{charatonik15:logic-c2+lo} and, recently, the graph neural networks community as well \citep{hauke25:gnns-c2+lo}.
We prove \ctwo with the linear order axiom to be another domain-liftable fragment by providing a polynomial-time algorithm for computing WFOMC over it.
However, as it will turn out, accessing a pair of consecutive elements or a pair of elements with $k$ elements in between (both of which are reasonable modeling constructs on ordered domains, e.g., when modeling hidden Markov models) may still incur a considerable computational overhead.
Thus, we continue by presenting an \emph{extended linear order axiom} specifying not only the linear order relation itself but also its successor relations.
We also devise a specialized algorithm that significantly reduces the overhead of accessing the predecessors or successors of a domain element.
While the immediate successor relation can be useful in many applications, the $k$-th successor relation is less common.
However, one potential application is to encode a $k \times n$ \emph{grid} structure within the WFOMC framework.
As a consequence, many probabilistic inference tasks on grids, e.g., the 2-dimensional Ising model with constant interaction strength, which have been shown to be equivalent to WFOMC~\citep{broecketal21:lifting-intro}, can be proved to be tractable with our algorithm if the grid's dimension $k$ is bounded by a constant.

Finally, we address whether two linear order axioms can also be domain-liftable.
That is a natural question arising from results on decidability of the finite satisfiability problem, which state that \ctwo with one linear order is decidable while \ctwo with two linear orders is not \citep{charatonik15:logic-c2+lo}.
However, such a question also goes beyond most of the research in the area of lifted inference to date.
All existing results only allow axioms on a single distinguished relation.
The techniques involved in the works outlined above, such as the Matrix-Tree Theorem for the tree axiom \citep{bremen23:tree-axiom}, the recursion formulas for the connectedness axiom and the acyclicity axiom \citep{malhotra23:wfomc-axioms}, and polynomials for WFOMC with axioms \citep{kuangetal24:wfomc-polynomials}, cannot be generalized to axioms on multiple relations trivially.
In investigating the domain-liftability of two linear orders, we lay the foundations for answering questions about the boundary of domain-liftability for WFOMC in the presence of axioms on two distinguished binary relations.

This article is an extended version of two previously published conference papers.
Domain-liftability of \ctwo with the linear order axiom was first published as ``Lifted Inference with Linear Order Axiom, AAAI'23'' and the efficient algorithm for accessing successor relations was published as ``Faster Lifting for Ordered Domains with Predecessor Relations, ECAI'25''.
We extend those contributions by analyzing domain-liftability in the presence of two linear order axioms, as well as their weakened (successor axiom) and strengthened (acyclicity axiom) counterparts.
The new contributions are available as a standalone electronic preprint at \url{https://arxiv.org/abs/2508.11515}.

\subsection{Our Contributions}

We consider WFOMC over \ctwo with linear order axioms and successor axioms on one or two distinguished binary relations, and obtain both positive and negative results on their computational complexity in terms of domain-liftability.

\begin{itemize}
    \item First, {\bf we prove that the fragment of \ctwo with a linear order axiom on a single binary relation is domain-liftable}.
    We present an algorithm computing WFOMC over the aforementioned logical fragment in time polynomial in the domain size.

    \item Second, we note that rather than the domain ordering itself, successor relationships defined implicitly by the ordering may be more important for modeling purposes.
    However, accessing the immediate successors, or even the $k$-th successors, using only the linear order relation, while domain-liftable, may incur excessive computational overhead.
    {\bf We devise another algorithm which natively supports access to a domain element's immediate successor or even the $k$-th successor, where $k$ is a constant with respect to the domain size}.

    \item Third, {\bf we prove that WFOMC over \fotwo with two linear order relations is \class{\#P_1}-hard}.
    We obtain the hardness in three steps.
    First, we construct a variant of the tiling problem, and show that it is \class{\#P_1}-hard to compute.
    Then, we introduce another (intermediate) axiom, specifically the \emph{grid axiom}, which encodes a general grid.
    We then encode WFOMC with the grid axiom as WFOMC with two linear order axioms.
    Since a linear order axiom can in fact be encoded using the acyclicity axiom,%
    \footnote{We note that while the domain-liftability of the acyclicity axiom (which subsumes the linear order axiom) has been established in \citet{malhotra23:wfomc-axioms}, the results on linear order presented here predate that development and offer a direct, specialized algorithmic approach.}
    requiring a relation to represent a directed acyclic graph, we also prove that WFOMC for \fotwo{} with two acyclic relations is \class{\#P_1}-hard.
    Clearly, since $\fotwo \subset \ctwo$, the negative results also hold in the presence of counting quantifiers.
    
    \item Last, {\bf we prove that \ctwo with a linear order relation and a successor relation (of another linear order) is domain-liftable}.
    We present yet another algorithm solving WFOMC over such a fragment in time polynomial in the domain size.
    Combining that with our first two positive results also implies domain-liftability of \ctwo with a linear order relation, its successor relations, and another successor relation, possibly with cardinality constraints.
\end{itemize}

Our theoretical results indicate the boundary of domain-liftability for WFOMC with linear order axioms.
From the perspective of the number of linear order axioms, our results clearly show that we can achieve domain-liftability with one linear order but not with two.
From the perspective of the power of linear order axioms, our results show that the intractability of two linear order axioms arises from the extra information obtained somewhere beyond the successor relation of the second linear order.

\subsection{Related Work}
Our work builds on a long stream of research in the area of lifted inference, investigating the tractability of WFOMC.
Primarily, we base our work on the results of \citet{broeck11:wfomc-ufo2,broecketal14:wfomc-skolem,kuzelka21:wfomc-c2} establishing the domain-liftability of \ctwo and \citet{beame15:wfomc-fo3} providing a hardness result for $\fo^3$.
Apart from the complexity results themselves, we also build all of our algorithms on the same principles as the so-called \emph{domain recursion rule} \citep{broeck11:wfomc-ufo2,kazemietal16:new-liftable-classes}, originally used to prove the domain-liftability of \fotwo.
We continue in a string of techniques based on augmenting the language of WFOMC with various axioms, such as the functionality axiom \citep{kuusisto18:function-axiom}, the tree axiom \citep{bremen23:tree-axiom}, the acyclicity and connectedness axioms \citep{malhotra23:wfomc-axioms}, or any axiom expressible using graph polynomials \citep{kuangetal24:wfomc-polynomials}.

Another closely related line of research is the exploration of decidability in the finite satisfiability problem for first-order logics, whose frontier also falls between \fotwo{} and $\fo^3$ due to the decidability of the former fragment \citep{graedel97:fo2-sat} and the undecidability of latter fragment \citep{moore62:fo3-unsat}.
It has also been proved that \fotwo{} with 8 linear order axioms \citep{otto01:unsat-fo2+8order}, \fotwo{} with three linear order axioms \citep{kieronski11-unsat-fo2+lo}, \fotwo{} with two linear order axioms and their successor axioms \citep{manuel10:unsat-fo2+2successor} and \ctwo{} with two linear order axioms \citep{charatonik15:logic-c2+lo} are all undecidable.
We also utilize the encoding of an undecidable tiling problem using a two-variable first-order sentence with multiple axioms that restrict the structures to be grid-like, i.e., a trick commonly used in finite satisfiability \citep{otto01:unsat-fo2+8order,kieronski11-unsat-fo2+lo,schwentick12:unsat-fo2+2order,kieronski05:unsat-fo2+3eq,kieronski12-unsat-fo2+2eq}

Last but not least, we also focus on the practical scalability of our polynomial-time algorithms.
In that regard, we relate to efforts to design efficient algorithms for computing \wfomc, such as the works of \citet{vanbremen21:fast-wfomc,meng24:recursive-wfomc,kidambi25:practical-fomc}.

\section{Background}
\label{sec:2-background}
Throughout this paper, we denote the set of natural numbers from 1 to $n$ as $[n]$.
We assume the set of all natural numbers $\nat$ to contain zero, and we also assume that $0^0=1$.

We use boldface letters such as $\bm{k}$ to denote a vector and $k_i$ to denote the $i$-th element of the given vector $\bm{k}$.
Since all of our vectors only have non-negative integer entries, we denote the sum of vector elements by $|\bm{k}|$, i.e., given a vector $\bm{k}$ of length $d$, we have $|\bm{k}| = \sum_{i=1}^d k_i$.
When it comes to operating on vectors, we use the standard element-wise addition and subtraction, i.e., $\bm{a} + \bm{b} = (a_1 + b_1, a_2 + b_2, \dots)$ and $\bm{a} - \bm{b} = (a_1 - b_1, a_2 - b_2, \dots)$.
We use $\vecdelta_i$ to denote the $i$-th basis vector, where the $i$-th element is $1$, and the rest are $0$.
For a non-negative integer vector $\bm{k} = (k_1, k_2,\ldots,k_d)$ with $|\bm{k}|=n$, we use $\binom{n}{\bm{k}}$ to denote its multinomial coefficient, i.e.,
\begin{align*}
   \binom{n}{\bm{k}} = \binom{n}{k_1, k_2,\ldots,k_d} = \frac{n!}{\prod_{i=1}^{d} k_i!}.
\end{align*}

As we will also study complexity, recall that we define complexity classes in terms of Turing machines, often abbreviated as TM \citep{turing36:tm,book74:languages-and-complexities}.
Let us review a few properties of TMs that are critical to our proofs; we leave the formal definitions of a TM and its computation to more specialized literature, e.g., \citet{hopcroft06:automata-theory}.
A TM is a theoretical model of computation that takes place on a \emph{tape} or possibly on several tapes.
Each tape consists of infinitely many discrete cells, each containing one symbol from either the input alphabet or a set of special \emph{tape symbols} which necessarily include the \emph{blank (empty) symbol}.
Without loss of generality, we may assume that the tapes are \emph{one-sided}, meaning that each tape has a starting cell and spans to infinity on only one side.
Each tape also has an associated \emph{head} determining which position on the tape is currently being read.
Apart from the currently read symbols, TM also has a finite set of states.
A transition mapping then determines the subsequent action of a TM at a particular state, given the specific symbols read by each head.
The transition specifies a new state of the TM, changes to each currently read cell, and the movement of each head.
If the transition mapping is a function, we call the TM \emph{deterministic}, otherwise we refer to the machine as \emph{nondeterministic}.
The time complexity of a TM is then the total number of transitions before the machine halts, expressed as a function of the input length.

The complexity of a particular problem $P$ may not be directly described in terms of a TM.
We may instead start with a problem $P'$ whose hardness has already been established.
Then, we \emph{reduce} $P'$ of known complexity to the problem $P$, i.e., we solve $P'$ using an oracle for $P$ (we also say that we \emph{encode} $P'$ using $P$).
If the reduction does not increase the complexity (usually, we aim for a polynomial-time reduction), we prove that $P$ is at least as hard as $P'$.

\subsection{First-Order Logic with Axioms}
We consider a function-free fragment of first-order logic.
The language is defined by a finite set of \emph{variables}, a finite set of \emph{constants}, and a finite set of \emph{predicates} (also called \emph{relations}).
If we have a predicate symbol $P$ of arity $k$, we also write $P/k$.
A formula $\alpha$ is then defined inductively as
\begin{align*}
    \alpha ::= P(x_1, \dots, x_k) \mid (\alpha \land \alpha) \mid (\alpha \lor \alpha) \mid (\neg \alpha) \mid (\alpha \to \alpha) \mid (\alpha \leftrightarrow \alpha) \mid (\exists x: \alpha) \mid (\forall x: \alpha),
\end{align*}
where $P/k$ is a predicate symbol and the \emph{terms} $x_1, \dots, x_k$ are logical variables or constants.%
\footnote{Whenever possible, without introducing ambiguities, we will drop additional parentheses in logical formulas for readability purposes.}

We call the formula in the form $P(x_1, \dots, x_k)$ an \emph{atom}, and an atom or its negation a \emph{literal}.
A formula is called a \emph{sentence} if quantifiers bind all variables in the formula.
A formula is called \emph{ground} if it contains no variables.

Given a first-order logic formula $\phi$, we write $\preds{\phi}$ for the set of all predicates that appear in $\phi$ (i.e., the \emph{vocabulary of $\phi$}).
To define truth, we adopt definitions from the so-called \emph{Herbrand logic} \citep{hinrichs06:herbrand-logic}.
The Herbrand universe $\dom$ is the set of constants.
The Herbrand base \hb is the set of all ground atoms that can be formed from the predicates in $\mathcal{P}_{\phi}$ and the constants in $\dom$.
A \emph{possible world} (also called an \emph{interpretation}) $\omega$ is an arbitrary subset of the \hb.
Elements of $\omega$ are interpreted to be true, and all other elements from the \hb are interpreted to be false.
The satisfaction relation $\omega \models \phi$ is then defined in the usual way.
If $\omega$ satisfies a sentence $\phi$, it is called a \emph{model} of $\phi$.
We denote $\fomodels{\phi}{n}$ as the set of all models of the sentence $\phi$ over the domain $[n]$.
Note that the interpretation of a binary relation $R$ can be regarded as a directed graph $G(R)$ where the domain is the vertex set and the true ground literals of $R$ are the edges.

In this work, we are especially interested in the following fragments of first-order sentences:
Sentences with at most two logical variables which form the so-called \fotwo fragment, and
\fotwo sentences with \emph{counting quantifiers} (i.e., $\exists^{=k}$, $\exists^{\le k}$ and $\exists^{\ge k}$) which are then called \ctwo sentences.
Counting quantifiers generalize the traditional existential quantifier, which only counts \emph{at least one}. The quantifier $\exists^{=k}$ restricts the number of assignments of the quantified variable satisfying the subsequent formula to exactly $k$; similarly for $\exists^{\le k}$ and $\exists^{\ge k}$.
Sentences from the fragments above may be further augmented with \emph{cardinality constraints}, which are expressions of the form $(|P| \bowtie k)$ where $P$ is a relation and $\bowtie$ is a comparison operator from $\{<, \le, =, \ge, >\}$.
Such constraints are imposed on the number of distinct positive ground literals of $P$ in a model.
Note that while cardinality constraints and counting quantifiers may be easily interchangeable in some cases, there are many scenarios where we benefit from having both these syntactic constructs.

\begin{example}
Consider the following formulas:
\begin{align*}
    \phi_1 &= (\forall x \exists y: A(x) \to B(x, y)) \land (|A| = 3),\\
    \phi_2 &= (\forall x: \neg E(x, x)) \land (\forall x \exists^{=1} y: E(x, y)).
\end{align*}

$\phi_1$ contains a cardinality constraint requiring exactly three atoms on the predicate $A$ to be present in each of its models, i.e., $\omega_1 = \{A(1), A(2), A(4), B(1, 1), B(2, 1), B(4, 1)\}$ is a model of $\phi_1$, while $\omega_1' = \{A(1), B(1, 1), B(2, 1)\}$ is not one of its models, even though $\omega_1'$ is still a model of $(\forall x \exists y: A(x) \to B(x, y)$).

On the other hand, $\phi_2$ contains a counting quantifier. The subformula $(\forall x \exists^{=1} y: E(x, y))$ essentially requires the $E/2$ relation to be a function, i.e., every element from the domain must be mapped to exactly one element. Together with the first part of the formula, we can interpret $\phi_2$ as a directed graph without loops such that each vertex has exactly one outgoing edge.

Note that the cardinality constraint $(|A| = 3)$ could be easily replaced by a formula with a counting quantifier, specifically $(\exists^{=3} x: A(x))$.
However, the counting subformula in $\phi_2$ does not permit such a simple transformation to cardinality constraints.
\end{example}

Any fragment we work with can additionally be extended by some \emph{axioms}.
An \emph{axiom} is a special constraint on a binary relation $R$ requiring the graph $G(R)$ to be a specific combinatorial structure.
Some axioms explored in other works include the \emph{tree axiom} \citep{bremen23:tree-axiom}, requiring $G(R)$ to be a directed rooted tree, or the acyclicity axiom \citep{malhotra23:wfomc-axioms,kuangetal24:wfomc-polynomials}, requiring the graph to be acyclic.
Formally speaking, when talking about axioms, we should say that there is ``an axiom on a distinguished binary relation'', e.g., there is an acyclicity axiom on the distinguished binary relation $R$ (which restricts $G(R)$ to be an acyclic graph), however, we often abbreviate that statement and talk about ``a constrained relation'', e.g., we have an acyclic relation $R$.

\subsection{Weighted-First Order Model Counting}
The \emph{first-order model counting} problem (\fomc) asks to compute the number of models of a first-order sentence over the domain of a given size.

\begin{definition}[First-Order Model Counting]
\label{def:fomc}
The \fomc of a first-order sentence $\Psi$ over a finite domain of size $n$ is defined as
\begin{equation*}
  \symbfomc(\Psi, n) = |\fomodels{\Psi}{n}|.
\end{equation*}
\end{definition}

Note that given \cref{def:fomc}, we are computing isomorphic models as distinct.

\begin{example}
    Consider the sentence $\phi = \exists x: P(x)$.
    
    For an arbitrary $n\in\nat$, any non-empty subset of $\hb = \{P(1), P(2), \ldots, P(n)\}$ is an element of the set \fomodels{\phi}{n}.
    Therefore, we have $$\symbfomc(\phi, n) = 2^n - 1.$$
    Specifically, $\omega_1 = \{P(1)\}$ and $\omega_2 = \{P(2)\}$ are distinct models, and each contributes 1 to the overall model count.
\end{example}

The \emph{weighted first-order model counting} problem (\wfomc) additionally expects a pair of weighting functions $(\weights)$, both mapping relations in $\mathcal{P}_{\Psi}$ to real weights.%
\footnote{In general, the weights $(\weights)$ could map the predicates to any ring including complex numbers, such as in \citet{kuzelka21:wfomc-c2}, and polynomials. Our assumption of real weights is mainly to achieve shorter descriptions.}
Given a set $L$ containing literals with predicates from $\mathcal{P}_{\Psi}$, the weight of $L$ is defined as
\begin{equation*}
W(L, \weights):= \prod_{l \in L_T}w(\pred{l}) \cdot \prod_{l \in L_F}\negw(\pred{l}),
\end{equation*}
where $L_T$ (resp. $L_F$) denotes the set of positive (resp. negative) literals in $L$, and $\pred{l}$ maps a literal $l$ to its corresponding predicate.
We omit the symbols $\weights$ and write $W(L)$ for brevity when the weighting functions are apparent from the context.

\begin{example}
\label{ex:model-weight}
Consider the sentence $\Gamma = \forall x \forall y: (S(x) \to R(x,y))$ and the weighting functions $w(S) = 3, w(R) = 2, \negw(S) = \negw(R) = 1$.

The weight of the literal set
$L = \{S(1), \lnot S(2), R(1,1), R(1,2), R(2,1), \lnot R(2,2)\}$
is
\begin{equation*}
  W(L) = w(S) \cdot \negw(S) \cdot \left(w(R)\right)^3  \cdot \negw(R) = 24.
\end{equation*}
\end{example}

\begin{definition}[Weighted First-Order Model Counting]
    \label{def:wfomc}
    The \wfomc of a first-order sentence $\Psi$ over a finite domain of size $n$ under weighting functions $(\weights)$ is defined as
    \begin{equation*}
        \wfomc(\Psi, n, \weights) := \sum_{\mu \in \fomodels{\Psi}{n}} W(\mu, \weights).
    \end{equation*}
\end{definition}

\begin{example}
Consider the sentence $\Gamma = \forall x \forall y: (S(x) \to R(x,y))$ along with the same weighting functions as in \cref{ex:model-weight} again.
Then, we have
\begin{equation*}
\wfomc(\Gamma, n, \weights) = (3 \cdot 2^n + 3^n)^n.
\end{equation*}

In fact, for each domain element $i \in [n]$, either $S(i)$ is true and $R(i,j)$ is true for all $j \in [n]$ which contributes the weight $3 \cdot 2^n$, or $S(i)$ is false and $R(i,j)$ is not limited which contributes the weight $(1+2)^n$. Multiplying by the contributed weight for every $i$ yields the value above.
\end{example}

As the weighting functions $(\weights)$ are defined in terms of relations, all positive ground literals of the same relation are assigned the same weight, and so are all the negative ground literals of the same relation.
Therefore, the \wfomc we consider is also referred to as \emph{symmetric} in other literature~\citep{beame15:wfomc-fo3}.
Furthermore, if the sentence $\Psi$ is ground, \wfomc becomes the well-known \emph{weighted model counting} problem~\citep{chavira08:wmc}, specifically,
$$\symbwmc(\Psi, \weights) = \sum_{\mu\models\Psi}W(\mu, \weights).$$
Finally, it is worth noting that \fomc can be viewed as a special case of \wfomc where the weighting functions map all predicates to $1$.

\subsection{Complexity of WFOMC}
In this work, we consider the complexity of \wfomc with respect to the domain size.
That is, when measuring the complexity of \wfomc, the problem can be regarded as fixing the sentence and the weighting functions, with the domain size $n$ as the only (unary) input.
A sentence, or a class of sentences, is said to be \emph{domain-liftable}~\citep{broeck11:wfomc-ufo2} if for any pair of weighting functions, \wfomc can be computed in time polynomial in the domain size $n$.%
\footnote{Such complexity measure is also referred to as \emph{data complexity} of \wfomc in other literature, e.g., in \citet{beame15:wfomc-fo3}.}
Such a notion of \emph{tractability} makes sense from a practical standpoint, since in any real-world \wfomc applications, the sentence and the weighting functions are usually fixed, and the domain size is the only varying input \citep{broecketal21:lifting-intro}.

WFOMC is a counting problem where we can specify the input in unary.
A complexity class relevant for such tasks is the class \class{\#P_1} defined in \citet{valiant79:counting-complexity}.
A function $f: \{1\}^* \to \nat$ is in \class{\#P_1} if and only if there is a polynomial-time nondeterministic Turing machine $M$ with unary input alphabet such that $f(x)$ equals the number of accepting paths of $M$ on the input $x$.
A problem is \class{\#P_1}-hard if all \class{\#P_1} functions can be computed in polynomial time with an oracle for this problem.
Specifically, \wfomc of a sentence or a class of sentences is \class{\#P_1}-hard if all \class{\#P_1} functions can be computed in polynomial time with an oracle for \wfomc of the sentence(s) with any pair of weighting functions.
The following corollary provides evidence that \class{\#P_1} contains functions that are hard to compute.

\begin{corollary}{(Corollary from \citet[Theorem 1]{book74:languages-and-complexities})}
If \class{\#P_1 \subseteq FP}, then \class{E=NE}.
\end{corollary}

Previous work \citep{beame15:wfomc-fo3} showed that there is a universal Turing machine that can simulate all Turing machines representing \class{\#P_1} problems.

\begin{lemma}{(Universal Turing Machine for \class{\#P_1} \citep{beame15:wfomc-fo3})}
\label{lemma:utm}
There is a multi-tape linear-time nondeterministic Turing machine $M$ such that the input alphabet of $M$ is unary and computing its number of accepting paths on a given input is \class{\#P_1}-hard.
\end{lemma}

\begin{remark}
\label{remark:wfomc-membership}
\fomc of a fixed first-order logic sentence is a \class{\#P_1} problem.
In fact, a nondeterministic Turing machine can guess a model of a sentence and then verify it in time polynomial in the domain size.
Therefore, the number of accepting paths of this Turing machine equals the number of models of the sentence.
However, \wfomc of a fixed first-order logic sentence and fixed weighting functions is not necessarily in \class{\#P_1} since the number of accepting paths of any Turing machine must be a natural number. However, the weights in \wfomc can be negative or even non-integers.
\end{remark}

Another previous work \citep{kuzelka21:wfomc-c2} has shown that WFOMC of \ctwo{} with cardinality constraints can be computed by WFOMC of \fotwo{}, regardless of the axioms involved in the sentence.

\begin{lemma}{(Eliminating Counting Quantifiers and Cardinality Constraints \citep[Proposition 5 and Theorem 4]{kuzelka21:wfomc-c2})}
\label{lemma:c2+cc}
For any \ctwo{} sentence $\Psi$ (possibly with cardinality constraints) and the conjunction of axioms $A$, there is an \fotwo{} sentence $\Psi'$ without counting quantifiers or cardinality constraints such that $\symbwfomc(\Psi \land A, n, \weights)$ can be reduced to $\symbwfomc(\Psi' \land A, n, \weights)$ for any domain size $n$ and any pair of weighting functions $(\weights)$.
\end{lemma}

The proof of the lemma ultimately relies on an oracle for computing the \wfomc of an \fotwo sentence, hence, in the remainder of the text, we primarily focus on the \fotwo fragment. At the same time, all the results are also applicable to the \ctwo{} fragment with cardinality constraints.

\subsection{Domain-Liftability of the Two-Variable Fragment}
The fragment of \fotwo is known to be domain-liftable.
There are two proofs of that fact available in the literature, one making use of first-order knowledge compilation \citep{broeck11:wfomc-ufo2, broecketal14:wfomc-skolem} and one making use of the concept of \emph{cells} \citep{beame15:wfomc-fo3}, which are also referred to as \emph{1-types} in the literature on finite logics \citep{graedel97:fo2-sat}.
Let us now review the proof by \citet{beame15:wfomc-fo3}, as several of the concepts therein will be helpful to us later.

By the normalization in~\citet{graedel97:fo2-sat} and the technique of eliminating existential quantifiers in~\citet{broecketal14:wfomc-skolem}, computing WFOMC for any sentence in \fotwo can be reduced to WFOMC of a sentence having the form $\;\forall x\forall y: \psi(x,y)$, where $\psi(x, y)$ is a quantifier-free formula and it only contains relations of arity at most 2.%
\footnote{In this work, we focus on sentences containing relations of arity $1$ or $2$. All of our results are also applicable to cases with relations of arity $0$ using the \emph{Shannon expansion} as explained in \citet{beame15:wfomc-fo3}.}
The reduction respects the \wfomc value.

\begin{definition}[Cell a.k.a. 1-type]
  \label{def:cell}
  A \emph{cell} of a first-order sentence $\Psi$ is a maximal consistent set of literals formed from relations in $\preds{\Psi}$ using only a single variable $x$.
\end{definition}

\begin{definition}[2-table]
A \emph{2-table} of a first-order sentence $\Psi$ is a maximal consistent set of literals formed from relations in $\preds{\Psi}$ where each literal uses exactly two variables $x, y$.
\end{definition}

Intuitively, a cell (a 1-type) interprets unary and reflexive binary relations for a single domain element, and a~2-table interprets binary relations for a pair of distinct domain elements.
Alternatively, a domain element realizes a cell, or a pair of elements realizes a 2-table.
Note that, depending on the context, we shall treat cells and 2-tables as either a set of literals per the definitions above, or as a conjunction over the literals in the set.

\begin{example}
\label{ex:cells}
Consider the sentence $\Gamma = \forall x \forall y: S(x) \to R(x, y)$, which has four cells and four 2-tables, specifically
\begin{equation*}
\begin{aligned}[c]
C_1(x) &= \neg S(x) \wedge \neg R(x,x), \\
C_2(x) &= \neg S(x) \wedge R(x,x),\\
C_3(x) &= S(x) \wedge \neg R(x,x),\\
C_4(x) &= S(x) \wedge R(x,x),
\end{aligned}
\qquad
\begin{aligned}[c]
\pi_1(x, y) &= \lnot R(x,y) \land \lnot R(y,x)\\
\pi_2(x, y) &= \lnot R(x,y) \land R(y,x)\\
\pi_3(x, y) &= R(x,y) \land \lnot R(y,x)\\
\pi_4(x, y) &= R(x,y) \land R(y,x).
\end{aligned}
\end{equation*}
\end{example}

Given a WFOMC instance $\symbwfomc(\Psi, n, \weights)$ for an \fotwo sentence $\Psi$, the Skolemization trick in \citet{broecketal14:wfomc-skolem} produces a universally quantified \fotwo sentence $\Psi'$ of size $O(|\Psi|)$ and a pair of weighting functions $(w', \negw')$ such that $\symbwfomc(\Psi, n, \weights) = \symbwfomc(\Psi', n, w', \negw')$.
Specifically, $\Psi'$ is in the form $\forall x \forall y: \psi(x,y)$ where $\psi(x,y)$ is a quantifier-free \fotwo sentence.
Therefore, $\Psi'$ can be expanded as the conjunction of ground formulas over the domain $[n]$:
\begin{equation*}
  \begin{aligned}
    \Psi' = \left( \bigwedge_{1 \le a < b \le n} \psi(a,b) \land \psi(b,a) \right) \land \left( \bigwedge_{c\in[n]} \psi(c,c) \right).
  \end{aligned}
\end{equation*}

Let $C = \{C_1, C_2, \cdots, C_p\}$ be the set of all cells of $\Psi'$ and $D$ be the set of all its 2-tables.
Note that $p$ solely depends on the number of distinct predicates in $\preds{\Psi'}$ and is thus constant with respect to $n$.
Let $C_{\tau_a}$ denote the cell realized by the element $a$ ($\tau_a \in [p]$).
Note that if we know the realized cell $C_{\tau_a}$ for a particular domain element $a$, only possible worlds containing positive literals from the formula $C_{\tau_a}(a)$ can be models.
In other words, $C_{\tau_a}(a)$ becomes \emph{evidence} that we condition on.
Hence, we may substitute unary and reflexive binary literals in $\psi(a,b) \land \psi(b, a)$ with true or false according to $C_{\tau_a}$ and $C_{\tau_b}$.
For a pair of elements $(a,b)$, the resulting formula will not share any ground literals with any other pair, hence the 2-tables between each pair of elements can be selected independently.
Moreover, the formula only depends on $C_{\tau_a}$ and $C_{\tau_b}$.
Therefore, WFOMC can be computed as follows:
\begin{equation}
\label{eq:basic-wfomc}
  \begin{aligned}
    & \symbwfomc(\Psi', n, w', \negw')
    = \sum_{\tau_1, \cdots, \tau_n \in [p]} \ \prod_{i=1}^n W(C_{\tau_i}) \prod_{1 \le i < j \le n} r_{\tau_i,\tau_j},
  \end{aligned}
\end{equation}
where
\begin{equation*}
  r_{s,t} = \sum_{\substack{\pi \in D, \\ C_{s}(a) \land C_{t}(b) \land \pi(a,b) \models \psi(a,b) \land \psi(b,a)}} W(\pi),
\end{equation*}
and $a,b$ are arbitrary domain elements.
Informally, since the simplification of $\Psi'$ depends only on $C_{\tau_a}$ and $C_{\tau_b}$, we sum over all possible cell realizations and construct models for each such realization independently.
For a fixed $(\tau_1, \cdots, \tau_n) \in [p]^n$, the model will naturally contain the unary and binary reflexive literals from $C_{\tau_1}, \cdots, C_{\tau_n}$, hence the factors $W(C_{\tau_i})$.
Moreover, it will contain 2-tables such that, together with the fixed cell realization, they satisfy the formula $\forall x\forall y: \psi(x,y)$.
There may be more 2-tables to choose from to construct a single model, so we sum over all of those to obtain the factors $r_{\tau_i,\tau_j}$.

Let us now rewrite \Cref{eq:basic-wfomc} so that we may clearly evaluate it in polynomial time, i.e., prove that \fotwo is, indeed, domain-liftable.
Consider possible partitions of $[n]$ into $p$ disjoint sets; each partition representing an assignment of a subset of $[n]$ to a particular cell.
A particular assignment of elements to the cell $C_l$ means that those elements realize the cell $C_l$.
Naturally, empty partitions are also allowed, meaning that some cells may not be realized by any domain elements.

For each model $\mu$ of $\Psi'$ over the domain $[n]$, there is a unique partition $\mathcal{C} = (C_1, C_2, \dots, C_p)$ consistent with $\mu$ on the unary and reflexive binary atoms.%
\footnote{We slightly abuse notation here. $C_l$ may now refer to a particular cell or a set of domain elements assigned to that cell.}
We call $\mathcal{C}$ the \emph{partition} of $\mu$, and $(|C_1|, \dots, |C_p|)$ the \emph{cell configuration} of $\mu$.
Observe that the model $\mu$ is also a model of the ground formula 
$$\Phi_{\mathcal{C}} = \bigwedge_{i,j\in[p]}\bigwedge_{a\in C_i, b\in C_j} \psi(a,b),$$
and thus we can write 
$$\symbwfomc(\Psi, n, \weights) = \sum_{\mathcal{C}} \symbwmc(\Phi_{\mathcal{C}}, \weights).$$

We can rewrite $\Phi_\mathcal{C}$ to separate the cases when $C_i \neq C_j$ and also when $C_i = C_j$ but $a \neq b$:
\begin{equation}
  \label{eq:wfomc_lineage}
  \begin{aligned}
    \Phi_{\mathcal{C}} = \bigwedge_{i,j\in[p]: i<j}\left(\bigwedge_{a\in C_i, b\in C_j} \psi(a,b)\land\psi(b,a)\right)\land \bigwedge_{i\in[p]} \left(\bigwedge_{a,b\in C_i: a < b} \left(\psi(a,b)\land\psi(b,a)\right)\land \bigwedge_{c\in C_i} \psi(c,c)\right).
  \end{aligned}
\end{equation}

Since the partition $\mathcal{C}$ can be treated as evidence we conditioned on, we know the truth values of the unary and reflexive binary atoms in $\Phi_{\mathcal{C}}$.
As we have already mentioned above, we can simplify the formula $\psi(x,y)$ by replacing every occurrence of unary and reflexive binary atoms with true or false as appropriate.
Write $\psi_{ij}(x,y)$ for the simplified version of $\psi(x,y) \land \psi(y,x)$ when $x$ and $y$ belong to $C_{i}$ and $C_{j}$ respectively, which leads to the following reformulation of \Cref{eq:wfomc_lineage}:
\begin{equation}
  \label{eq:wfomc_lineage_simplified}
  \begin{aligned}
    \Phi_{\mathcal{C}} = \bigwedge_{i,j\in[p]: i<j}\left(\bigwedge_{a\in C_i, b\in C_j} \psi_{ij}(a,b)\right)\land \bigwedge_{i\in[p]} \left(\bigwedge_{a,b\in C_i: a < b} \psi_{ii}(a,b)\land \bigwedge_{c\in C_i} \psi_{ii}(c,c)\right).
  \end{aligned}
\end{equation}

Note that the subformulas $\psi_{ij}$ do not share any atoms and each of them is fully determined by a single 2-table.
They each can be treated as an independent subproblem when searching for models of $\Phi_{\mathcal{C}}$.
Also note that the subformula $\bigwedge_{c\in C_i} \psi_{ii}(c,c)$ is simply a trivial tautology or contradiction, since $\psi(c,c)$ contains only unary and reflexive binary literals.
Therefore, to allow the existence of a model, we must first ensure that $\bigwedge_{c\in C_i} \psi_{ii}(c,c)$ becomes a tautology. That can be achieved by assigning domain elements only to so-called \emph{valid cells} first introduced by \citet{vanbremen21:fast-wfomc}.

\begin{definition}[Valid Cell]
  \label{def:valid-cell}
  A \emph{valid cell} of a first-order sentence $\Psi$ is a cell $C(x)$ of $\Psi$ such that $\Psi \land C(t)$ is satisfiable for any domain element $t$.
\end{definition}

Should we work with cells that are not valid, we would end up solving subproblems whose \wfomc is necessarily zero.
Therefore, we usually implicitly work only with valid cells and use $p$ to denote the number of valid cells rather than all cells.

\begin{example}
    Consider $\Gamma = \forall x \forall y: S(x) \to R(x, y)$ from \cref{ex:cells} again, whose four cells we have already identified.
    Clearly, $C_3(x)$ is not a valid cell, since for a domain element $i$, we would require $S(i)$ to be true, yet $R(i, i)$ to be false.

    Now, consider $n=4$ and the partition $\mathcal{C} = (\{1, 2\}, \{3\}, \{\}, \{4\})$.
    Let us further use colored notation for increased readability:
    \begin{align*}
      \Phi_\mathcal{C} &= \alpha \land \beta\\
      \alpha &= \bigwedge_{i,j\in[p]: i<j}\left(\bigwedge_{a\in C_i, b\in C_j} {\color{red}\psi(a,b)}\land{\color{Green}\psi(b,a)}\right)\\
      \beta &= \bigwedge_{i\in[p]} \left(\bigwedge_{a,b\in C_i: a < b} \left({\color{red}\psi(a,b)}\land{\color{Green}\psi(b,a)}\right)\land \bigwedge_{c\in C_i} {\color{blue}\psi(c,c)}\right)
    \end{align*}
    Following the substitution procedure of \Cref{eq:wfomc_lineage} we obtain
    \begin{align*}
    \alpha &= ({\color{red} S(1) \to R(1,3)}) \land ({\color{Green} S(3) \to R(3,1)}) \land ({\color{red} S(2) \to R(2,3)}) \land ({\color{Green} S(3) \to R(3,2)}) \land ({\color{red} S(1) \to R(1,4)})\\
    &\land ({\color{Green} S(4) \to R(4,1)}) \land ({\color{red} S(2) \to R(2,4)}) \land ({\color{Green} S(4) \to R(4,2)}) \land ({\color{red} S(3) \to R(3,4)}) \land ({\color{Green} S(4) \to R(4,3)})\\
    \\
    \beta &= ({\color{red} S(1) \to R(1,2)}) \land ({\color{Green} S(2) \to R(2,1)})\\
    &\land ({\color{blue} S(1) \to R(1,1)}) \land ({\color{blue} S(2) \to R(2,2)}) \land ({\color{blue} S(3) \to R(3,3)}) \land ({\color{blue} S(4) \to R(4,4)})
    \end{align*}

    While $\alpha \land \beta$ is ground, individual subformulas $\psi(x, y) \land \psi(y, x)$ share atoms, such as the atom $S(1)$ being shared across $\psi(1,2)\land\psi(2,1)$, $\psi(1,3)\land\psi(3,1)$, $\psi(1,4)\land\psi(4,1)$ and also $\psi(1,1)$.

    Let us now simplify $\alpha$ and $\beta$ according to \Cref{eq:wfomc_lineage_simplified}:
    \begin{align*}
    \alpha' &= ({\color{red} \bot \to R(1,3)}) \land ({\color{Green} \bot \to R(3,1)}) \land ({\color{red} \bot \to R(2,3)}) \land ({\color{Green} \bot \to R(3,2)}) \land ({\color{red} \bot \to R(1,4)})\\
    &\land ({\color{Green} \top \to R(4,1)}) \land ({\color{red} \bot \to R(2,4)}) \land ({\color{Green} \top \to R(4,2)}) \land ({\color{red} \bot \to R(3,4)}) \land ({\color{Green} \top \to R(4,3)})\\
    \\
    \beta' &= ({\color{red} \bot \to R(1,2)}) \land ({\color{Green} \bot \to R(2,1)})\\
    &\land ({\color{blue} \bot \to \bot}) \land ({\color{blue} \bot \to \bot}) \land ({\color{blue} \bot \to \top}) \land ({\color{blue} \top \to \top})
    \end{align*}

    There are no more shared atoms in $\alpha' \land \beta'$, and the search for a (weighted) model of $\Phi_\mathcal{C}$ can be decomposed into smaller searches for models of the subformulas $\psi_{ij}$.
    Naturally, a similar approach may be applied when only computing the total number of (weighted) models, rather than enumerating them all.
    We solve $\symbwmc$ for each of the mutually independent subproblems and then multiply the values together to obtain the final weighted model count.
\end{example}

Notice that when we are only interested in \wmc of $\Phi_\mathcal{C}$, the particular cell partition $\mathcal{C}$ does not matter.
Since the weights are defined on predicates, the cell configuration $(|C_1|, \dots, |C_p|)$ carries all the necessary information.
Thus, let us define partial model counts as
\begin{equation}
    \begin{aligned}
        r_{ij} := \symbwmc(\psi_{ij}(a, b), \weights), \qquad w_{k} := \symbwmc(\psi(c,c) \land C_k(c), \weights) = W(C_k, \weights),
    \end{aligned}
    \label{eq:r-w}
\end{equation}
where $a,b,c\in\dom$.
Observe that $r_{ij} = r_{ji}$ and the values are defined with respect to the formulas $\psi_{ij}$, whereas $w_k$ is defined using the non-simplified $\psi(c,c)$ so that we do not remove the unary and binary reflexive literals from the weighted model count completely.

Now, we can finally write
\begin{align}
\label{eq:2wfomc}
\symbwfomc(\Psi, n, \weights) =
\sum_{\bm{k}\in\nat^p:|\bm{k}|=n} \binom{n}{\bm{k}}&\prod_{i,j\in[p]:i<j}r_{ij}^{k_ik_j}\prod_{i\in[p]}r_{ii}^{\binom{k_i}{2}}w_i^{k_i},
\end{align}
which implies that universally quantified \fotwo sentence is domain-liftable since \Cref{eq:2wfomc} may be evaluated in time polynomial in $n$.
Since we already transformed an arbitrary \fotwo sentence into a universally quantified form, we can conclude that the entire fragment of \fotwo is domain-liftable.
Using \cref{lemma:c2+cc}, the domain-liftability result can be further extended to the fragment of \ctwo, possibly with cardinality constraints.

Note that the values from \Cref{eq:r-w} correspond to $r_{i,j}$ and $W(C_k)$ in \Cref{eq:basic-wfomc}.
Now, however, we have defined them in terms of the propositional Weighted Model Counting problem.
Intuitively, we solve small propositional subproblems, which we then cleverly combine in \Cref{eq:2wfomc} to solve the original WFOMC problem.
In terms of model enumeration (not model counting), we construct small partial models that are mutually independent (they do not share any atoms), which we then use to construct models of the input sentence by a clever application of the Cartesian product and union operations.

\subsection{First-Order Conditioning}
So far, the claim of domain-liftability of \fotwo has only focused on the case when the sentences are \emph{constant-free}.
When performing inference in SRL, we usually aim also to support so-called \emph{evidence}, which often takes the form of a conjunction of ground atoms, i.e., certain facts observed to be true or false.
We also say that we are \emph{conditioning} on some evidence, meaning we are restricting our models to satisfy additional conditions.

\citet{broeckdavis12:wfomc-evidence} have shown that the tractability result for \fotwo can be extended to conditioning on \emph{unary evidence}.
Specifically, \wfomc of an \fotwo sentence with unary evidence can be computed in time polynomial in the size of the evidence.
Unfortunately, the same does not hold for conditioning on binary relations, which turns out to be \class{\#P}-hard in general \citep{broeckdavis12:wfomc-evidence}.
Intuitively, the task is more challenging because it potentially breaks many symmetries, which are exploited to perform the task faster. 
\rev{However, even binary evidence can be processed in a domain-lifted manner as long as its Gaifman graph is of bounded (constant) treewidth \citep{kula26:evidence-gaifman-tw}, although a dedicated algorithm is required for such handling.}

Throughout this text, we usually implicitly work with sentences without evidence.
However, our algorithms support working with unary evidence in a domain-lifted manner, and we explicitly note when making adjustments to accommodate such conditioning.
We denote $\evidence{i}$ the conjunction of all evidence on the constant $i$, hence, we may always rewrite our sentences as
\begin{equation*}
    \Psi = \psi(x, y) \land \bigwedge_{i\in[n]} \evidence{i},
\end{equation*}
where $\psi(x,y)$ is a constant-free formula and $\evidence{i}$ could be an empty conjunction (i.e., a tautology) for any $i\in[n]$.

\begin{example}
Consider the domain $[n]$ such that $n\ge 3$ and the sentence
$$\psi = (\forall x \exists y: A(x) \land B(x) \to C(x, y)) \land A(1) \land (A(2) \lor C(2, 2)).$$
We can rewrite $\psi$ as
$$\psi' = (\forall x \exists y: A(x) \land B(x) \to C(x, y)) \land (\forall x: U(x) \leftrightarrow A(x) \lor C(x, x)) \land A(1) \land U(2),$$
where $U/1$ is a fresh predicate symbol.
Then, following our notation, we have
\begin{align*}
    \evidence{1} &= A(1),\\
    \evidence{2} &= U(2), \\
    \evidence{3} &= \top.
\end{align*}
\end{example}

\subsection{Linear Order Axiom}
As we have already mentioned above, many applications of lifted inference are over domains with inherent ordering on the objects involved~\citep{le20:probabilistic-time-series,vlasselaer14:dynamic-relational-models}.
Such ordering may be introduced by requiring a distinguished binary relation from our vocabulary to introduce a linear (total) ordering on the domain elements~\citep{libkin04:lo}.
Suppose a binary relation $R$ is a linear order of elements, i.e., for each model, it holds that
\begin{itemize}
    \item $R$ is reflexive, i.e., $\forall x: R(x,x)$,
    \item $R$ is anti-symmetric, i.e., $\forall x\forall y: (R(x,y)\land R(y,x)) \to (x = y)$,
    \item $R$ is transitive, i.e., $\forall x\forall y\forall z: (R(x,y) \land R(y,z))\to R(x,z)$,
    \item $R$ is a total relation, i.e., $\forall x\forall y: R(x,y)\lor R(y,x)$.
\end{itemize}

When using logics with three or more variables, it is feasible to append the sentences above to the formula one is working with.
However, in the context of \wfomc, that approach would not allow us to perform domain-lifted inference.
Since we are limited to at most two variables in our sentences, we deal with ordering elements differently.
We introduce a new syntactic construct which asserts that a specified relation $R$ must satisfy the conditions above.

\begin{definition}[Linear Order Axiom]
\label{def:loaxiom}
The \emph{linear order axiom} on the relation $R$, denoted as $\loaxiom(R)$, requires that $R$ should be a linear order of elements, i.e., $R$ is a reflexive, anti-symmetric, transitive, and total relation.
In other words, the graph $G(R)$ is an acyclic tournament with a loop at each vertex.
\end{definition}

Intuitively speaking about \cref{def:loaxiom}, if we have a logical sentence $\Psi$ possibly containing a binary predicate $R$, then a possible world $\omega$ is a model of $\Psi \land \loaxiom(R)$ if and only if $\omega$ is a model of $\Psi$ and the atoms on the predicate $R$ present in $\omega$ satisfy the linear order properties.

Often, for both clarity and brevity, we denote the binary relation constrained by the linear order axiom (relation $R$ in \cref{def:loaxiom}) by the special symbol $\leq$.
We then employ the standard infix notation $x \leq y$ to denote that $x$ is less than or equal to $y$.
We also make use of a shorthand $x<y$ meaning that $x$ is strictly less than $y$, i.e., it holds that $(x\le y) \land \neg(y\le x)$.
In \cref{sec:5-two-los}, where we analyze sentences with two linear orders, we differentiate the two distinguished relations using subscripts, i.e., as $\leq_1$ and $\leq_2$.

\begin{remark}
In the presence of the linear order axiom $\loaxiom(\lopred)$, valid cells are only those that contain the positive literal $(x \lopred x)$.
\end{remark}

\rev{
Semantically speaking, $\loaxiom(\le)$ postulates that there is some linear ordering defined by the predicate $\le$, but it does not restrict it any further.
Hence, model counting will consider all of the $n!$ possible orderings.
In some use cases, one may wish to only work with one specific ordering, which can always be expressed, without loss of generality, as the natural order $1\le2\le3\le\ldots\le n$.
For such purposes, let us define the \emph{natural linear order axiom} as well.
}

\begin{definition}{Natural Linear Order Axiom}
\label{def:natural-loaxiom}
    \rev{
    The \emph{natural linear order axiom} on the relation $\le$, denoted as $\lonat(\le)$, interpreted over the domain $[n]$, requires that $1\le2\le3\le\ldots\le n$.
    }
\end{definition}

\rev{
Note that while we could express $\lonat(\le)$ using evidence, i.e., explicitly include
$$\epsilon = (1\le1) \land (1\le2) \land (1\le3) \land \ldots \land (1\le n) \land (2\le2)\land(2\le3)\land\ldots(2\le n)\land\ldots\land(n\le n)$$ 
in our sentence, the evidence would be binary and lead to the aforementioned \class{\#P}-hardness (the Gaifman graph of $\epsilon$ is a complete graph)
Therefore, we again use an axiom to describe the situation instead.

The majority of our analysis and techniques will apply to both of the ordering contexts.
Moreover, we will switch between them as appropriate since, when the sentence is constant-free, searching for or counting models of $\Gamma \land \lonat(\lopred)$ turns out to be a subproblem of searching for or counting models of $\Gamma \land \loaxiom(\lopred)$.
As we will see, the situation becomes more complicated once we try to support evidence.
In those cases, we will carefully specify which ordering definition we are using.
}

\begin{example}
\label{ex:sequence-split}
\rev{
As a simple example of what the linear order axiom allows us to express, consider the sentence $\phi = \forall x \forall y: \psi(x, y) \wedge \lonat(\lopred)$, where
$$\psi(x, y) = (T(x) \wedge (x \lopred y)) \to T(y).$$

How can we interpret models of $\phi$?
Due to $\lonat(\lopred)$, the $\lopred$ predicate will order the domain such that $1 \leq 2 \leq 3$ for $n=3$.
Therefore, the domain will become the sequence $(1,2,3)$.

The formula $\psi(x, y)$ then seeks to split that sequence into its beginning (\textit{head}) and its end (\textit{tail}).
The predicate $T/1$ denotes the tail of the sequence.
Whenever there is a constant $c\in\dom$ for which $T(c)$ is in the model (i.e., $c$ is part of the tail), then for all constants $c'\in\dom$ that are \textit{greater than $c$}, $T(c')$ is also in the model.
All the other constants then belong to the sequence head.

If we had  $\phi = \forall x \forall y: \psi(x, y) \wedge \loaxiom(\lopred)$ instead, we would consider all other possible orderings as well, i.e., for $n=3$, we would also consider the sequences $(1, 3, 2)$, $(2,1,3)$, $(2,3,1)$, $(3,1,2)$ and $(3,2,1)$.
Note, however, that in the absence of conditioning, the splitting results would be isomorphic across different orderings.
Further restricting that, e.g., $\neg T(2)$ holds, would break that symmetry.
}
\end{example}

\section{Domain-Lifted Algorithm for WFOMC with the Linear Order}
\label{sec:3-single-lo}
In this section, we prove that the logical fragment of \ctwo extended by the linear order axiom is domain-liftable.
We do so by developing an algorithm that can compute \wfomc of a \ctwo sentence with the linear order axiom in time polynomial in the domain size.

We first present a new algorithm based on dynamic programming that computes \wfomc{} of a universally quantified \fotwo{} sentence incrementally, and does so in time polynomial in the domain size.
Note again that the assumption of universal quantification is not a limiting one, since, due to \citet{broecketal14:wfomc-skolem}, we can \emph{skolemize} in the context of \wfomc.
\rev{
Second, we show a minor adaptation of the new algorithm preserving the polynomial runtime, which allows us to compute \wfomc of both $\Psi = \psi \land \lonat(\lopred)$ and $\Psi = \psi \land \loaxiom(\lopred)$, where $\psi$ is a universally quantified \fotwo{} sentence with $\lopred\ \in\preds{\Psi}$.
}
Third, we observe that the algorithm can also accommodate unary evidence.
Finally, due to the reductions from \citet{kuzelka21:wfomc-c2}, which only require access to an oracle for WFOMC over \fotwo, we use the algorithm to compute WFOMC for any \ctwo sentence possibly containing the linear order axiom.

\subsection{A New Domain-Lifted Algorithm}
The new algorithm will operate incrementally.
The domain size will be inductively enlarged in a similar way as in the \textit{domain recursion rule} \citep{broeck11:wfomc-ufo2, kazemietal16:new-liftable-classes}.

Suppose we have partitioned the domain $[m]$ for any $m \leq n$ as $\mathcal{C}_1 = (C_1, C_2, \dots, C_p)$ and we would like to enlarge the domain with one more element, i.e., we seek $\mathcal{C}_2 = (C_1, \dots, C_l \cup \{m+1\}, \dots, C_p)$, where $C_l$ is an arbitrary cell.
Using \Cref{eq:wfomc_lineage_simplified}, we have the following relationship between the formulas $\Phi_{\mathcal{C}_1}$ and $\Phi_{\mathcal{C}_2}$:
\begin{equation}
  \begin{aligned}
  \Phi_{\mathcal{C}_2} = \Phi_{\mathcal{C}_1} \land \psi_{ll}(m+1,m+1) \land \bigwedge_{i\in[p]} \bigwedge_{a\in C_i} \psi_{il}(a,m+1).
  \end{aligned}
  \label{eq:ground_induction}
\end{equation}

Note that the conjuncts in \Cref{eq:ground_induction} still do not share any propositional variables (atoms), and thus they are all independent.
Therefore, we can compute $\symbwmc(\Phi_{\mathcal{C}_2}, \weights)$ as 
\begin{equation}
  \begin{aligned}
  \symbwmc(\Phi_{\mathcal{C}_2}, \weights) = \symbwmc(\Phi_{\mathcal{C}_1}, \weights)\cdot w_l \cdot \prod_{i\in[p]} r_{il}^{|C_i|}.
  \end{aligned}
  \label{eq:induction_wmc}
\end{equation}

Next, let $T_m(\bm{k})$ be the sum of weighted models of a sentence $\Psi$ over the domain $[m]$ whose cell configuration is $\bm{k}$.
Recall that the cell configuration $\bm{k}$ holds all the necessary information; we do not require knowing a particular partition corresponding to $\bm{k}$. 
Therefore, we have that 
\begin{align}
\label{eq:wmc-multinomial}
T_m(\bm{k}) = \binom{m}{\bm{k}}\cdot \symbwmc(\Phi_{\mathcal{C}}, \weights),    
\end{align}
where $\mathcal{C}$ is any partition whose cell configuration is $\bm{k}$ and the multinomial coefficient $\binom{m}{\bm{k}}$ counts the number of such partitions $\mathcal{C}$.

Now, we can describe the entire procedure.
First, compute \wmc for a single-element domain, which is equal to $w_l$ when assigning the given domain element to the cell $C_l$.
Next, let the function $T_i$ map each possible cell configuration on a domain of size $i$ to its corresponding weighted model count.
To compute $T_{i+1}(\bm{k})$, we must find all entries $T_i(\bm{k} - \vecdelta_j)$ such that $C_j$ is one of the cells.
Then we assign the new domain element $(i+1)$ to $C_j$ and use \Cref{eq:induction_wmc} to update the weighted model count.
Summing the results across all possible cells $C_j$ gives us $T_{i+1}(\bm{k})$.
Last, summing all entries for $T_n$ produces the final weighted model count.

Refer to \cref{alg:iwfomc} for a detailed pseudocode.
Note, however, that so far, we have only worked with \fotwo sentences. At the same time, the pseudocode incorporates changes from \cref{ssec:handle-linear-order} which enable the algorithm to compute \wfomc over \fotwo extended by the linear order axiom.
To limit ourselves to \fotwo, we would use values $w_k$ and $r_{ij}$ from \Cref{eq:r-w} instead of $\wlo_{k}$ and $\rlo_{ij}$ (\cref{iwfomc:line:wmc}), and we would return the value $\gamma$ (\cref{iwfomc:line:gamma}) rather than its multiple.

\begin{lemma}
\label{lemma:iwfomc-fo2}
Using the values $w_k$ and $r_{ij}$ from \Cref{eq:r-w}, the value $\gamma$ produced by \cref{alg:iwfomc} is equal to $\symbwfomc(\fotwoformula, n, \weights)$ of a universally quantified \fotwo{} sentence $\fotwoformula$.
Moreover, the value is computed in time polynomial in the domain size $n$.
\end{lemma}

\begin{proof}
We prove the correctness by induction on the domain size.
\begin{enumerate}
    \item First, assume $n = 1$.
    Suppose we have a domain partition $\mathcal{C} = (C_1, C_2, \ldots, C_p)$ with the cell configuration $\vecdelta_l$ for some $1\le l \le p$, i.e., we assign the single domain element to the cell $C_l$.
    Following \Cref{eq:wfomc_lineage_simplified,eq:r-w}, we have
    \begin{equation*}
        \begin{aligned}
        \Phi_\mathcal{C} = \psi_{ll}(1,1) \quad \text{and} \quad \symbwmc(\Phi_\mathcal{C}, \weights) = w_l. 
        \end{aligned}
    \end{equation*}
    Therefore, $\gamma = T_1(\vecdelta_l) = w_l$ which is consistent with \cref{iwfomc:line:init}.

    \item Next, assume that $T_i(\bm{k})$ holds \wfomc of $\Psi$ over the domain $[i]$ for a particular cell configuration $\bm{k}$.
    Take an arbitrary partition $\mathcal{C}_1 = (C_1, C_2, \ldots, C_p)$ such that $(|C_1|, |C_2|, \ldots, |C_p|) = \bm{k}$ and consider adding a new domain element $(i+1)$, i.e., consider $\mathcal{C}_2 = (C_1, \ldots, C_l \cup \{i+1\}, \ldots, C_p)$ with the cell configuration $\bm{k}_{new} = \bm{k}+\vecdelta_l$.
    \Cref{eq:ground_induction,eq:induction_wmc} give us
    \begin{equation*}
        \symbwmc(\Phi_{\mathcal{C}_2}, \weights) = \symbwmc(\Phi_{\mathcal{C}_1}, \weights)\cdot w_l \cdot \prod_{j\in[p]} r_{jl}^{|C_j|}.
    \end{equation*}
    Combining that with the induction hypothesis, we have
    \begin{equation*}
        T_{i+1}(\bm{k})\! =\! \sum_{l\in[p]} T_i(\bm{k} - \vecdelta_l)\! \cdot\! w_l\! \cdot\! \prod_{j\in[p]} r_{jl}^{(\bm{k} - \vecdelta_l)_j},
    \end{equation*}
    which is exactly the value stored in the table $T_{i+1}$ at the end of the $(i+1)$-st iteration of the for-loop on \cref{iwfomc:line:main_loop}.

    \item Finally, following \Cref{eq:wmc-multinomial}, we have
    \begin{equation*}
        T_n(\bm{k}) = \binom{n}{\bm{k}} \cdot \symbwmc(\Phi_\mathcal{C}, \weights)
    \end{equation*}
    for any $n\in\nat$ and any partition $\mathcal{C}$ with the cell configuration $\bm{k}$.
    Summing across all possible cell configurations (\cref{iwfomc:line:gamma}) produces $\symbwfomc(\Psi, n, \weights)$.
    Note that the final formula is also equivalent to \Cref{eq:2wfomc}.
\end{enumerate}

As for the time complexity, we may argue as follows:
The initialization, along with the base case handling, runs in time $O(1)$ with respect to $n$.
The main loop on \cref{iwfomc:line:main_loop} runs in $O(n)$.
The first nested loop (\cref{iwfomc:line:cell_loop}) is again independent of $n$, and the second (\cref{iwfomc:line:table_loop}) runs in $O(n^p)$.
The final sum on \cref{iwfomc:line:gamma} also runs in $O(n^p)$.
Overall, we have
\begin{align*}
    O(n) \cdot O(n^p) + O(n^p) \in O(n^{p+1}),
\end{align*}
which is polynomial in the domain size $n$, since $p$ is the number of (valid) cells which only depends on the fixed formula $\fotwoformula$.
\end{proof}

\subsection{Dealing with a Linear Order}
\label{ssec:handle-linear-order}

\begin{algorithm*}[tbp]
  \caption{WFOMC for \fotwo+$\loaxiom$}
  \label{alg:iwfomc}
  \KwIn{An \fotwo sentence $\Psi = \forall x\forall y: \psi(x,y)$ with $\{\lopred\}\subseteq \preds{\Psi}$, domain size $n$, weighting functions $(\weights)$}
  \KwOut{$\symbwfomc(\Psi \land \loaxiom(\lopred), n, \weights)$}
  Compute $\wlo_{i}, \rlo_{ij}$ for every $i,j\in[p]$ \tcp*{init} \label{iwfomc:line:wmc}
  $T_1(\bm{\vecdelta}_l) \leftarrow \wlo_l$ for each $l\in[p]$ \label{iwfomc:line:init} \tcp*{base case}
  \ForEach{$i=2$ \KwTo $n$}{ \label{iwfomc:line:main_loop}
    $T_i(\bm{k}) \leftarrow 0$ for every $\bm{k}\in\mathbb{N}^p$\\
    \ForEach{$l\in[p]$}{ \label{iwfomc:line:cell_loop}
      \ForEach{$(\bm{k}_{old}, W_{old}) \in T_{i-1}$}{ \label{iwfomc:line:table_loop}
        $W_{new} \leftarrow W_{old} \cdot \wlo_l \cdot \prod_{s\in[p]} \rlo_{sl}^{(\bm{k}_{old})_s}$\label{iwfomc:line:weight-update} \tcp*{assign the element $i$ to $C_l$}
        $\bm{k}_{new} \leftarrow \bm{k}_{old} + \vecdelta_l$\\
        $T_i(\bm{k}_{new}) \leftarrow T_i(\bm{k}_{new}) + W_{new}$\label{iwfomc:line:loop_sum}
      }
    }
  }
  $\gamma \gets \sum_{\bm{k}\in\mathbb{N}^p: |\bm{k}|=n} T_n(\bm{k})$\label{iwfomc:line:gamma}\\
  \Return $n! \cdot \gamma$\label{iwfomc:line:final-sum}
\end{algorithm*}

When adding the linear order axiom to the input sentence $\fotwoformula$, each model $\omega$ of $\psi$ must introduce a domain ordering.
Assume we find the set $\Omega$ of all models for one fixed ordering.
Having a domain permutation $\pi$, $$\Omega' = \bigcup_{\omega\in\Omega}\Set{\pi(\omega)}$$ will be the set of all models with respect to the new domain ordering defined by $\pi$.
Hence, the situation is symmetric for any particular ordering of the domain.

\begin{lemma}
\label{lemma:n-factorial}
Let $\Psi$ be a formula of the form $\Psi = \forall x \forall y: \psi(x,y) \wedge \loaxiom(\lopred)$, where $\lopred$ is a binary relation from $\mathcal{P}_{\psi}$.
Let $\Delta$ be a domain over which we want to compute \wfomc of $\Psi$.

If $\omega \models \Psi$ and $\pi$ is a permutation of $\Delta$, then $\pi(\omega) \models \Psi$, where application of $\pi$ to a possible world is defined by appropriate substitution of the domain elements in ground atoms. 
Moreover, if $\pi$ is not an identity, then $\omega \neq \pi(\omega)$.
\end{lemma}

\begin{proof}
If $\omega$ is a model of $\Psi$, we can partition $\omega$ into two disjoint sets:
One holding only atoms with the predicate $\lopred$ which we denote $\omega^\lopred$ and $\omega_\psi = \omega  \setminus \omega^\lopred$.
The set $\omega^\lopred$ defines an ordering of $\Delta$, and $\omega_\psi$ is then a model of $\forall x \forall y: \psi(x,y)$ respecting the ordering defined by $\omega^\lopred$.
Applying the permutation $\pi$ to $\omega^\lopred$ will define a different domain ordering.

Since there are no constants in $\psi$, $\pi(\omega_\psi)$ will still be a model of $\forall x \forall y: \psi(x,y)$ (we only apply a different substitution to the variables in $\psi$).
Moreover, since $\omega_\psi$ respected the ordering defined by $\omega^\lopred$,$\pi(\omega_\psi)$ will respect the new ordering defined by $\pi(\omega^\lopred)$.

Hence $\pi(\omega) = \pi(\omega^\lopred) \cup \pi(\omega_\psi)$ is another model of $\Psi$ and it must be different from $\omega$, because $\pi(\omega^\lopred)$ defines a different ordering than $\omega^\lopred$.
\end{proof}

\begin{corollary}
\label{cor:n-factorial}
\rev{
Given $\Psi = \forall x \forall y: \psi(x,y) \wedge \loaxiom(\lopred)$, we can compute \wfomc of $\Psi$ by solving the task for $ \Psi' = \forall x \forall y:\psi(x,y) \wedge \lonat(\lopred)$, and then multiplying the obtained value by the number of distinct domain orderings, i.e., 
\begin{align*}
    \symbwfomc(\Psi, n, \weights) = n! \cdot \sum_{\mu\in\fomodels{\Psi}{n}^\le} W(\mu, \weights),
\end{align*}
where $\fomodels{\Psi}{n}^\leq$ denotes the set of models of $\Psi$ over the ordered domain $\dom = \{1 \leq 2 \leq \dots \leq n\}$.
}
\end{corollary}

Using \cref{alg:iwfomc}, we can solve the subproblem as easily as computing \wfomc over \fotwo.
\cref{alg:iwfomc} incrementally enlarges the domain, one new element at a time.
That is beneficial to us, since we need to enforce that, when processing the $i$-th domain element, $i' < i$ for all previously processed domain elements $i'$.
Denote
\begin{equation}
\label{eq:psi<+psi<=}
    \begin{aligned}
        \psi^\le(x,y) &:= \psi(x,y)\land x\le y,\\
        \psi_{ij}^\le(x,y) &:= \psi_{ij}(x,y)\land x\le y,\\
    \end{aligned}
    \qquad
    \begin{aligned}
        \psi^<(x,y) &:= \psi(x,y) \land (x<y),\\
        \psi_{ij}^<(x,y) &:= \psi_{ij}(x,y) \land (x<y),
    \end{aligned}
\end{equation}
and define
\begin{equation}
  \begin{aligned}
    \rlo_{ij} := \symbwmc(\psi_{ij}^{<}(a,b), \weights), \qquad \wlo_k := \symbwmc(\psi^\le(c,c) \land C_k(c), \weights).
  \end{aligned}
  \label{eq:rlo-wlo}
\end{equation}
Recall that $\psi_{ij}(x,y)$ is the simplified formula of $\psi(x,y)\land\psi(y,x)$, where we have replaced unary and binary reflexive literals by true or false appropriately assuming that $x$ realized the cell $C_i$ and $y$ realized $C_j$.
The terms $\rlo_{ij}$ are now defined with respect to $\psi^<_{ij}$, meaning we enforce $x<y$; thus, it no longer holds that $\rlo_{ij} = \rlo_{ji}$.
Moreover, $\wlo_k$ is again defined with respect to the non-simplified formula $\psi^\le$ so that we do not remove the unary and binary reflexive literals from the weighted model count entirely.

Finally, we use the newly defined values in \cref{alg:iwfomc}.
That is all we need to ensure that $T_i(\bm{k})$ is the sum of weighted models in $\fomodels{\Psi}{i}^\le$ whose cell configuration is $\bm{k}$.

\begin{lemma}
\label{lemma:iwfomc-fo2+lo}
\rev{
The value $\gamma$ in \cref{alg:iwfomc} computes $\wfomc(\Psi \land \lonat(\lopred), n, \weights)$ for an \fotwo sentence $\Psi$ with $\lopred\in\preds{\Psi}$.
Moreover, the algorithm runs in time polynomial in the domain size $n$.
}
\end{lemma}

\begin{proof}
The proof is very similar to the proof of \cref{lemma:iwfomc-fo2}.
The main difference lies in using the subformulas $\psi_{ij}^<$ and $\psi_{ij}^\le$ instead of $\psi_{ij}$.

In the case of $n=1$, we have
 \begin{equation*}
        \begin{aligned}
        \Phi_\mathcal{C} = \psi_{ll}^\le(1,1) \quad \text{and} \quad \symbwmc(\Phi_\mathcal{C}, \weights) = \wlo_l, 
        \end{aligned}
\end{equation*}
hence, $\gamma = T_1(\vecdelta_l) = \wlo_l.$

In the induction step, we obtain
\begin{equation*}
  \begin{aligned}
  \Phi_{\mathcal{C}_2} = \Phi_{\mathcal{C}_1} \land \psi_{ll}^\le(i+1,i+1) \land \bigwedge_{j\in[p]} \bigwedge_{i'\in C_j} \psi_{jl}^<(i',i+1).
  \end{aligned}
\end{equation*}
The subformula $\psi_{ll}^\le(i+1,i+1)$ enforces reflexivity for the new element, i.e., $(i+1)\le (i+1)$.
The subformula $\psi_{jl}^<(i',i+1)$ then requires $i'<(i+1)$ for all already processed elements $i' \in [i]$, i.e., by the induction hypothesis, we will have the property $1\le 2\le\ldots\le i\le i+1$.
Note that the argument of independent conjuncts still applies as well; thus, we obtain
\begin{equation*}
  \begin{aligned}
  \symbwmc(\Phi_{\mathcal{C}_2}, \weights) = \symbwmc(\Phi_{\mathcal{C}_1}, \weights)\cdot \wlo_l \cdot \prod_{j\in[p]} \rlo_{jl}^{|C_j|}.
  \end{aligned}
\end{equation*}
Clearly, then, it holds that
\begin{equation*}
        T_{i+1}(\bm{k})\! =\! \sum_{l\in[p]} T_i(\bm{k} - \vecdelta_l)\! \cdot\! \wlo_l\! \cdot\! \prod_{j\in[p]} \rlo_{jl}^{(\bm{k} - \vecdelta_l)_j}.
    \end{equation*}
    
The rest of the proof remains the same, including the time complexity part. 
\end{proof}

Combining \cref{lemma:iwfomc-fo2+lo} and \cref{cor:n-factorial}, we obtain both correctness and time complexity of the entire \cref{alg:iwfomc} and we can state the following theorem:
\begin{theorem}
\label{th:fo2+lo_liftable}
The two-variable fragment of first-order logic with the linear order axiom is domain-liftable.   
\end{theorem}

The result can be easily extended to the fragment of $C^2$ using \cref{lemma:c2+cc}.

\subsection{Processing Evidence}
\label{ssec:iwfomc-evidence}

\rev{
So far, we have shown that \cref{alg:iwfomc} can compute \wfomc of a sentence of the form $\forall x \forall y: \psi(x, y) \land \loaxiom(\lopred)$ in a domain-lifted way.
However, $\psi(x, y)$ must be constant-free, i.e., we do not support conditioning on evidence.
Let us now adapt the algorithm to address that issue.
Recall that supporting arbitrary conditioning is intractable, but conditioning on unary literals can be handled in a domain-lifted way \citep{broeckdavis12:wfomc-evidence}.
While \citet{kula26:evidence-gaifman-tw} have shown that some binary evidence is liftable as well, it depends on a specialized algorithm working over the evidence's Gaifman graph.
Therefore, we will only aim to account for unary evidence in this section.

As we first introduced \cref{alg:iwfomc} for computing WFOMC over the \fotwo fragment, let us now again start with the same fragment alone.
We will add the linear order as the next step.
Recall that for \fotwo, we use the values $w_k$ and $r_{ij}$ rather than $\wlo_k$ and $\rlo_{ij}$, and we return the value $\gamma$ (\cref{iwfomc:line:gamma}) rather than its multiple (\cref{iwfomc:line:final-sum}).

Supporting unary evidence in the context of \cref{alg:iwfomc} for \fotwo is quite simple due to our use of (valid) cells.
Over the course of the computation, we repeatedly enforce that a particular domain element $i$ realizes (is assigned to) a particular cell $C_l(x)$, i.e., we restrict ourselves to possible worlds $\omega$ such that $\omega \models C_l(i)$.
If we want to accommodate $\evidence{i}$ on top of that, we may assign the element $ i $ only to cells $ C_l $ such that $C_l(i) \land \evidence{i}$ is satisfiable.
In other words, we restrict ourselves to possible worlds $\omega$ such that $\omega \models C_l(i) \land \evidence{i}$.
Moreover, note that, up to the value $\gamma$, the algorithm implicitly uses a particular domain ordering, which is explicitly enforced when we use the values $\wlo_k$ and $\rlo_{ij}$.
Therefore, the very same technique works not only for the \fotwo fragment alone but also for \fotwo sentences with the natural ordering of elements encoded by the axiom $\lonat$.

Let us explicitly formulate the changes to \cref{alg:iwfomc} in order to support evidence in \fotwo+$\lonat$ sentences.
First, the base case on \cref{iwfomc:line:init} becomes
\begin{equation*}
    T_1(\bm{\vecdelta}_l) = \begin{cases}
        \;\wlo_l & \text{if } C_l(1) \land \evidence{1} \text{ is satisfiable,}\\
        \;0 & \text{otherwise}.
    \end{cases}
\end{equation*}
When we initialize the table $T_1$ (i.e., when processing the domain element 1), we essentially filter out those cells $C_l$ such that $C_l(1)$ is not compatible with $\evidence{1}$.
In other words, $\evidence{1}$ restricts which cells may be realized by the element 1.
Second, for each $i$ chosen on \cref{iwfomc:line:main_loop}, the loop on \cref{iwfomc:line:cell_loop} only goes over $l\in[p]$ such that $C_l(i) \land \evidence{i}$ is satisfiable, i.e., when processing the $i$-th domain element, we again limit our selection of cells realized by the element $i$ only to those compatible with $\evidence{i}$.
}

\begin{example}
    Consider the sentence $\Gamma = (\forall x \forall y: S(x) \to R(x,y)) \land S(1) \land \neg S(2)$.
    
    Denote the valid cells of $\forall x \forall y: (S(x) \to R(x,y))$ as
    \begin{align*}
        C_1(x) &= \neg S(x) \wedge \neg R(x,x), \\
        C_2(x) &= \neg S(x) \wedge R(x,x),\\
        C_3(x) &= S(x) \wedge R(x,x).  
    \end{align*}
    Let us sketch the work of \cref{alg:iwfomc} on $\Gamma$.
    Note that $\Gamma$ does not contain the linear order axiom; hence, we will make use of values $r_{ij}$ and $w_k$ from \Cref{eq:r-w} rather than values $\rlo_{ij}$ and $\wlo_k$ from \Cref{eq:rlo-wlo}.
    The algorithm will work as follows:
    \begin{enumerate}
        \item Since $\evidence{1}\land C_1(1)$ and $\evidence{1}\land C_2(1)$ are both unsatisfiable formulas, the table $T_1$ will be initialized such that $T_1(1,0,0)=T_1(0,1,0)=0$ and $T_1(0,0,1)=w_3$.

        \item For $i=2$, we will only go over $l\in\{1,2\}$ since $\evidence{2}\land C_3(2)$ is unsatisfiable.
        \begin{itemize}
            \item For $l=1$, we will have
                \begin{align*}
                    T_2(2,0,0) &= T_1(1,0,0)\cdot w_1 \cdot r_{11} = 0,\\
                    T_2(1,1,0) &= T_1(0,1,0)\cdot w_1 \cdot r_{12} = 0,\\
                    T_2(1,0,1) &= T_1(0,0,1)\cdot w_1 \cdot r_{13} = w_1\cdot w_3 \cdot r_{13},\\
                \end{align*}

            \item and for $l=2$, there will be
                \begin{align*}
                    T_2(1,1,0) &= T_2(1,1,0) + T_1(1,0,0)\cdot w_2 \cdot r_{12} = 0,\\
                    T_2(0,2,0) &= T_1(0,1,0)\cdot w_2 \cdot r_{22} = 0,\\
                    T_2(0,1,1) &= T_1(0,0,1)\cdot w_2 \cdot r_{23} = w_2\cdot w_3 \cdot r_{23}.\\
                \end{align*}
        \end{itemize}

        \item For $i\ge3$, $l$ will no longer be restricted and will go over all $\{1,2,3\}$.
    \end{enumerate}
\end{example}

\rev{
It may not be immediately obvious why the same procedure, extended by the multiplication by $n!$ at the end, would not work for the linear order axiom $\loaxiom$.
Since we are excluding conditioning on binary relations, we avoid evidence of the form $a < b$ for some two domain elements $a,b\in\dom$, which would clearly break the symmetry of the problem, invalidating \cref{cor:n-factorial}.
However, symmetry may also be broken by using unary evidence, as we demonstrate in the following example.
\begin{example}
    Consider the sentence $$\Gamma = \loaxiom(\lopred) \land (\forall x \forall y: ((H(x) \land T(y)) \to x\lopred y) \land (H(x) \lor T(x)) \land (\neg H(x) \lor \neg T(x))) \land H(2) \land T(1).$$
    For each domain element, either $H/1$ or $T/1$ must hold, but not both at the same time; i.e., each element is either in the sequence head or the tail.
    
    First, if we algorithmically process the domain elements with the explicit order $1\le2\le\ldots\le n$, we will find no models.
    However, even on a domain of size $n=2$, the problem has the model $\omega=\{H(2), T(1), 2\le 1, 2\le2,1\le1\}$.

    Second, even if we find the model $\omega$ (or its corresponding model count of 1), it is not correct to multiply by $n!=2$, since, as we have already pointed out, the ordering $1\le2\le3$ does not produce a model.
\end{example}

To support evidence in the presence of the linear order axiom $\loaxiom$, another approach is thus required.
A possible solution is to encode the evidence as \emph{cardinality constraints}, as in \citet{wang24:wfoms}, and then address them instead (using the techniques from \citet{kuzelka21:wfomc-c2}, i.e., \cref{lemma:c2+cc}).
Although \citet{wang24:wfoms} did not work with a linear order, it is not hard to see that their encoding works even in the presence of the axiom $\loaxiom$.
See \cref{app:evidence-by-ccs} for details.
}

\section{Successor Relations}
\label{sec:4-successor-relations}
Using the linear order axiom, we have computed WFOMC over an ordered domain in a domain-lifted manner.
Therefore, we can model various scenarios that require sequential behavior of the domain.
For instance, consider hidden Markov models, where one needs to encode the transition probabilities between hidden states, or combinatorics problems in which we arrange objects in a row so that those satisfying a particular property must be next to each other.
Note, however, that to reason about such scenarios, we may need more than a single binary relation interpreted as a linear order.
For a given domain element, we may need to access its predecessor or successor in the linear order.
We may even require access to the $k$-th predecessor or the $k$-th successor.
Naturally, such information is implicitly contained within the ordered domain, but can we obtain it explicitly?
Moreover, can we do so without incurring additional computational costs?

To answer the questions above, let us start by formally defining the successor relation we wish to study.
\rev{
We will focus on the arbitrary linear order with the axiom $\loaxiom$ while again solving the problem for $\lonat$ in the process.
}
Let us denote the relation that \emph{$y$ is the (immediate) successor of $x$} as $Succ_1(x, y)$.
Naturally, we may also say that \emph{$x$ is the (immediate) predecessor of $y$}. 
The relation will be defined with respect to the linear ordering of the domain enforced by the axiom $\loaxiom(\lopred)$.
To define $Succ_1(x, y)$ explicitly using first-order logic, we may write that
\begin{equation}
    \forall x\forall y: Succ_1(x,y) \leftrightarrow \left((x < y) \land (\neg \exists z: (x < z) \land (z < y))\right).
  \label{eq:successor}
\end{equation}
If needed, we may also use the relation $Succ_1$ to inductively define the \emph{$k$-th successor relation} such as
\begin{equation}
     \forall x\forall y: Succ_k(x,y) \leftrightarrow (\exists z: Succ_{k-1}(x,z) \land Succ_1(z,y)),
     \label{eq:kth-successor}
\end{equation}
where $Succ_k(x,y)$ expresses that $y$ is the $k$-th successor of $x$.

Clearly, \Cref{eq:successor,eq:kth-successor} would not allow us to perform domain-lifted computations, since they require three distinct variables.
Before we address that issue, let us consider two examples that showcase the utility of successor relations.

\begin{example}
\label{ex:basic-succ}
Consider modeling the following problem from combinatorics:
    ``How many ways can we put 3 math books and 5 English books on a shelf if all the math books must stay together and all the English books must also stay together?''
    
    Using the predicate $Succ_1(x, y)$ defined with respect to a linear order, we may encode the task as follows:
    \begin{align*}
        &\left(\forall x: (m(x) \lor e(x)) \land \neg (m(x) \land e(x))\right) \\
        \land\;&\left(\forall x: e_{first}(x) \leftrightarrow \left( e(x) \land \neg(\exists y: e(y) \land Succ_1(y,x))\right)\right) \\
        \land\;&\left(\forall x: m_{first}(x) \leftrightarrow \left( m(x) \land \neg(\exists y: m(y) \land Succ_1(y,x))\right)\right) \\
        \land\;&(|e_{first}| \le 1) \land (|m_{first}| \le 1) \land (|m| = 3) \land (|e| = 5),
    \end{align*}
where $e(x)$ and $m(x)$ denote the book $x$ to be an English book or a math book, respectively, and $e_{first}$ and $m_{first}$ denote the first English book and the first math book on the shelf.
The cardinality constraints $(|e_{first}| \le 1)$ and $(|m_{first}| \le 1)$ ensure that all English books and all math books are together.
\end{example}

\begin{example}
\label{ex:grid-kxn}
    Consider encoding a grid $k \times n$ where $k$ is a constant.
    Assuming an ordered domain with access to the immediate and $k$-th successor relations, we can encode the grid as shown in \cref{fig:grid-kxn}.
    The $Succ_1$ relation represents the vertical (blue) adjacency, and the $Succ_k$ relation represents the horizontal (red) adjacency.
    See \cref{app:grid-by-successors} for an exact encoding.

\begin{figure}[tbp]
    \centering
    \begin{tikzpicture}[
        % Modern syntax for defining node styles
        roundnode/.style={circle, draw, inner sep=0pt, minimum size=2mm}
    ]
        % Parameters (Renamed to avoid 'calc' library conflicts)
        \def\cols{8}   % number of columns (was \n)
        \def\rows{4}   % number of rows (was \k)
        \def\xstep{1.0}
        \def\ystep{1.0}
        
        % Draw nodes
        \foreach \i in {1,...,\cols}{
            \foreach \j in {1,...,\rows}{
                \node[roundnode] (v\i\j) at ({\xstep*\i},{-\ystep*\j}) {};
            }
        }
        
        % Horizontal red arrows
        \foreach \j in {1,...,\rows}{
            \foreach \i [evaluate=\i as \ip using int(\i+1)] in {1,...,\numexpr\cols-1}{
                \path[->, red] (v\i\j) edge (v\ip\j);
            }
        }
        
        % Vertical blue arrows (upward)
        \foreach \i in {1,...,\cols}{
            \foreach \j [evaluate=\j as \jp using int(\j+1)] in {1,...,\numexpr\rows-1}{
                \path[<-, blue] (v\i\jp) edge (v\i\j);
            }
        }
        
        % Diagonal blue arrows (upright)
        \foreach \i [evaluate=\i as \j using int(\i+1)] in {1,...,\numexpr\cols-1}{
            \path[->, blue] (v\i\rows) edge (v\j1);
        }
        
        % Braces and labels (Using the safe \cols and \rows variables)
        \draw[decorate, decoration={brace, amplitude=6pt}] 
            ($(v11)+(-0.1,0.2)$) -- ($(v\cols1)+(0.1,0.2)$)
            node[midway, yshift=10pt]{$n$};
            
        \draw[decorate, decoration={brace, amplitude=6pt}] 
            ($(v1\rows)+(-0.2,-0.1)$) -- ($(v11)+(-0.2,-0.1)$)
            node[midway, xshift=-12pt]{$k$};
            
    \end{tikzpicture}
    \caption{A grid encoded using the immediate and the $k$-th successor relations.}
    \label{fig:grid-kxn}

    \Description{The figure shows domain elements represented as vertices formed into a grid. The grid is enforced by a vertical adjacency, representing the immediate successor relation, going downwards between lines and diagonally from the last line to the first one, and a horizontal adjacency going left-to-right between columns, representing the $k$-th successor relation.}
    
\end{figure}
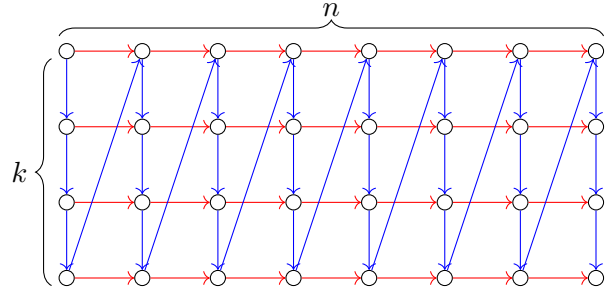
\end{example}

Instead of the three-variable equivalences in \Cref{eq:successor,eq:kth-successor}, having a linear order relation $\lopred$, we may express the successor relationships using $\ctwo$ sentences.
For instance, to encode $Succ_1(x, y)$, we could proceed as follows:
\begin{itemize}
    \item First, we would denote the first and the last element in the sequence:
    \begin{equation}
    \label{eq:succ1:first-last}
        \begin{aligned}
            &(\forall x : First(x) \leftrightarrow \forall y: \ (x\lopred y))   \\
            \land\;&(\forall x : Last(x) \leftrightarrow \forall y: \ (y\lopred x)) 
        \end{aligned}
    \end{equation}

    \item Next, we would require that every element except for the last has exactly one successor, and every element except for the first has exactly one predecessor:
    \begin{equation}
    \label{eq:succ1:bijection-like}
        \begin{aligned}
            &(\forall x : \lnot Last(x) \to \exists^{=1} y: \ Succ_1(x,y)) \\
            \land\; &(\forall x : \lnot First(x) \to \exists^{=1} y: \ Succ_1(y,x)) \\
        \end{aligned}
    \end{equation}

    \item Finally, we would restrict the relation to only \emph{go from left to right}:
    \begin{equation}
    \label{eq:succ1:left2right}
        \begin{aligned}
            \forall x \forall y &: Succ_1(x,y) \to (x < y)    
        \end{aligned}
    \end{equation}        
\end{itemize}
We refer the readers to \Cref{app:successor-by-lo} to check the correctness of the sentences as well as their possible generalization to the $k$-th successor relation.

Let us denote the conjunction of \Cref{eq:succ1:first-last,eq:succ1:bijection-like,eq:succ1:left2right} by $\xi$.
The sentence $\xi$ captures the immediate successor relationship with respect to the linear order relation $\lopred$ using the fragment of \ctwo.
By \cref{lemma:c2+cc}, we may, therefore, use \Cref{alg:iwfomc} to perform domain-lifted inference over an ordered domain while also having access to the immediate successor relation.
Unfortunately, the encoding comes with non-negligible additional computational costs.
\begin{remark}
\label{remark:n^cp}
    Suppose we want to compute $\symbwfomc(\Psi \land \xi \land \loaxiom(\lopred), n, \weights)$.
    If $\Psi$ has $p$ (valid) cells, the input sentence will have $4 \cdot c \cdot p$ (valid) cells where $c$ is a positive integer constant.
\end{remark}
The multiple $4=2^2$ comes from the two fresh relations $First$ and $Last$.
While we could, in the context of \cref{alg:iwfomc}, handle the new relations using evidence (filtering out most of the new cells), since we know that only two positive literals, specifically $First(1)$ and $Last(n)$, must hold, the primary problem is the other factor.
The constant $c$ is due to the added complexity of dealing with counting quantifiers, which is a consequence of the transformations underlying \cref{lemma:c2+cc}.
The form of $c$ has been analyzed in more detail by \citet{toth24:c2-bound-to-beat}.

Recall that the number of cells determines the degree of our polynomial runtime.
While from a theoretical standpoint, we maintain polynomial complexity in $n$, in practice, the polynomial degree becomes too high.%
\footnote{Note that we may develop more efficient successor encodings than the one in $\xi$. However, all encodings known to us incur a similar increase in the polynomial's degree.}
Let us, therefore, develop a different way of handling the successor relations rather than using raw \ctwo to access them via the linear order axiom.
Namely, we generalize \cref{alg:iwfomc} to support an extended version of the linear order axiom, which provides access to successor relations directly.

\subsection{A Faster Algorithm for the Linear Order Axiom with its Successor Relations}
In this section, we devise a second \wfomc algorithm that natively supports the linear order axiom along with its successor relations, and exhibits a much better runtime than using \cref{alg:iwfomc} along with the sentences in \Cref{eq:succ1:first-last,eq:succ1:bijection-like,eq:succ1:left2right} would.

Let us start with the simplest case: working only with the linear order and its immediate successor relation.
Recall that in \cref{alg:iwfomc}, $T_m(\bm{k})$ denotes the weighted sum of models with a given cell configuration $\bm{k}$ over an ordered domain $1\leq 2\leq \dots \leq m$.
We can view the update from $T_m(\bm{k})$ to $T_{m+1}(\bm{k'})$ as appending a new element $(m+1)$ to the ordered domain such that $i < m+1$ for all elements $i \leq m$.
If we further require that $T_m(\bm{k})$ only accounts for the models consistent with the immediate successor relation $Succ_1$, then we can use the update rule even for computing \wfomc with the immediate successor relation.
However, how can we enforce such an additional requirement?
As it turns out, accomplishing that is easier than it may seem.
In addition to the already stored cell configuration, we also need to record the cell assignment of the previous element $m$ to ensure that the updated weight considers only models containing $Succ_1(m, m+1)$.

Let $\fomodels{\Psi}{n}^{\leq,Succ_1}$ denote the set of models of $\Psi$ over the domain $[n]$ such that the predicates $\leq$ and $Succ_1$ are interpreted as a standard linear order relation and its immediate successor relation, respectively.
We use $T_m(\bm{k}, t)$ to denote the weighted sum of models in $\fomodels{\Psi}{m}^{\leq,Succ_1}$ such that the element $m$ is assigned to the cell $C_t$ and the cell configuration is $\bm{k}$.
The weights $\rlo_{ij}$ and $\wlo_k$ are redefined as
\begin{equation*}
  \begin{aligned}
    \rlo_{ij} &= \symbwmc(\psi_{ij}^<(a,b)\land \neg Succ_1(a,b)\land  \neg Succ_1(b,a),\weights),\\
    \wlo_k &= \symbwmc(\psi^\le(c,c) \land C_k(c) \land \neg Succ_1(c,c), \weights).
  \end{aligned}
\end{equation*}
Informally, $\rlo_{ij}$ represents the 2-tables containing two elements that are not adjacent in the linear order.
To account for two consecutive elements, we introduce a new weight $\rpred_{ij}$ as
\begin{equation}
  \begin{aligned}
    \rpred_{ij} = \symbwmc(&\psi_{ij}^<(a,b)\land Succ_1(a,b) \land\neg Succ_1(b,a), \weights).
  \end{aligned}
  \label{eq:rpred}
\end{equation}
Note again that both $\rlo_{ij}$ and $\rpred_{ij}$ are not symmetric and $\wlo_k$ is defined in terms of the non-simplified formula $\psi^\le$ rather than $\psi_{ij}^\le$.

Then $T_{m+1}(\bm{k}+\vecdelta_t, t)$ can be computed as
\begin{equation}
  \begin{aligned}
  T_{m+1}(\bm{k}+\vecdelta_t, t) = \Bigg(\sum_{l\in[p]} T_m(\bm{k}, l) \cdot \rlo_{lt}^{k_l-1}\cdot \rpred_{lt} \cdot \prod_{i\in[p]:i\neq l} \rlo_{it}^{k_i}\Bigg) \cdot \wlo_{t}.
  \end{aligned}
  \label{eq:wfomc_incremental_pred1}
\end{equation}
Intuitively, we want to evaluate the weighted model count on the domain $[m+1]$ with the cell configuration $\bm{k}+\vecdelta_t$ (i.e., such that the element $m+1$ realizes the cell $C_t$ and $\bm{k}$ is the cell configuration for the domain $[m]$).
Since the new element realizes $C_t$, we multiply by $\wlo_t$.
Then we sum over all possible assignments of the element $m$; and for a particular assignment of $m$ to the cell $C_l$, we do the following:
First, we multiply by $\rlo_{lt}$ for each domain element $c\in C_l$ such that $c<m$, since those elements are not (immediate) predecessors of $m+1$.
Second, we multiply by $\rpred_{lt}$ to account for the consecutive pair $(m,m+1)$.
Finally, we multiply by $\rlo_{it}$ for all other domain elements assigned to $C_i$ ($i\neq l$) to account for all the remaining non-consecutive pairs of elements.

Finally, the \wfomc{} of a sentence with the immediate successor relation $Succ_1$ can be computed as 
\begin{equation}
n!\cdot \sum_{\bm{k}\in\mathbb{N}^p: |\mathbf{k}|=n}\sum_{t\in[p]} T_n(\bm{k}, t).    
\end{equation}

Now, let us generalize to the case when one wishes to work with the $k$-th successor relation.
Note that the successor relations are all still with respect to the domain ordering defined by a linear order axiom $\loaxiom(\lopred)$.
For such purposes, we introduce an \emph{extended linear order axiom} denoting not only the linear order relation itself, but also its successor relations.

\begin{definition}[Extended Linear Order Axiom]
    The extended linear order axiom on a relation $\lopred$, denoted as $\loaxiom(\lopred, Succ_1, \ldots, Succ_k)$ requires that $\lopred$ is a linear order of elements and $Succ_i$ is the $i$-th successor relation with respect to the ordering defined by $\lopred$ for each $i \in [k]$.
\end{definition}

We now generalize \cref{alg:iwfomc} to support the extended linear order axiom.
We directly present the algorithm in~\cref{alg:iwfomc2} without repeating the verbose explanation, which should be clear from the previous section.
We use $T_m(\bm{k}, t_1, \dots, t_k)$ to denote the weighted sum of models on an ordered domain such that the element $(m-s+1)$ is assigned to the cell $C_{t_s}$ for every $s\in[\min(m,k)]$, and the cell configuration is $\bm{k}$.
The weights $\rlo_{ij}$ and $\wlo_k$ are now defined respectively as
\begin{align*}
  \rlo_{ij}&=\symbwmc(\psi_{ij}^<(a,b)\land \bigwedge_{s\in[k]} (\neg Succ_s(a,b)\land \neg Succ_s(b,a)), \weights),\\
  \wlo_k &=\symbwmc(\psi^\leq(c,c) \land C_k(c) \land \bigwedge_{s\in[k]} \neg Succ_s(c,c), \weights).\\
\end{align*}
For each $s\in[k]$, we define the weight $\rpred_{ij,s}$ as
\begin{align*}
  \rpred_{ij,s} = \symbwmc(&\psi_{ij}^<(a,b)\land Succ_s(a,b) \land \neg Succ_s(b,a)\land
  \bigwedge_{s'\in[k]: s' \neq s} (\neg Succ_{s'}(a,b)\land \neg Succ_{s'}(b,a)), \weights).
\end{align*}

\cref{alg:iwfomc2} performs dynamic programming similar to \cref{alg:iwfomc}.
For $m=1$, the algorithm initializes $T_1(\vecdelta_l, l, \dots, l)$ for every $l\in[p]$ (\cref{line:base_case}).
For $m \ge 2$, the value of $T_m(\bm{k}, t_1, \dots, t_k)$ is computed in the loop on \cref{line:cell_loop}.
We note that when updating $T_m$, only the first $s$ successor relations that already exist in $[m]$ are considered (\cref{line:pred_loop}).

\begin{algorithm*}[tbp]
  \caption{WFOMC for \fotwo+$\loaxiom(\lopred, Succ_1,\dots, Succ_k)$}
  \label{alg:iwfomc2}
  \DontPrintSemicolon
  \KwIn{An \fotwo sentence $\Psi = \forall x\forall y: \psi(x,y)$ with $\{\lopred, Succ_1,\dots, Succ_k\} \subseteq \preds{\Psi}$, integer $n$, functions $(\weights)$}
  \KwOut{$\symbwfomc(\Psi \land \loaxiom(\lopred, Succ_1,\dots, Succ_k), n, \weights)$}
  Compute $\wlo_{i}, \rlo_{ij}$ and $\rpred_{ij,s}$ for every $i,j\in[p]$ and $s\in[k]$\\
  $T_1(\vecdelta_l, l, \dots, l) \leftarrow \wlo_l$ for each $l\in[p]$ \label{line:base_case} \tcp*{$k$ copies of $C_l$}
  \ForEach{$i=2$ \KwTo $n$}{ \label{line:main_loop}
    $T_i(\bm{k}, t_1,\dots,t_k) \leftarrow 0$ for every $\bm{k}\in\mathbb{N}^p$, $t_1,\dots,t_k\in[p]$\\
    \ForEach{$l\in[p]$}{ \label{line:cell_loop}
      \ForEach{$(\bm{k}_{old}, t_1, \dots t_k, W_{old}) \in T_{i-1}$}{ \label{line:old_loop}
        $\bm{k}_{new} \leftarrow \bm{k}_{old} + \vecdelta_l$\\
        $\bm{k}, W \gets \bm{k}_{old}, W_{old}$ \\
        \ForEach{$s\in[\min(i-1, k)]$}{ \label{line:pred_loop}
          $W \leftarrow W \cdot \rpred_{t_sl,s}$ \tcp*{deal with the $s$-th predecessor}
          $\bm{k} \leftarrow \bm{k} - \vecdelta_{t_s}$\tcp*{remove one element from the $s$-th cell}
          \lIf(\tcp*[f]{update the cell assignment for $T_i$}){$s = 1$}{
            $t_s' \leftarrow l$ 
            \textbf{else} $t_s' \leftarrow t_{s-1}$ 
          }
        }
        $T_i(\bm{k}_{new}, t_1',\dots,t_k') \leftarrow T_i(\bm{k}_{new}, t_1',\dots,t_k') + W\cdot \wlo_l \cdot \prod_{j\in[p]} \rlo_{jl}^{k_j}$ \label{line:pred_update}
      }
    }\label{line:cell_end}
  }
  \Return $n!\cdot \sum_{\bm{k}\in\mathbb{N}^p: |\mathbf{k}|=n}\sum_{t_1,\dots,t_k\in[p]} T_n(\bm{k}, t_1,\dots,t_k)$
\end{algorithm*}

\begin{theorem}\label{thm:general_linear_order}
  The two-variable fragment of first-order logic with the extended linear order axiom %$\loaxiom(\le, Succ_1, \dots, Succ_k)$
  is domain-liftable.
\end{theorem}

\begin{proof}
To prove the claim, it is sufficient to show both correctness and polynomial-time complexity of \cref{alg:iwfomc2}.
The correctness can be shown by a similar argument as for \cref{alg:iwfomc}.
We will point out only the critical step of the induction.

Any \fotwo{} sentence with the extended linear order axiom can be transformed into $$\Psi = \forall x\forall y: \psi(x,y)\land \loaxiom(\leq, Succ_1, \dots, Succ_k)$$ by the same transformation as for the \fotwo{} fragment.
Consider the \wfomc{} of $\Psi$ over a domain of size $n$ under weighting functions $(\weights)$.

Let $\fomodels{\Psi}{n}^{\leq,Succ_1,\dots,Succ_k}$ denote the set of models of $\Psi$ over the domain $[n]$ where the predicates $\leq$, $Succ_1, \dots, Succ_k$ are interpreted as a standard order relation and its successor relations, respectively.
We show that $T_m(\bm{k}, t_1, \dots, t_k)$ computed in \cref{alg:iwfomc2} is the weighted sum of models in $\fomodels{\Psi}{m}^{\leq,Succ_1,\dots,Succ_k}$ such that the cell configuration is $\bm{k}$ and the element $\max(m-s+1, 1)$ is assigned to the cell $C_{t_s}$ for every $s\in[k]$.

Consider any $m\in[n]$.
For every $s\in[k]$, denote by
$$\Omega_s^m = \{(1, 1+s), (2, 2+s), \dots, (m - s, m)\}$$
the set of consecutive pairs of elements in the ordered domain $1\le 2\le \dots \le m$ with a distance of $s$.
Let $$\Omega = \bigcup_{s\in[k]} \Omega_s^m.$$
Define 
\begin{align*}
  \fotwoformula_{ij}^{\leq,\bot}(x,y) = &\fotwoformula_{ij}^\leq(x,y)\land \bigwedge_{s\in[k]} \neg Succ_s(x,y),\\
  \fotwoformula_{ij}^{<,\bot}(x,y) = &\fotwoformula_{ij}^<(x,y)\land \bigwedge_{s\in[k]} \Big(\neg Succ_s(x,y)\land \neg Succ_s(y,x)\Big),
\end{align*}
and for every $s\in[k]$, define
\begin{align*}
  \fotwoformula^{<, s}_{ij}(x,y) = \fotwoformula_{ij}^{<}(x,y)\land \bigwedge_{t\in[k]: t\neq s} \Big(\neg Succ_t(x,y)\land \neg Succ_t(y,x)\Big)\land Succ_s(x,y)\land \neg Succ_s(y,x).
\end{align*}
Consider a cell partition $\mathcal{C} = (C_1, \dots, C_p)$ of $[m]$.
The ground formula $\Phi_{\mathcal{C}}$ from \Cref{eq:wfomc_lineage} can now be written as
\begin{align*}
  \Phi_{\mathcal{C}} = &\bigwedge_{i,j\in[p]: i<j}\Bigg(\bigwedge_{\substack{a\in C_i, b\in C_j:\\ a < b\land (a,b)\notin \Omega}} \fotwoformula_{ij}^{<,\bot}(a,b)\land \bigwedge_{s\in[k]}\bigwedge_{\substack{a\in C_i, b\in C_j:\\ (a,b)\in \Omega_s^m}} \fotwoformula_{ij}^{<,s}(a,b)\Bigg)\land \\
  &\bigwedge_{i\in[p]} \Bigg(\bigwedge_{\substack{a\in C_i, b\in C_i:\\ a < b\land \{a,b\}\notin \Omega}} \fotwoformula_{ii}^{<,\bot}(a,b)\land \bigwedge_{s\in[k]}\bigwedge_{\substack{a\in C_i, b\in C_i:\\ \{a,b\}\in \Omega_s^m}} \fotwoformula_{ii}^{<,s}(a,b)\land \bigwedge_{c\in C_i} \fotwoformula_{ii}^{\le,\bot}(c,c)\Bigg).
\end{align*}
Similar to \Cref{eq:ground_induction}, we can write the relationship between two partitions $\mathcal{C}_1 = (C_1, \dots, C_p)$ and $\mathcal{C}_2 = (C_1, \dots, C_l\cup \{m+1\}, \dots, C_p)$ as
\begin{align*}
  \Phi_{\mathcal{C}_2} = \Phi_{\mathcal{C}_1}\land \bigwedge_{i\in[p]}\Bigg(\bigwedge_{\substack{a\in C_i: \\(a,h+1)\notin \Omega}} \fotwoformula_{il}^{<,\bot}(a,m+1)\land \bigwedge_{s\in[k]}\bigwedge_{\substack{a\in C_i:\\ (a,m+1)\in \Omega_s^{m+1}}} \fotwoformula_{il}^{<,s}(a,m+1)\Bigg)\land \fotwoformula_{ll}^{\le,\bot}(m+1,m+1).
\end{align*}
Notice that each conjunct in the above formula is independent, and all these conjuncts are independent of $\Phi_{\mathcal{C}_1}$.
Define
\begin{align*}
  \rlo_{ij} &:= \symbwmc(\fotwoformula_{ij}^{<,\bot}(a,b), \weights), \qquad \wlo_i := \symbwmc(\fotwoformula^{\leq,\bot}(c,c), \weights)
\end{align*}
and for every $s\in[k]$, define
\begin{align*}
  \rpred_{ij,s} &= \symbwmc(\fotwoformula_{ij}^{<,s}(a,b), \weights).
\end{align*}
Then the weighted model count of $\Phi_{\mathcal{C}_2}$ can be computed from the weighted model count of $\Phi_{\mathcal{C}_1}$ as
\begin{align*}
  \symbwmc(\Phi_{\mathcal{C}_2}, \weights) = \symbwmc(\Phi_{\mathcal{C}_1}, \weights) \cdot \wlo_l \cdot \prod_{i\in[p]}\left( \prod_{\substack{a\in C_i\\ (a,m+1)\notin \Omega}} r_{il} \cdot \prod_{i\in[p]} \prod_{\substack{a\in C_i\\ (a,m+1)\in \Omega_s^m}} \rpred_{il,s}\right).
\end{align*}
The above equation directly leads to the update of $T_{m+1}(\bm{k}_{new}, t_1, \dots, t_k)$ in \cref{alg:iwfomc2} (from \cref{line:pred_loop} to \cref{line:pred_update}).
By the meaning of $T_m(\bm{k}, t_1, \dots, t_k)$, the correctness of the algorithm follows.

Next, let us consider the algorithm's complexity.
The outer loop at \cref{line:main_loop} runs $\mathcal{O}(n)$ times, and the loop for cells at \cref{line:cell_loop} is in $\mathcal{O}(p)$.
The inner loop goes through all possible $(\bm{k}_{old}, t_1, \dots, t_k, W_{old})$ in $T_{m-1}$, whose number can be bounded as $\mathcal{O}(p^k\cdot n^{p-1})$.
Thus, the overall complexity is $\mathcal{O}(p^{k+1}\cdot n^p)$, which is polynomial in $n$.
\end{proof}

Compare the newly established time complexity of \cref{alg:iwfomc2} with the bound in \cref{remark:n^cp}.
Previously, the polynomial's degree was multiplied by a potentially high constant; now, we only have a constant as the polynomial's coefficient.

\begin{remark}
  When there are inconsecutive successor relations in the linear order axiom, e.g., $\loaxiom(\lopred, Succ_1, Succ_{10})$, the complexity of \cref{alg:iwfomc2} is still $\mathcal{O}(p^{k+1}\cdot n^p)$, where $k$ is the highest order of the successor relations, as we need to maintain the cell assignment of the last $k$ elements in $T_{m-1}$ in order to obtain the correct $t_k$ for $T_m$ (which is $t_{k-1}$ in $T_{m-1}$).
\end{remark}

\begin{remark}
  \cref{alg:iwfomc2} can be modified to support unary evidence similarly as \cref{alg:iwfomc} was modified in \cref{ssec:iwfomc-evidence}.
\end{remark}

\subsection{Cyclic Successor Relation}
\label{ssec:cyclic-successor}
In some applications, we may need to consider the (closed) cyclic successor relation, where the successor of the last element is the first element.
It is straightforward to encode the cyclic successor relation $CySucc/2$ with the immediate successor relation $Succ_1$:
\begin{equation}
  \begin{aligned}
    &(\forall x: First(x) \leftrightarrow \forall y: (x \leq y))\\
    \land\;&(\forall x: Last(x) \leftrightarrow \forall y: (y \leq x))\\
    \land\;&(\forall x\forall y: CySucc(x,y) \leftrightarrow (Succ_1(x,y)\lor (Last(x)\land First(y)))).
  \end{aligned}
  \label{eq:cyclic_pred}
\end{equation}
\rev{
Note, however, that simply appending \Cref{eq:cyclic_pred} to the input sentence with $p$ valid cells would result in $4p$ valid cells of the new sentence, as we have already pointed out in \cref{remark:n^cp}.
Nevertheless, let us revisit the idea of defining the $First$ and $Last$ relations using unary evidence.
Specifically, for the fixed ordering $1\le2\le3\le\cdots\le n$ enforced by the axiom $\lonat(\le)$, we do not need to define $First$ and $Last$ in terms of the linear order relation $\lopred$.
Instead, we can append unary evidence $\epsilon$ such that $$\epsilon =First(1) \land \neg Last(1) \land \neg First(n) \land Last(n) \land \bigwedge_{i=2}^{n-1} \left(\neg First(i) \land \neg Last(i)\right).$$

Following the cell-filtering idea described in \cref{ssec:iwfomc-evidence}, at most $p$ cells will have to be considered at any given iteration of our algorithm.
The extra $3p$ cells created by introducing the new relations will be excluded.
Note that as long as we guarantee that the relations $First$ and $Last$ do not appear anywhere else in the sentence (a reasonable assumption given the fact, that those relations are only auxiliary), and the original sentence contained no evidence to begin with, we can use the cell-filtering approach even to compute WFOMC in the presence of the (extended) linear order axiom $\loaxiom$.

Hence, to support the cyclic successor relation $CySucc$ with only small additional overhead, we only need to append $\epsilon$ and the sentence
$$\forall x\forall y: CySucc(x,y) \leftrightarrow (Succ_1(x,y)\lor (Last(x)\land First(y)))$$
to the input theory.
}

\begin{example}
    \label{ex:ws-like}
    To demonstrate the modeling capabilities of the cyclic successor relation, consider the following theory:
    \begin{equation*}
        \begin{aligned}
            &(\forall x : \neg Edge(x, x))\\
            \land\;&(\forall x \forall y : Edge(x,y) \to Edge(y,x))\\   
            \land\;&(\forall x \forall y : CySucc(x,y) \to Edge(x,y))\\
            \land\;&(|Edge|=2\cdot(n+m))
        \end{aligned}
    \end{equation*}

    The first two sentences require $G(Edge)$ to be an undirected graph without loops.
    The cardinality constraint further requires that there are $n+m$ edges in the graph.\footnote{The factor of two comes from the fact that the relation $Edge/2$ is inherently directed.}
    Finally, the sentence with $CySucc$ partially restricts the graph structure, enforcing a cyclic chain going over all vertices (i.e., the sentence dictates the use of $n$ afforded edges).
    Hence, we are modeling a cyclic chain with $m$ additional edges chosen at random. Our techniques enable us to perform exact lifted inference over such a structure.
    The graph is shown in \cref{fig:ws-simplified} for $n=6$ and $m=3$ with red dashed lines showing possible random connections.

    The model at hand can be viewed as a simplified version of the random graph model of \citet{watts98:ws-model}. See \cref{sec:7-experiments} for details.

    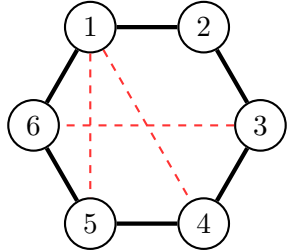
\begin{figure}[hbt]
        \centering
\begin{tikzpicture}[
    % Node Style
    node_style/.style={circle, draw=black, thick, fill=white}, 
    % Edge Styles
    cycle_edge/.style={draw=black, ultra thick}, 
    chord_edge/.style={draw=red!80, thick, dashed}, 
]

    % --- PARAMETERS ---
    \def\n{6}
    \def\radius{1.5cm}

    % --- 1. DRAW NODES ---
    % Start at 120 degrees so Node 1 is Top-Left and Node 2 is Top-Right.
    % This maintains the compact "flat" top/bottom edges.
    \foreach \i in {1,...,\n} {
        \node[node_style] (n\i) at ({120 - 360/\n * (\i - 1)}:\radius) {\i};
    }

    % --- 2. DRAW THE CYCLE (Solid) ---
    \foreach \i [evaluate=\i as \next using {int(mod(\i,\n)+1)}] in {1,...,\n} {
        \draw[cycle_edge] (n\i) -- (n\next);
    }

    % --- 3. DRAW SHORTCUTS (Dashed) ---
    \draw[chord_edge] (n1) -- (n4);
    \draw[chord_edge] (n1) -- (n5);
    \draw[chord_edge] (n3) -- (n6);

\end{tikzpicture}
        
        \caption{An example of the simplified Watts-Strogatz model on 6 nodes with $K=2$ and 3 additional (rewired) edges.}
        \label{fig:ws-simplified}
\Description{The figure shows an undirected graph on six vertices with a single circuit going over all the nodes.
Three additional dashed edges are present, specifically \{1, 4\}, \{1, 5\}, \{3, 6\}.}
        
    \end{figure}

\end{example}

\section{Hardness of Two Linear Order Axioms}
\label{sec:5-two-los}
Let us now extend our analysis beyond a single axiom, such as the linear order we have seen so far.
Can we perform domain-lifted inference over sentences with two linear order axioms?
Specifically, can we work with two distinct domain orderings simultaneously while maintaining tractability?
The question could be generalized to ask whether we may perform domain-lifted inference over two or more axioms (not necessarily linear orders).
The work of \citet{kuangetal24:wfomc-polynomials} has studied such a setting for scenarios where the axioms are all on a single binary relation (e.g., a relation $R$ is both acyclic and connected).
However, our original question about two linear orders was posed in a way that required axioms on two distinguished binary relations.
Hence, the techniques used therein are not applicable here.

In this section, we show that \wfomc over sentences with two linear order relations (i.e., with the linear order axiom on two distinguished binary relations) is, in fact, a \class{\#P_1}-hard problem and, thus, (likely) not domain-liftable.
Hence, the task is also not tractable for any axioms \emph{more general} than the linear order, such as the acyclicity axiom, since a linear order is an acyclic structure with further restrictions.
\rev{
Note that for these theoretical results, we will focus only on the arbitrary linear ordering enforced by the axiom $\loaxiom$, which is a more standard construct used in mathematical logic.
}

Similar to the proofs of undecidability results of \citet{otto01:unsat-fo2+8order,kieronski11-unsat-fo2+lo}, we prove the hardness by reducing WFOMC of a particular sentence to a hard tiling problem.
We first define a tiling problem whose number of valid tilings is \class{\#P_1}-hard to compute.
Then, we introduce an intermediate axiom, the \emph{grid axiom}, which requires that domain elements form a grid, enabling two distinguished binary relations $H$ and $V$ to represent the horizontal and vertical successor relations, respectively.
We show that WFOMC of \fotwo sentences with a grid axiom is \class{\#P_1}-hard by encoding the \class{\#P_1}-hard tiling problem using the grid.
Subsequently, we show that the grid axiom can be implemented by two linear order relations, which produces the hardness result for two linear orders.

\subsection{The Tiling Problem}

The hardness comes from a variant of the tiling problem.
A \emph{tiling system} is a tuple $(\mathcal{T}, R_H, R_V)$ where $\mathcal{T}$ is a set of tiles and $R_H, R_V \subseteq \mathcal{T} \times \mathcal{T}$ are the horizontal and vertical constraints.
A \emph{tiling} of $(\mathcal{T}, R_H, R_V)$ on a grid is a mapping $T: \nat \times \nat \to \mathcal{T}$ such that
\begin{itemize}
  % \item $(T(i,j), T(i,j+1)) \in R_H$ for each $1 \le i \le n$ and $1 \le j < m$, and
  % \item $(T(i,j), T(i+1,j)) \in R_V$ for each $1 \le i < n$ and $1 \le j \le m$.
  \item $(T(i,j), T(i,j+1)) \in R_H$ for each $i\in\nat$ and $j\in\nat$, and
  \item $(T(i,j), T(i+1,j)) \in R_V$ for each $i\in\nat$ and $j\in\nat$.
  
\end{itemize}
Intuitively, the problem is to assign a set of tiles, e.g., we can think of a tile as a square with colors on its four edges that we cannot rotate, which are the classical Wang tiles \citep{wang61:-tiles}, to a grid (a first quadrant) under the constraint that only specified pairs of tiles can be horizontally or vertically adjacent.
Refer to \Cref{fig:tiling_problem} for a sketch of the problem.

\begin{figure}[t]
    \centering
    \begin{tikzpicture}
        
        % Define styles for the tiles
        \tikzstyle{placed} = [draw=black, thick, fill=blue!15]
        \tikzstyle{unplaced} = [draw=gray, thick, dashed, fill=gray!5]

        % 1. Draw the "unplaced" (dashed/gray) tiles first so they sit behind
        \foreach \x/\y in {4/0, 5/0, 3/1, 4/1, 5/1, 2/2, 3/2, 4/2, 1/3, 2/3, 3/3, 0/4, 1/4, 2/4} {
            \path[unplaced] (\x,\y) rectangle (\x+1,\y+1);
        }

        % 2. Draw the "placed" (solid) tiles originating from the bottom left
        \foreach \x/\y in {0/0, 1/0, 2/0, 3/0, 0/1, 1/1, 2/1, 0/2, 1/2, 0/3} {
            \path[placed] (\x,\y) rectangle (\x+1,\y+1);
        }

        % 3. Label specific tiles i, j, k
        % We choose i=(1,1), j=(2,1), k=(1,2)
        \node at (1.5, 1.5) {\Large $a$};
        \node at (2.5, 1.5) {\Large $b$};
        \node at (1.5, 2.5) {\Large $c$};

        % Draw subtle arrows to show the adjacency relations
        \draw[->, thick, black!70, >=Stealth] (1.7, 1.5) -- (2.3, 1.5);
        \draw[->, thick, black!70, >=Stealth] (1.5, 1.7) -- (1.5, 2.3);

        % 4. Draw the Axes (1st Quadrant)
        \draw[->, ultra thick, >=Stealth] (0,0) -- (7.5,0) node[right] {$x$};
        \draw[->, ultra thick, >=Stealth] (0,0) -- (0,6.5) node[above] {$y$};

        % 5. Add ellipses to indicate infinity
        % Horizontal infinity
        \node at (6.5, 0.5) {\Large $\dots$};
        \node at (6.5, 1.5) {\Large $\dots$};
        
        % Vertical infinity
        \node at (0.5, 5.5) {\Large $\vdots$};
        \node at (1.5, 5.5) {\Large $\vdots$};

        % Diagonal infinity
        % \node at (5.5, 3.5) {\Large $\ddots$};
        % \node at (4.5, 4.5) {\Large $\ddots$};

        % 6. Add the Legend
        \node[draw=black, thick, fill=white, align=left] at (5.5, 5.5) {
            $(a,b) \in R_H$ \\[1ex]
            $(a,c) \in R_V$
        };

    \end{tikzpicture}
    \caption{A sketch of the tiling problem. Solid tiles represent the area already successfully tiled, while dashed tiles represent the expanding frontier. Adjacent placed tiles $a, b$, and $c$ must satisfy the horizontal and vertical matching relations, $R_H$ and $R_V$. The grid extends infinitely along the positive $x$ and $y$ axes.}
    \Description{A 2D plot showing the first quadrant of a Cartesian coordinate system. Solid blue square tiles are packed into the bottom-left corner next to the origin. Three specific solid tiles are labeled a, b, and c, where b is to the right of a, and c is above a. Arrows point from a to b, and from a to c. A legend box in the top right states T(a,b) in R_H and T(a,c) in R_V. Adjacent to the solid tiles is a frontier of dashed gray square tiles. Ellipsis dots indicate that the grid extends infinitely to the right and upwards.}
    \label{fig:tiling_problem}
\end{figure}
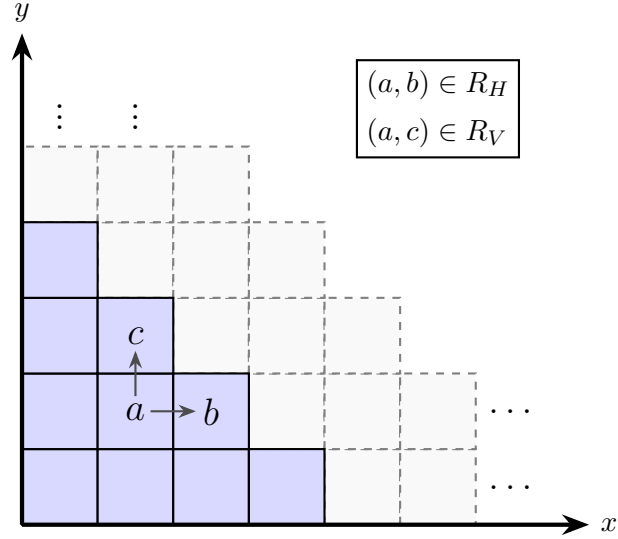

Tiling problems tend to be hard; in fact, the general tiling problem presented above is undecidable.
That can be proved by reducing the problem of whether, given a TM $M$, \emph{$M$ halts on the empty input} to an instance of the tiling problem \citep{lewis98:tiling-npc}.%
\footnote{Halting on empty input is a special case of the well-known halting problem whose undecidability was proved already by \citet{turing36:tm}.
Any instance of the halting problem can be reduced to halting on empty input by ``hard-coding'' the input, i.e., we construct a new TM which first writes the original input onto the empty tape and then behaves as the original machine.}
The reduction essentially boils down to having each row of the tiling represent the tape of $M$ at a given time step; the horizontal dimension represents the tape, and the vertical dimension represents time.
The first row of the tiling represents the empty tape, and each successive row then represents the state of the tape after $M$ takes one applicable transition.
If the TM never halts on the empty input, then we may tile rows to infinity.
If $M$ halts after $k$ steps, then we may tile at most $k$ rows.

To avoid the undecidability of the tiling problem, we may define a less ambitious problem, specifically the \emph{bounded tiling problem}.
That problem inputs the tiling system $(\mathcal{T}, R_H, R_V)$, the size of the grid $n\times n$, and the first row of the tiling $t_1, \cdots, t_n$, and asks if there is a tiling on the grid consistent with the given first row.
It was proved in \citet{lewis98:tiling-npc} that the problem is \class{NP}-complete as it encodes the accepting paths of a polynomial-time nondeterministic Turing machine.
Moreover, the properties of the encoding can be summarized in the following lemma.

\begin{lemma}{(Summarization of the Proof of \citet[Theorem 7.2.1]{lewis98:tiling-npc})}
\label{lemma:bounded-tiling}
For any one-tape nondeterministic Turing machine, there is a tiling system $(\mathcal{T}, R_H, R_V)$ such that any input $x$ of the Turing machine can be transformed to a tile sequence $t_1, \ldots, t_{q+2} \in \mathcal{T}$ where $q$ is the time bound for running $x$, and each accepting path of the Turing machine on input $x$ uniquely corresponds to a tiling with the first row $t_1, \ldots, t_q$ and vice versa.
In particular, the tiles $t_3, \ldots, t_{|x|+2}$ represent the input $x$, and $t_{|x|+3}, \ldots, t_q$ are the same tile representing the default empty symbol on the tape.
\end{lemma}

Let us now take inspiration from the bounded tiling problem and define our own counting version of a tiling problem, which will be useful in our setting of \wfomc.
Similarly to the bounded tiling, we will use a bounded grid, and the first row will be restricted to a specific format.
Note that the task will be to count all possible tilings, rather than determine whether a tiling exists.

\begin{definition}[1-1-N-M Counting Tiling Problem]
Given a tiling system $(\mathcal{T}, R_H, R_V)$ and four distinguished tiles $t_1, t_2, t_3, t_4 \in \mathcal{T}$, the \emph{1-1-N-M counting tiling problem} $\tiling{\mathcal{T}, R_H, R_V, t_1, t_2, t_3, t_4}$ inputs two positive integers $n,m$ in unary and asks to compute the number of tilings of $(\mathcal{T}, R_H, R_V)$ on a grid of $(n+m+2)$ rows and $(n+m+2)$ columns where the first row consists of one tile of $t_1$, one tile of $t_2$, $n$ tiles of $t_3$ and $m$ tiles of $t_4$ in order.
\end{definition}

\begin{remark}
Since the 1-1-N-M counting tiling problem takes $n$ and $m$ as input in unary, we measure its time complexity in terms of $n+m$.
\end{remark}

\begin{lemma}
\label{lemma:counting-tiling}
There is a \class{\#P_1}-hard 1-1-N-M counting tiling problem $\tiling{\mathcal{T}, R_H, R_V, t_1, t_2, t_3, t_4}$.
\end{lemma}

%\begin{proof}[Proof Sketch]
The lemma is proved mainly by encoding the Turing machine in \Cref{lemma:utm} in the same way as in \Cref{lemma:bounded-tiling} so that each tiling of certain $n,m$ corresponds to an accepting path of the Turing machine on the input $n$ with time and space bound of $n+m+2$.
Recall that the TM in \Cref{lemma:utm} is a multi-tape linear time nondeterministic Turing machine with unary alphabet and computing its number of accepting paths on a given input is \class{\#P_1}-hard.
%\end{proof}

\begin{proof}
Let $M$ be the Turing machine in \Cref{lemma:utm}.
Let us first transform $M$ to a single-tape TM $M'$ which preserves all other properties of $M$ (i.e., nondeterminism, unary input alphabet, and \class{\#P_1}-hardness of computing the number of accepting paths).
The transformation can be accomplished using well-known tricks, as described, e.g., in \citet[Proof of Theorem 1.6]{arora09:computational-complexity}.
Most notably, we \emph{interleave} the tapes into one, i.e., we store the first symbol from each tape, then the second symbol, and so on.
$M'$ will only require a constant number of passes over its tape to simulate one step of $M$, and its time complexity will be the square of the complexity of $M$.
Hence, we may find a constant $c$ such that $M'$ always terminates within $cn^2$ steps for any input of size $n$, and we can bound the size of its tape by the same value.

Since $M'$ is a one-tape nondeterministic TM, by \cref{lemma:bounded-tiling}, we may encode it as a bounded tiling problem.
Let $(\mathcal{T}, R_H, R_V)$ be the tiling system to encode $M'$.
Moreover, as the input of $M'$ is in unary, we use the same encoding to obtain the first row $t_1, \ldots, t_{cn^2+2}$ so that $t_3 = \ldots = t_{n+2}$ is the tile representing the input symbol (i.e., the only element of the unary input alphabet), and $t_{n+3} = \ldots = t_{cn^2+2}$ is the tile representing the default empty symbol.
By the same argument as in \citet[Theorem 7.2.1]{lewis98:tiling-npc}, each tiling on the $(cn^2+2) \times (cn^2+2)$ grid with the specified first row corresponds to an accepting path of $M'$ on the input size $n$.
Therefore, $\tiling{\mathcal{T}, R_H, R_V, t_1, t_2, t_3, t_{n+3}}$ on input $n$ and $m = cn^2-n$ equals the number of accepting paths of $M'$ on the input size $n$.
\end{proof}

\subsection{The Grid Axiom}
\label{ssec:grid-axiom}
\cref{lemma:counting-tiling} states that there are hard 1-1-N-M counting tiling problems, specifically \class{\#P_1}-hard.
Let us now define an intermediate axiom, namely the \emph{grid axiom}, which will be able to encode such a hard counting problem.
The grid axiom requires that the elements of the domain be arranged into a grid, and that the horizontal and vertical successor relations be accessed via the binary relations $H$ and $V$, respectively.

\begin{definition}{(The Grid Axiom)}
\label{def:grid-axiom}
The \emph{grid axiom} over a domain of size $n^2$, denoted by $\gridaxiom(H,V)$, involves two binary relations $H$ and $V$ such that there is a way to arrange the $n^2$ elements into a grid of $n$ rows and $n$ columns, and
\begin{itemize}
  \item $H(x,y)$ is true if and only if $x$ is at the $i$-th row and $j$-th column and $y$ is at the $i$-th row and $j+1$-th column, for some $1 \le i \le n$ and $1 \le j < n$, and
  \item $V(x,y)$ is true if and only if $x$ is at the $i$-th row and $j$-th column and $y$ is at the $i+1$-th row and $j$-th column, for some $1 \le i < n$ and $1 \le j \le n$.
\end{itemize}
\end{definition}

With the grid axiom, we can encode the 1-1-N-M counting tiling problem in \Cref{lemma:counting-tiling} such that each valid tiling corresponds to the same number of models of the sentence (specifically, given one tiling, the number of models is equal to the factorial of the grid size since we may reorder the domain elements without changing the grid structure), and thus computing the number of models of the sentence is as hard as computing the 1-1-N-M counting tiling problem.

\begin{lemma}
\label{lemma:grid-tm}
There is an \fotwo sentence with a grid axiom whose WFOMC is \class{\#P_1}-hard.
\end{lemma}

The proof will assume that the number of distinct tiles (i.e., the size of the set $\mathcal{T}$) is bounded a constant as we will introduce a new predicate $T_i$ for each unique tile $t_i\in\mathcal{T}$.
Afterwards, the proof consists only of a technical construction of a first-order theory that properly links the grid axiom and a tiling.
See \Cref{fig:grid_axiom_mapping} for a sketch.

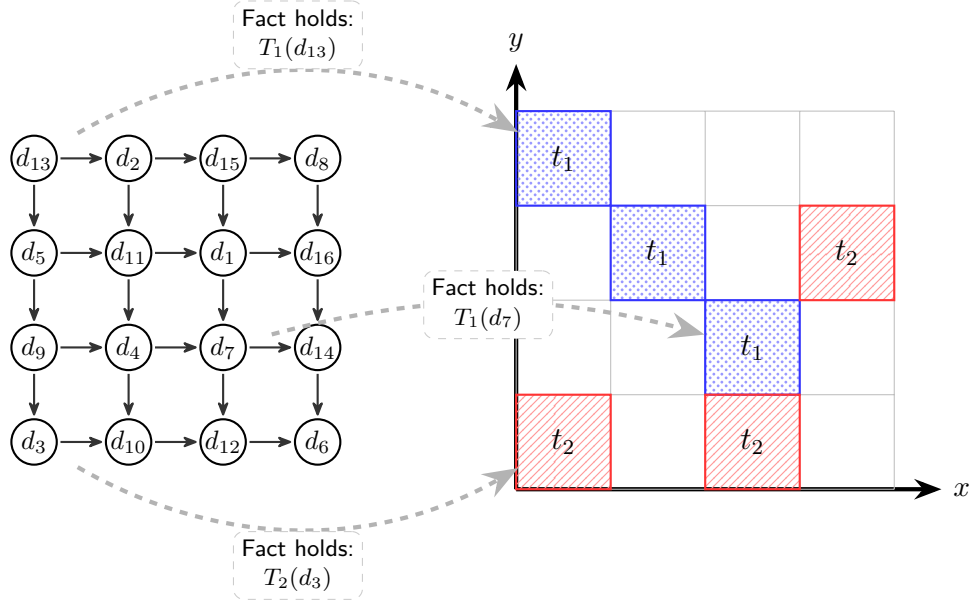
\begin{figure}[t]
    \centering
    \begin{tikzpicture}[
        % --- Styles for Left Side (Logical Domain) ---
        d_node/.style={circle, draw=black, thick, fill=white, minimum size=0.6cm, inner sep=0pt, font=\small},
        h_edge/.style={->, >={Stealth[round]}, draw=black!80, thick, shorten >=1pt, shorten <=1pt},
        v_edge/.style={->, >={Stealth[round]}, draw=black!80, thick, shorten >=1pt, shorten <=1pt},
        % --- Styles for Right Side (Tiling Grid) ---
        tile_grid_line/.style={draw=gray!50, thin},
        fill_T1/.style={draw=blue!80, thick, fill=blue!10, pattern=crosshatch dots, pattern color=blue!40},
        fill_T2/.style={draw=red!80, thick, fill=red!10, pattern=north east lines, pattern color=red!40},
        % --- Styles for Mapping Arrows ---
        map_arrow/.style={->, >={Stealth[scale=1.3]}, draw=gray!60, ultra thick, dashed, shorten >=10pt, shorten <=10pt},
        map_label/.style={fill=white, font=\footnotesize\sffamily, draw=gray!40, thin, rounded corners, inner sep=3pt, align=center}
    ]

    % Grid step size to space nodes and tiles without using global scaling
    \def\gstep{1.25} 

    % ================= LEFT SIDE: LOGICAL DOMAIN =================
    \begin{scope}[local bounding box=logicDomain]
        % 1. Draw Nodes with scrambled labels
        % Shifted by +0.5*\gstep so the node centers perfectly align with the tile centers on the right.
        \foreach \x/\y/\lbl in {
            0/3/d_{13}, 1/3/d_{2},  2/3/d_{15}, 3/3/d_{8},
            0/2/d_{5},  1/2/d_{11}, 2/2/d_{1},  3/2/d_{16},
            0/1/d_{9},  1/1/d_{4},  2/1/d_{7},  3/1/d_{14},
            0/0/d_{3},  1/0/d_{10}, 2/0/d_{12}, 3/0/d_{6}
        } {
            \node[d_node] (D-\x-\y) at (\x*\gstep + 0.5*\gstep, \y*\gstep + 0.5*\gstep) {$\lbl$};
        }

        % 2. Draw H relations (Left-to-Right)
        \foreach \y in {0,1,2,3} {
            \foreach \x [count=\nextx from 1] in {0,1,2} {
                \draw[h_edge] (D-\x-\y) -- (D-\nextx-\y);
            }
        }
        % 3. Draw V relations (Top-to-Bottom / Downwards)
        \foreach \x in {0,1,2,3} {
            \foreach \y/\prevy in {1/0, 2/1, 3/2} {
                \draw[v_edge] (D-\x-\y) -- (D-\x-\prevy);
            }
        }
    \end{scope}
    
    % ================= RIGHT SIDE: TILING GRID =================
    \begin{scope}[xshift=7cm, local bounding box=tilingGrid]
        % 1. Draw Axes (1st Quadrant)
        \draw[->, ultra thick, >=Stealth] (0,0) -- (4.5*\gstep,0) node[right] {$x$};
        \draw[->, ultra thick, >=Stealth] (0,0) -- (0,4.5*\gstep) node[above] {$y$};

        % 2. Draw faint background grid (4x4 squares)
        \draw[tile_grid_line, step=\gstep] (0, 0) grid (4*\gstep, 4*\gstep);

        % 3. Place specific tiles perfectly aligned with grid squares
        % Target corresponding to d_13 (top-left)
        \path[fill_T1] (0*\gstep, 3*\gstep) rectangle (1*\gstep, 4*\gstep);
        \node[font=\bfseries\large] (T-0-3) at (0.5*\gstep, 3.5*\gstep) {$t_1$};
        
        % Target corresponding to d_7 (middle-right)
        \path[fill_T1] (2*\gstep, 1*\gstep) rectangle (3*\gstep, 2*\gstep);
        \node[font=\bfseries\large] (T-2-1) at (2.5*\gstep, 1.5*\gstep) {$t_1$};
        
        % Target corresponding to d_3 (bottom-left)
        \path[fill_T2] (0*\gstep, 0*\gstep) rectangle (1*\gstep, 1*\gstep);
        \node[font=\bfseries\large] (T-0-0) at (0.5*\gstep, 0.5*\gstep) {$t_2$};

        % Place a few context tiles for flavor (where the old mappings used to be)
        \path[fill_T1] (1*\gstep, 2*\gstep) rectangle (2*\gstep, 3*\gstep);
        \node[font=\bfseries\large] at (1.5*\gstep, 2.5*\gstep) {$t_1$};
        
        \path[fill_T2] (2*\gstep, 0*\gstep) rectangle (3*\gstep, 1*\gstep);
        \node[font=\bfseries\large] at (2.5*\gstep, 0.5*\gstep) {$t_2$};
        
        \path[fill_T2] (3*\gstep, 2*\gstep) rectangle (4*\gstep, 3*\gstep);
        \node[font=\bfseries\large] at (3.5*\gstep, 2.5*\gstep) {$t_2$};
    \end{scope}

    % ================= MAPPING ARROWS =================
    % Connect specific domain elements to their corresponding tiles
    % Arching D-0-3 to T-0-3 over the top of the grids
    \draw[map_arrow] (D-0-3) to[out=30, in=150] node[map_label, midway, above=2pt] {Fact holds:\\$T_1(d_{13})$} (T-0-3);
    
    % Mid-grid mapping for D-2-1 to T-2-1
    \draw[map_arrow] (D-2-1) to[out=15, in=165] node[map_label, midway] {Fact holds:\\$T_1(d_{7})$} (T-2-1);
    
    % Arching D-0-0 to T-0-0 under the bottom of the grids
    \draw[map_arrow] (D-0-0) to[out=-30, in=-150] node[map_label, midway, below=2pt] {Fact holds:\\$T_2(d_{3})$} (T-0-0);

    \end{tikzpicture}
    
    \caption{A sketch illustrating the relationship between a logical model satisfying the grid axiom and a valid tiling. On the left, domain elements $d_i$ are arranged into a logical structure by the $H$ and $V$ relations (e.g., the facts $H(d_{3}, d_{10})$ and $V(d_{11},d_4)$ hold). On the right, the geometric interpretation shows tiles occupying grid squares corresponding to those domain elements. The dashed arrows demonstrate how facts holding in the logical domain dictate the placement of specific tile types in the grid.}
    
    \Description{A diagram with two side-by-side structures illustrating the mapping between a logical model and a geometric tiling. The left side shows 16 circular nodes in a 4x4 layout, representing domain elements with arbitrarily scrambled indices. They are connected by horizontal arrows pointing rightward (the H relation) and vertical arrows pointing downward (the V relation). The right side shows a geometric Cartesian grid with x and y axes. Certain 1x1 squares in this grid are filled with tiles: t_1 (blue dotted) and t_2 (red striped). Three thick dashed arrows map specific nodes from the logical side to specific tiles on the geometric side: node d_13 maps to a top-left t_1 tile with an arrow arching over the top, node d_7 maps to another t_1 tile in the middle, and node d_3 maps to a bottom-left t_2 tile with an arrow arching underneath.}
    \label{fig:grid_axiom_mapping}
\end{figure}

\begin{proof}
Let $\tiling{\mathcal{T}, R_H, R_V, t_1, t_2, t_3, t_4}$ be an instance of the \class{\#P_1}-hard 1-1-N-M counting tiling problem presented in \Cref{lemma:counting-tiling}.
We now encode $\tiling{\mathcal{T}, R_H, R_V, t_1, t_2, t_3, t_4}$ for the input $n,m$ by an \fotwo sentence with a grid axiom $\gridaxiom(H,V)$ over a domain of size $(n+m+2)^2$.
Elements are expected to be arranged into $n+m+2$ rows and $n+m+2$ columns, where the element at the $i$-th row and the $j$-th column represents the cell at the $i$-th row and the $j$-th column of the tiling grid.
Let $\mathcal{T}=\{t_1, \ldots, t_k\}$.
We introduce fresh unary relations $T_1(x), \ldots, T_k(x)$ where $T_i(x)$ being true indicates that the cell $x$ is mapped to the tile $t_i$.

Formally, we define the following sentences:
\begin{itemize}
  \item The top and left borders of the grid are identified by unary relations $Top$ and $Left$ according to $H$ and $V$.
  \begin{equation}
  \label{eq:tm1}
    \begin{aligned}
      &(\forall x : Top(x) \leftrightarrow \forall y: \lnot V(y,x)) \\
      \land \; &(\forall x : Left(x) \leftrightarrow \forall y: \lnot H(y,x)).
    \end{aligned}
  \end{equation}

  \item Each grid cell is mapped to a unique tile.
  \begin{equation}
  \label{eq:tm2}
    \forall x : \bigvee_{i=1}^k \left( T_i(x) \land \bigwedge_{j \neq i} \lnot T_j(x) \right).
  \end{equation}

  \item The tiling satisfies the constraints $R_H, R_V$.
  \begin{equation}
  \label{eq:tm3}
    \begin{aligned}
        &\left(\forall x \forall y : H(x,y) \to \bigvee_{(t_i, t_j \in R_H)} (T_i(x) \land T_j(y))\right)  \\
        \land \; & \left(\forall x \forall y : V(x,y) \to \bigvee_{(t_i, t_j \in R_V)} (T_i(x) \land T_j(y))\right).
    \end{aligned}
  \end{equation}
  
  \item The first row consists of one tile of $t_1$, one tile of $t_2$, $n$ tiles of $t_3$ and $m$ tiles of $t_4$ in order.
  \begin{equation}
  \label{eq:tm4}
    \begin{aligned}
        & (\forall x : R_1T_1(x) \leftrightarrow (Top(x) \land Left(x))) \\
        \land \;& (\forall x : R_1T_2(x) \leftrightarrow \exists y: (H(y,x) \land R_1T_1(y))) \\
        \land \;& (\forall x : R_1T_3(x) \leftrightarrow \exists y: \bigl(H(y,x) \land (R_1T_2(y) \lor R_1T_3(y))\bigl)) \\
        \land \;& (\forall x : R_1T_4(x) \leftrightarrow \exists y: \bigl(H(y,x) \land (R_1T_3(y) \lor R_1T_4(y))\bigl)) \\
        \land \;& (|R_1T_3| = n) \\
        \land \;& (|R_1T_4| = m) \\
        \land \;& \bigwedge_{i=1}^4 (\forall x : R_1T_i(x) \to T_i(x)),
    \end{aligned}
  \end{equation}
  where $R_1T_i$ ($i = 1,2,3,4$) are auxiliary fresh unary relations.
\end{itemize}

Let $\Psi$ be the conjunction of \Cref{eq:tm1,eq:tm2,eq:tm3,eq:tm4}.
Each valid tiling of $\tiling{\mathcal{T}, R_H, R_V, t_1, t_2, t_3, t_4}$ is encoded by $((n+m+2)^2)!$ models of $\Psi \land \gridaxiom(H,V)$ over a domain of size $(n+m+2)^2$ due to the $((n+m+2)^2)!$ ways to organize a grid of $(n+m+2)^2$ elements.
Consequently, any algorithm computing FOMC of $\Psi \land \gridaxiom(H, V)$ in time polynomial in the domain size yields an algorithm to compute the number of valid tilings of $\tiling{\mathcal{T}, R_H, R_V, t_1, t_2, t_3, t_4}$ in time polynomial in $n+m$, thus the FOMC is \class{\#P_1}-hard.

Note that $\Psi$ is an \fotwo sentence with cardinality constraints.
By \Cref{lemma:c2+cc}, there is another \fotwo sentence $\Psi'$ without cardinality constraints such that WFOMC of $\Psi'$ with a grid axiom is \class{\#P_1}-hard as well.
\end{proof}

While proving the hardness of the grid axiom is a stepping stone for our primary negative result, its own significance should not be disregarded.
Unfortunately, tractable inference over grids is likely not possible.
That may seem particularly discouraging when thinking about various probabilistic inference tasks on grids, e.g., the aforementioned 2-dimensional Ising model with constant interaction strength, which are equivalent to WFOMC~\citep{broecketal21:lifting-intro}.
Nevertheless, all hope is not lost due to our results in \cref{sec:4-successor-relations} which suggest that we may use the $k$-th successor relation for grids of size $k\times n$ where $k$ is a constant with respect to $n$.
See \cref{ex:grid-kxn} for details.

\subsection{Implementing the Grid Axiom by Two Linear Order Axioms}
Next, we finally turn to the two-variable fragment with two linear order axioms.
We show that having the linear order axiom on two distinguished binary relations allows us to encode the grid axiom.
Hence, we prove that WFOMC with two linear orders is at least as hard as WFOMC with the grid axiom, i.e., \class{\#P_1}-hard.

Let us start by discussing why similar grid encodings available in the literature on finite satisfiability with several linear orders \citep{otto01:unsat-fo2+8order,kieronski11-unsat-fo2+lo} are not applicable in model counting. 
Those encodings aim to ensure that every model of the constructed sentence satisfies a property called \emph{grid-like} defined in \citet{otto01:unsat-fo2+8order}.
A model is grid-like if a grid can be homomorphically embedded into it, i.e., there is a mapping $G$ from the grid cells to domain elements along with two binary relations $H, V$ such that the horizontal (resp. vertical) successor relation of any pair of cells $x,y$ implies $H(G(x), G(y))$ (resp. $V(G(x),G(y))$.
Using the relations $H$ and $V$, one can encode the tiling problem using a similar technique as in \cref{ssec:grid-axiom}.
However, $G$ is neither required to be injective nor surjective, and the encoding makes no restriction on unmapped elements or elements that are not horizontally or vertically adjacent.
This results in non-unified numbers of ways to extract a grid from a model and allows for the same grid to be extracted from different models.
For example, one can split the elements into four groups, each encoding a grid.
Additionally, given a model encoding a valid tiling, one can add more true ground literals for $H$ and $V$.
Therefore, computing FOMC or WFOMC of such sentences does not make sense when determining the number of valid tilings.

That is why we give a new encoding of a true grid defined by \Cref{def:grid-axiom} using two linear order relations such that each grid exactly corresponds to $(n^2)!$ models and thus the number of accepting paths of the Turing machine encoded by the grid can be counted precisely.
\Cref{fig:L1L2} shows the target shape of our grid with two linear orders.
Our encoding benefits from the tractability of emulating arbitrary counting quantifiers and cardinality constraints in polynomial time with oracles for WFOMC as stated in \Cref{lemma:c2+cc}, which is not accessible in the finite satisfiability problem when we only have access to the SAT oracle for \fotwo.

\begin{lemma}
\label{lemma:grid-2lo}
Let $\Psi$ be an \fotwo sentence with distinguished binary relations $H$ and $V$, and $n \ge 2$ be an integer.
There is an \fotwo sentence $\Psi'$ with distinguished binary relations $\lopred_1$ and $\lopred_2$ such that WFOMC of $(\Psi \land \gridaxiom(H,V), n^2)$ can be reduced to WFOMC of $(\Psi' \land \loaxiom(\lopred_1) \land \loaxiom(\lopred_2), n^2)$.
\end{lemma}

\begin{proof}
Our goal is to arrange the $n^2$ elements to a grid of $n$ rows and $n$ columns using $\lopred_1$ and $\lopred_2$ so that the horizontal and vertical successor relations $H, V$ satisfying \Cref{def:grid-axiom} can be defined.
To begin, we again define three useful relations with respect to each linear order, namely the first element $First_i$, the last element $Last_i$, and the immediate successor $S_i$.%
\footnote{We use just $S$ rather than $Succ$ to refer to the successor relation here to emphasize that $S_1$ and $S_2$ are two distinct immediate successor relations defined with respect to, in order, $\le_1$ and $\le_2$.}
Hence, we reuse \Cref{eq:succ1:first-last,eq:succ1:bijection-like,eq:succ1:left2right} for each of our linear order predicates $\lopred_i$ where $i \in \{1,2\}$:
\begin{equation}
    \label{eq:2lo-succ1}
    \begin{aligned}
        & (\forall x:  First_i(x) \leftrightarrow \forall y: \ (x \lopred_i y )) \\
        \land \; &(\forall x: Last_i(x) \leftrightarrow \forall y: \ (y \lopred_i x )) \\
        \land \; &(\forall x: \lnot Last_i(x) \to \exists^{=1} y: \ S_i(x,y)) \\
        \land \; & (\forall x: \lnot First_i(x) \to \exists^{=1} y: \ S_i(y,x)) \\
        \land \; & (\forall x \forall y: (S_i(x,y) \to (x <_i y))
    \end{aligned}
\end{equation}

\paragraph{Building the Skeleton using $\lopred_1$}
We now build a basic skeleton of the grid by making constraints on the $\lopred_1$ predicate.
We mark the elements at the left border and the right border by unary relations $Left$ and $Right$ with the following constraints:
\begin{itemize}
  \item Each border contains $n$ elements, and the two borders are disjoint.
  \begin{equation}
  \label{eq:2lo-LeftRight1}
    \begin{aligned}
        & (\forall x : \lnot Left(x) \lor \lnot Right(x)) \\
        \land\; &(|Left| = n) \\
        \land\; &(|Right| = n) \\
    \end{aligned}
  \end{equation}
  
  \item The first element of $\lopred_1$ is at the left border, and the last element is at the right border.
  \begin{equation}
  \label{eq:2lo-LeftRight2}
  \begin{aligned}
    & (\forall x : First_1(x) \to Left(x)) \\
    \land \; &(\forall x: Last_1(x) \to Right(x)).  
  \end{aligned}
  \end{equation}
  
  \item $Right$ and $Left$ should be adjacent except for $First$ and $Last$.
  \begin{equation}
  \label{eq:2lo-LeftRight3}
    \forall x \forall y: S_1(x,y) \to (Right(x) \leftrightarrow Left(y))
  \end{equation}
\end{itemize}

Now, we have exactly $n$ $Left$-$Right$ pairs in order.
Each pair can be treated as an indicator of the start and the end of a row.
Therefore, elements are arranged in the shape shown in \Cref{fig:L1-1}, i.e., there are $n$ rows, each starting with a $Left$ mark and ending with a $Right$ mark, but each row can be of different length.

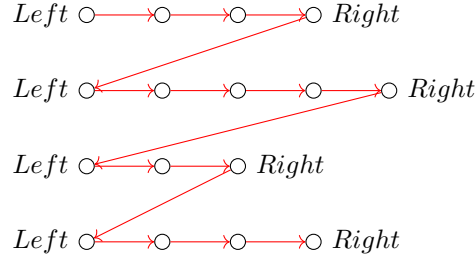
\begin{figure}
  \centering
      \begin{tikzpicture}
        \tikzstyle{roundnode}=[circle, draw, inner sep=0pt, minimum size=2mm]
        \tikzstyle{roundnodeleft}=[circle, draw, inner sep=0pt, minimum size=2mm, label=left:{\small $Left$}]
        \tikzstyle{roundnoderight}=[circle, draw, inner sep=0pt, minimum size=2mm, label=right:{\small $Right$}]
        \foreach \j in {0,3} {
            \node[roundnodeleft] (a0\j) at (0,\j) {};
            \foreach \i in {1,2} \node[roundnode] (a\i\j) at (\i,\j) {};
            \node[roundnoderight] (a3\j) at (3,\j) {};
            \foreach \i in {0,1,2}
                \pgfmathtruncatemacro{\k}{\i+1}
                \path[->,red] (a\i\j) edge (a\k\j);
        }
        \node[roundnodeleft] (a02) at (0,2) {};
        \node[roundnoderight] (a42) at (4,2) {};
        \foreach \i in {1,2,3} \node[roundnode] (a\i2) at (\i,2) {};
        \foreach \i in {0,1,2,3}
            \pgfmathtruncatemacro{\k}{\i+1}
            \path[->,red] (a\i2) edge (a\k2);
        \node[roundnodeleft] (a01) at (0,1) {};
        \node[roundnode] (a11) at (1,1) {};
        \node[roundnoderight] (a21) at (2,1) {};
        \path[->,red] (a01) edge (a11);
        \path[->,red] (a11) edge (a21);
        \path[->,red] (a33) edge (a02);
        \path[->,red] (a42) edge (a01);
        \path[->,red] (a21) edge (a00);
      \end{tikzpicture}
  \caption{A possible skeleton after enforcing constraints in \Cref{eq:2lo-LeftRight1,eq:2lo-LeftRight2,eq:2lo-LeftRight3}, which consists of $n$ rows, each starting with $Left$ and ending with $Right$. The red route represents the $\lopred_1$ relation.}
  \label{fig:L1-1}
\Description{The figure shows 16 domain elements represented as vertices ordered into a sequence by the immediate successor relation.
The elements are arranged into four rows by designating a column of Left nodes and a column of Right nodes.
The enforced rows have different lengths.}
\end{figure}

Now, let us mark the top and bottom nodes with unary relations $Top$ and $Bottom$.
If a node is marked with such a label, all nodes within the same row will have the same label.
However, only the row containing $First_1$ can be $Top$, and only the row containing $Last_1$ can be $Bottom$.
Therefore, only the first row is marked $Top$ and only the last row is marked $Bottom$.
We further restrict the two rows to length $n$.
\begin{equation}
\label{eq:2lo-TopBottom}
  \begin{aligned}
    & (\forall x \forall y: (S_1(x,y) \land \lnot Right(x)) %\\ & \:\:\:\:\:\:\:\:\:\:\:\:\: 
    \to \bigl((Top(x) \leftrightarrow Top(y)) \land (Bottom(x) \leftrightarrow Bottom(y))\bigl) \\
    \land\; & (\forall x: First_1(x) \leftrightarrow (Left(x) \land Top(x))) \\
    \land\; & (\forall x: Last_1(x) \leftrightarrow (Right(x) \land Bottom(x))) \\
    \land\; & (|Top| = n) \\
    \land\; & (|Bottom| = n)
  \end{aligned}
\end{equation}

\paragraph{Restricting the Shape Using $\lopred_2$}
Next, we force the skeleton to be a grid by making constraints using the other linear order relation $\lopred_2$.

\begin{itemize}
  \item Both linear orders share the same starting and ending positions.
  \begin{equation}
  \label{eq:2lo-L2-1}
    \forall x: (First_1(x) \leftrightarrow First_2(x)) \land (Last_1(x) \leftrightarrow Last_2(x))
  \end{equation}

  \item If $y$ is the $\lopred_1$-successor of $x$ and $x$ is not at the right border, $x$ should be $\lopred_2$-smaller than $y$.
  \begin{equation}
  \label{eq:2lo-L2-2}
    \forall x \forall y: (S_1(x,y) \land \lnot Right(x)) \to (x\lopred_2 y)
  \end{equation}

  \item If $y$ is the $\lopred_2$-successor of $x$ and $x$ is not at the bottom border, $x$ should be $\lopred_1$-smaller than $y$.
  \begin{equation}
  \label{eq:2lo-L2-3}
    \forall x \forall y: (S_2(x,y) \land \lnot Bottom(x)) \to (x\lopred_1 y)
  \end{equation}

  \item The $\lopred_1$-successor and the $\lopred_2$-successor of element $x$ are different.
  \begin{equation}
  \label{eq:2lo-L2-4}
    \forall x \forall y: \lnot S_1(x,y) \lor \lnot S_2(x,y)
  \end{equation}
\end{itemize}

The following three observations imply the route of $\lopred_2$ and the arrangement of elements:
\begin{enumerate}
  \item For any element $x$, the elements in the same row with $x$ but on its left should be $\lopred_2$-smaller than $x$ due to \Cref{eq:2lo-L2-2}.
  In other words, each row will be visited by $\lopred_2$ from left to right.

  \item For each $x$ not at the bottom, let $y$ be its successor in $\lopred_2$.
  $y$ cannot lie in the rows above $x$ due to \Cref{eq:2lo-L2-3}, or in the same row with $x$ because otherwise, by Observation (1), $y$ can only be the element next to $x$ at its right but that is against \Cref{eq:2lo-L2-4}.

  \item Based on the two observations above, the route of $\lopred_2$ can be partitioned into several paths, each of which only descends until reaching the bottom.
  The length of each path is at most $n$, as there are only $n$ rows, and there can be at most $n$ paths since the bottom side has a length of $n$.
  Given that the size of the domain is $n^2$, there must be exactly $n$ paths, each of length $n$.
  Therefore, $\lopred_2$ follows the blue route in \Cref{fig:L1L2} and the arrangement of elements forms a square grid.
\end{enumerate}

\begin{figure}
  \centering
      \begin{tikzpicture}
        \tikzstyle{roundnode}=[circle, draw, inner sep=0pt, minimum size=2mm]
        \foreach \i in {0,1,2,3}
            \foreach \j in {0,1,2,3}
                \node[roundnode] (a\i\j) at (\i,\j) {};
        \foreach \i in {0,1,2} {
            \foreach \j in {0,1,2,3}
                \pgfmathtruncatemacro{\k}{\i+1}
                \path[->,red] (a\i\j) edge (a\k\j);
            \pgfmathtruncatemacro{\k}{\i+1}
            \path[->,red] (a3\k) edge (a0\i);
        }
        \foreach \j in {0,1,2} {
            \foreach \i in {0,1,2,3}
                \pgfmathtruncatemacro{\k}{\j+1}
                \path[->,blue] (a\i\k) edge (a\i\j);
            \pgfmathtruncatemacro{\k}{\j+1}
            \path[->,blue] (a\j0) edge (a\k3);
        }
        \node[] (Top) at (1.5,3.5) {\small $Top$} ;
        \node[] (Bottom) at (1.5,-0.5) {\small $Bottom$} ;
        \node[] (Left) at (-0.5,1.5) {\small $Left$};
        \node[] (Right) at (3.6,1.5) {\small $Right$};
      \end{tikzpicture}
  \caption{The elements form a grid using the linear orders $\lopred_1$ (the red route) and $\lopred_2$ (the blue route).}
  \label{fig:L1L2}
\Description{The figure shows 16 domain elements represented as vertices ordered into a four-by-four grid.
The grid is defined by two immediate successor relations: one going left-to-right, as before, which defines the Left and Right columns, and the other going top-to-bottom, which defines the Top and Bottom rows.}
\end{figure}
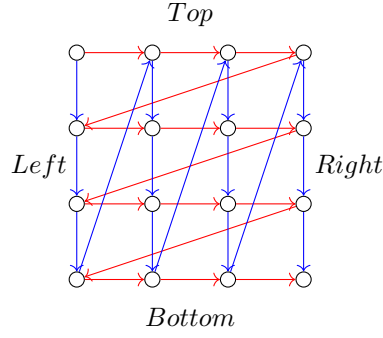

Given the grid structure of elements, the horizontal and vertical successor relations can be defined as follows.
\begin{equation}
\label{eq:2lo-get-hv}
  \begin{aligned}
    & (\forall x \forall y: H(x,y) \leftrightarrow (S_1(x,y) \land \lnot Right(x))) \\
    \land\; &(\forall x \forall y: V(x,y) \leftrightarrow (S_2(x,y) \land \lnot Bottom(x))).
  \end{aligned}
\end{equation}

One can verify that \Cref{def:grid-axiom} is satisfied by assigning coordinates to elements based on the grid in Observation (3).

\paragraph{Obtaining $\symbwfomc(\Psi, n^2)$}
Let $\Psi_l$ be the conjunction of $\Psi$ and \Cref{eq:2lo-succ1,eq:2lo-LeftRight1,eq:2lo-LeftRight2,eq:2lo-LeftRight3,eq:2lo-TopBottom,eq:2lo-L2-1,eq:2lo-L2-2,eq:2lo-L2-3,eq:2lo-L2-4,eq:2lo-get-hv}, which is a \ctwo sentence with cardinality constraints.
By the analysis above, we obtain a unique grid from each model of $\Psi_l \land \loaxiom(\lopred_1) \land \loaxiom(\lopred_2)$.
In addition, each assignment of elements into a grid uniquely corresponds to a model of such a sentence.
Therefore,
\begin{equation*}
  \begin{aligned}
    & \symbwfomc(\Psi \land \gridaxiom(H,V), n^2, \weights)
    = \symbwfomc(\Psi_l \land \loaxiom(\lopred_1) \land \loaxiom(\lopred_2), n^2, \weights)
  \end{aligned}
\end{equation*}
holds for any weighting functions $(\weights)$.

Finally, by \Cref{lemma:c2+cc}, there is a sentence $\Psi'$ in \fotwo without cardinality constraints such that the problem of computing \wfomc of $(\Psi_l \land \loaxiom(\lopred_1) \land \loaxiom(\lopred_2), n^2)$, as well as the problem of computing \wfomc of $(\Psi \land \gridaxiom(H,V), n^2)$, can be reduced to \wfomc of $(\Psi' \land \loaxiom(\lopred_1) \land \loaxiom(\lopred_2), n^2)$.
\end{proof}

We are now ready to state the main hardness result of our work.
\begin{theorem}\label{thm:2lo}
There is an \fotwo sentence with two linear order relations whose WFOMC is \class{\#P_1}-hard to compute.
\end{theorem}

\begin{proof}
It follows naturally from \Cref{lemma:grid-2lo,lemma:grid-tm}.
\end{proof}

As a consequence of our hardness result, we may also claim that any fragment with two axioms \emph{more general} than the linear order is also not tractable.
By \emph{more general} we mean any axiom that can be used to express linear order, one example being the \emph{acyclicity axiom} \citep{malhotra23:wfomc-axioms,kuangetal24:wfomc-polynomials} which we will denote $\acyclicityaxiom(R)$, and it will simply require the graph $G(R)$ to be acyclic.

As pointed out in \citet{kuangetal24:wfomc-polynomials}, a linear order relation $\lopred$ can be encoded in the following way:
We first attach an acyclicity axiom to a binary relation $R$ and require $G(R)$ to be a tournament, and then let $\lopred$ be the copy of $R$ with the additional property of reflexivity.
Formally, a linear order relation $\lopred$ can be encoded by the following sentence:
\begin{equation}
\label{eq:acyclicity-by-lo}
  \begin{aligned}
    & \acyclicityaxiom(R) \land (\forall x \forall y: (x \neq y) \to (R(x,y) \lor R(y,x)) \\
    \land\; & (\forall x \forall y: (x \neq y) \to (x \lopred y \leftrightarrow R(x,y))) \\
    \land\; & (\forall x : x\lopred x).
  \end{aligned}
\end{equation}

The equality symbol can be eliminated, e.g., by introducing a fresh binary predicate $Eq$ such that
\begin{equation*}
        (\forall x: Eq(x, x)) \land (\forall x \exists^{=1} y: Eq(x, y)).%
\footnote{Note that dealing with equality using \ctwo is a generalization of a trick presented already in \citet[Lemma 3.5]{beame15:wfomc-fo3} which was designed to eliminate the equality relation from \fotwo.}
\end{equation*}
Again, by \cref{lemma:c2+cc}, we may reduce \wfomc of a sentence conjoined with \Cref{eq:acyclicity-by-lo} to \wfomc over \fotwo.

Since we may express the linear order in \fotwo using the acyclicity axiom, we obtain another hardness result for \fotwo with two acyclic relations.
\begin{corollary}
There is an \fotwo sentence with two acyclic relations whose WFOMC is \class{\#P_1}-hard to compute.
\end{corollary}

\section{Domain-Liftability of a Linear Order and Another Successor}
\label{sec:6-lo+succ}
The hardness result from the previous section is negative.
We cannot perform domain-lifted inference over logical fragments with linear order axioms on two distinguished binary relations.
Consequently, we cannot perform domain-lifted inference over logical fragments with two axioms more general than the linear order, such as the acyclicity axiom, either.
However, this does not preclude the possibility of tractable inference with axioms on two relations.
In this section, we consider the case when we have one linear order axiom and one immediate successor relation of another linear order.
I.e., the domain is ordered in two distinct ways, and we have explicit access to the first linear order (including its successor relations). However, from the second ordering, we only see the immediate successor relation.
By devising a third polynomial-time algorithm, we prove that such a logical fragment is domain-liftable.
In other words, after the previous hardness result for two linear order axioms, we weaken the power of the second one in exchange for tractability.

As we start talking about a successor relation that is part of an unknown linear order, we reuse \Cref{eq:successor} to define yet another axiom for our syntax.

\begin{definition}[Successor Axiom]\label{def:succ}
The \emph{successor axiom}, denoted by $\succaxiom(R)$, requires that $R$ should be the successor relation of a linear order. That is, let $\lopred$ be some linear order relation, then $R$ satisfies
\begin{equation*}
  \forall x \forall y : \Bigl(R(x,y) \leftrightarrow \bigl( (x < y) \land \bigl( \lnot \exists z: (x < z) \land (z < y) \bigl) \bigl) \Bigl).
\end{equation*}
In other words, the graph $G(R)$ is the longest directed path.
\end{definition}

\subsection{Yet Another Recursion}
In developing the aforementioned algorithm, we will once again build on the domain recursion rule and the idea of computing WFOMC recursively.
Recall \Cref{eq:basic-wfomc}, which was our starting point for performing any domain-lifted WFOMC computation over two logical variables:
\begin{equation*}
  \begin{aligned}
    & \symbwfomc(\Psi, n, w, \negw)
    = \sum_{\tau_1, \dots, \tau_n \in [p]} \ \prod_{k=1}^n w_{\tau_k} \prod_{1 \leq i < j \leq n} r_{\tau_i \tau_j},
  \end{aligned}
\end{equation*}
where
\begin{equation*}
    \begin{aligned}
        w_k = W(C_k, \weights)\qquad \text{and} \qquad
        r_{st} = \sum_{\substack{\pi \in D, \\ C_{s}(a) \land C_{t}(b) \land \pi(a,b) \models \psi(a,b) \land \psi(b,a)}} W(\pi, \weights)
    \end{aligned}
\end{equation*}
for $C = \{C_1, C_2, \dots, C_p\}$ being the (valid) cells of $\Psi$ and $D$ being the set of all 2-tables of $\Psi$.
Define
\begin{equation}\label{eq:f-def}
  f_m( \tau_1, \dots, \tau_m) = \prod_{k=1}^m w_{\tau_k} \prod_{1 \leq i < j \leq m} r_{\tau_i \tau_j}
\end{equation}
for each $m \in [n]$ and $\tau_1, \dots, \tau_m \in [p]$.
Note that the function $f_m( \tau_1, \dots, \tau_m)$ is a slightly more general version of $T_m(\bm{k})$ introduced in \cref{sec:3-single-lo}.
While $T_m(\bm{k})$ only considers a cell configuration $\bm{k}$ (i.e., how many domain elements have been assigned to each cell) over $m$ elements, the function $f_m( \tau_1, \dots, \tau_m)$ specifies the particular cell $C_{\tau_i}$ for each element from $[m]$.

Then, following the same idea as in \Cref{eq:induction_wmc}, for $1 \le m < n$, it holds that
\begin{equation}
    \label{eq:f-recursion}
    \begin{aligned}
        & f_{m+1}( \tau_1, \dots, \tau_{m+1})
        = f_m(\tau_1, \dots, \tau_m) \cdot w_{\tau_{m+1}} \cdot \prod_{i=1}^{m} r_{\tau_i \tau_{m+1}},
    \end{aligned}
\end{equation}
and the WFOMC can be obtained as
\begin{equation}
\label{eq:f-answer}
  \symbwfomc(\Psi, n, w, \negw) = \sum_{\tau_1, \dots, \tau_n \in [p]} f_n(\tau_1, \dots, \tau_n).
\end{equation}

\subsection{The Algorithm}
We will again consider $\Psi$ to be a universally quantified \fotwo sentence which can be accomplished using the same Skolemization procedure as before \citep{broecketal14:wfomc-skolem}.
Then our task is to compute $\symbwfomc(\Psi \land \loaxiom(\lopred) \land \succaxiom(S), n, w, \negw)$ where $\lopred$ and  $S$ are distinguished binary relations from $\preds{\Psi}$.
Again, let $C = \{C_1, C_2, \dots, C_p\}$ be the set of possible cells (1-types) of $\Psi$, and $C_{\tau_i}$ ($\tau_i \in [p]$) be the cell realized by the element $i$.

Each of the two axioms implies an order of elements.
We can, without loss of generality, fix the order of $\lopred$ to be $1 \leq 2 \leq 3 \leq \dots \leq n$.
For any other possible order of $\lopred$, applying a permutation on the indices of elements maps the models to those for the fixed $\lopred$ order while preserving the weight, which is formally stated in \cref{lemma:n-factorial}.

\paragraph{Computing WFOMC for a Fixed $S$-Order}
As a warm-up, let us also fix the order of $S$ as a permutation of $[n]$.
We compute the WFOMC recursively following the idea of \Cref{eq:f-def,eq:f-recursion,eq:f-answer}.
For each $s,t \in [p]$ and $k \in \{1,2,3\}$, define
\begin{equation*}
  \begin{aligned}
  % r_{s,t,k} = \sum_{\substack{\pi \in D, \\ C_{s}(a) \land C_{t}(b) \land \pi(a,b) \models \psi(a,b) \land \psi(b,a) \land \phi_k(a,b)}} W(\pi),
  r_{s,t,k} &= \symbwmc(\psi_{st}(a,b)\land\phi_k(a,b), \weights),\\
  \wlo_s &= \symbwmc(\psi^\le(c,c) \land C_s(c) \land \neg S(c,c), \weights),
  \end{aligned}
\end{equation*}
where $a,b,c\in\dom$ are arbitrary domain elements, $\psi_{st}(a,b)$ is still defined as the simplified version of $\psi(a,b)\land\psi(b,a)$ given that $C_s(a)\land C_t(b)$ holds, and
\begin{equation*}
  \begin{aligned}
    \phi_1(a,b) &= (a < b) \land \lnot S(a,b) \land \lnot S(b,a), \\
    \phi_2(a,b) &= (a < b) \land S(a,b) \land \lnot S(b,a), \\
    \phi_3(a,b) &= (a < b) \land \lnot S(a,b) \land S(b,a).
  \end{aligned}
\end{equation*}

The term $r_{s,t,k}$ is a refinement of $r_{\tau_i\tau_j}$ in \Cref{eq:f-def} since $\lopred$ and $S$ are totally interpreted and, therefore, $r_{\tau_i\tau_j}$ can only take one of the three values from $\{r_{\tau_i,\tau_j,1}, r_{\tau_i,\tau_j,2}, r_{\tau_i,\tau_j,3}\}$.
Note that we always compute $r_{\tau_i,\tau_j}$ for those $i,j$ such that $i<j$ and we fix the order of $\lopred$ as the natural order, hence $\phi_k(a,b)$ does not involve the cases in which $b < a$ holds.
% The term $r_{s,t,k}$ is a refinement of $r_{\tau_i,\tau_j}$ in \Cref{eq:f-def} since $\lopred$ and $S$ are totally interpreted and, therefore, $r_{\tau_i,\tau_j}$ can only take one of the three values from $\{r_{\tau_i,\tau_j,1}, r_{\tau_i,\tau_j,2}, r_{\tau_i,\tau_j,3}\}$.
% Note that we always compute $r_{\tau_i, \tau_j}$ for those $i,j$ such that $i<j$ and we fix the order of $\lopred$ as the natural order, hence $\phi_k(a,b)$ does not involve the cases in which $\lnot (a\lopred b) \land (b\lopred a)$ holds.

The function $f$ in \Cref{eq:f-def} can be adapted as follows:
\begin{equation}
\label{eq:g-def}
  g_m(\tau_1, \dots, \tau_m) = \prod_{i=1}^m \wlo_{\tau_i} \prod_{1 \le i < j \le m} r_{\tau_i,\tau_j,\kappa_{i,j}},
\end{equation}
where each $\kappa_{i,j}$ takes the only value from $\{1,2,3\}$ such that $\phi_{\kappa_{i,j}}(i,j)$ is satisfied.
We then obtain the recursive computation for $g$ as
\begin{equation}\label{eq:g-recursion}
  \begin{aligned}
  & g_{m+1}(\tau_1, \dots, \tau_{m+1})
  = g_m(\tau_1, \dots, \tau_{m}) \cdot \wlo_{\tau_{m+1}} \cdot \prod_{i=1}^{m} r_{\tau_i,\tau_{m+1},\kappa_{i,m+1}}.
  \end{aligned}
\end{equation}

\paragraph{Computing WFOMC for All $S$-Orders}
The computation above has two main drawbacks.
It computes WFOMC for a fixed order of $S$, and a concrete cell sequence $(\tau_1, \dots, \tau_{m+1})$ is used to determine the weights.
Those two issues might lead to $n! \cdot p^n$ possible $g$-values to be computed.
The key to addressing the issues is figuring out how the terms $r_{\tau_i,\tau_{m+1},\kappa_{i,m+1}}$ can appear in \Cref{eq:g-recursion} (i.e., the finer selection of $\kappa_{i,m+1}$ from the three values with respect to $S$).

Imagine that we place the elements $1, 2, \dots, n$ in a line in order (defined by $\lopred$) from left to right.
The order implied by $S$ is like a string connecting the $n$ elements.
If we focus on the first $m$ elements, ignoring other elements, the string breaks into several \emph{segments}.
Formally, a segment is a maximal sequence of elements $x_1, x_2, \dots, x_k$ such that every element in the sequence is within the first $m$ elements of the order of $\lopred$, and $S(x_1, x_2) \land \dots \land S(x_{k-1}, x_k)$ holds.
We call the element $x_1$ the head of the segment, and $x_k$ the tail.
An example is shown in \Cref{fig:segments} where we consider the order of $S$ as $5 \to 2 \to 8 \to 1 \to 6 \to 4 \to 7 \to 10 \to 3 \to 9$.
Looking at only the prefix $\{1,2,\dots,6\}$, there are three segments $5 \to 2$, $1 \to 6 \to 4$ and $3$.

\begin{figure}
  \centering
      \begin{tikzpicture}[scale=0.9]
        \tikzstyle{vertex}=[circle,draw, scale=0.8, minimum size=3mm]
        \tikzstyle{redvertex}=[circle,draw,red, scale=0.8, minimum size=3mm]
        \node[redvertex] (5) at (5,2.8) {$5$};
        \node[redvertex] (2) at (2,2.8) {$2$};
        \node[vertex] (8) at (8,2.1) {$8$};
        \node[redvertex] (1) at (1,1.4) {$1$};
        \node[redvertex] (6) at (6,1.4) {$6$};
        \node[redvertex] (4) at (4,0.7) {$4$};
        \node[vertex] (7) at (7,0.7) {$7$};
        \node[vertex] (10) at (10,0.7) {$10$};
        \node[redvertex] (3) at (3,0) {$3$};
        \node[vertex] (9) at (9,0) {$9$};
        \path[->,red,thick] (5) edge (2);
        \path[->] (2) edge (8);
        \path[->] (8) edge (1);
        \path[->,red,thick] (1) edge (6);
        \path[->,red,thick] (6) edge (4);
        \path[->] (4) edge (7);
        \path[->] (7) edge (10);
        \path[->] (10) edge (3);
        \path[->] (3) edge (9);
        \draw[dashed,thick] (6.5,-0.4)--(6.5,3.2);
      \end{tikzpicture}
  \caption{An illustration of the segments (red nodes and red arrows) of the order of $S$ ($5 \to 2 \to 8 \to 1 \to 6 \to 4 \to 7 \to 10 \to 3 \to 9$) on a prefix of elements $\{1, 2, \dots, 6\}$.}
  \label{fig:segments}
\Description{The figure shows a sequence of 10 domain elements connected by another successor relation forming the sequence (5, 2, 8, 1, 6, 4, 7, 10, 3, 9).
The elements (of the first linear order) 7 to 10 are cut off, indicating that we consider only the prefix 1 to 6.
With such a cut-off, we have three segments with respect to the successor relation, specifically, (5,2), (1, 6, 4), and a singleton (3).}
\end{figure}
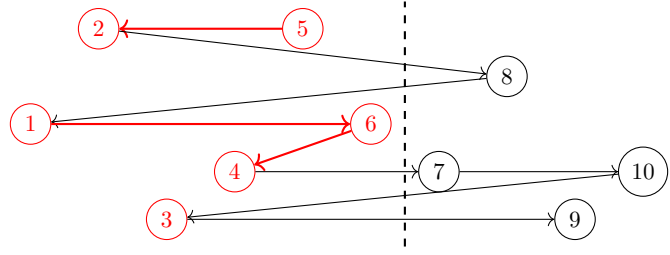

%This observation suggests the selection of $\kappa_{i,m}$ as follows: If $i$ is the head of a segment within the first $m-1$ elements but is replaced by $m$, then $\kappa_{i,m} = 2$; If $i$ is the tail of a segment within the first $m-1$ elements but is replaced by $m$, $\kappa_{i,m}=3$; otherwise $\kappa_{i,m}=1$.
%\lnote{give an example of the computation of $\kappa_{i,m}$?}

Denote a segment with the head realizing the cell $C_s$ and the tail realizing $C_t$ by $\segment{s}{t}$.
Let $\bm{k} = (k_1, \dots, k_p)$ be a cell configuration.
% be the vector of length $p$ where $k_i$ is the number of elements $k \in [m]$ such that $\tau_k = i$. In other words, $\bm{k}$ is the cell configuration for the first $m$ elements.
Observe that the value $\lambda = \prod_{i=1}^{m} r_{\tau_i,\tau_{m+1},\kappa_{i,m+1}}$ can be classified into one of the following four types depending on $\tau_{m+1}$ and the connection between the element $m+1$ and the former $m$ elements with respect to $S$:

\begin{itemize}
  \item The new element $m+1$ merges two segments $\segment{a}{b}$ and $\segment{c}{d}$ by linking the tail of $\segment{a}{b}$ to $m+1$ and linking $m+1$ to the head of $\segment{c}{d}$. If $b \neq c$, then $\lambda$ takes the value
  \begin{equation*}
    \begin{aligned}
    \lambda_{\text{merge1}}(b,c,\tau_{m+1},\bm{k}) =\ & r_{b,\tau_{m+1},2} \cdot r_{c,\tau_{m+1},3} \cdot \left(r_{b,\tau_{m+1},1}\right)^{k_b-1}
     \cdot \left(r_{c,\tau_{m+1},1}\right)^{k_c-1} \cdot \prod_{\substack{s \in [p],\\ s \neq b, \\ s \neq c}} \left(r_{s,\tau_{m+1},1}\right)^{k_s}.
    \end{aligned}
  \end{equation*}
  If $b=c$ (i.e., the same cell is realized by the tail of the first segment and by the head of the second segment), then $\lambda$ takes the value
  \begin{equation*}
    \begin{aligned}
    \lambda_{\text{merge2}}(b,\tau_{m+1},\bm{k}) =\ & r_{b,\tau_{m+1},2} \cdot r_{b,\tau_{m+1},3}
     \cdot \left(r_{b,\tau_{m+1},1}\right)^{k_b-2}
    \cdot \prod_{\substack{s \in [p],\\ s \neq b}} \left(r_{s,\tau_{m+1},1}\right)^{k_s}.
    \end{aligned}
  \end{equation*}
  \item $m+1$ extends a segment $\segment{a}{b}$ serving as the new head. In that case, $\lambda$ takes the value
  \begin{equation*}
    \lambda_{\text{head}}(a,\tau_{m+1},\bm{k}) =\ r_{a,\tau_{m+1},3} \cdot \left(r_{a,\tau_{m+1},1}\right)^{k_a-1} \cdot \prod_{\substack{s \in [p],\\ s \neq a}} \left(r_{s,\tau_{m+1},1}\right)^{k_s}.
  \end{equation*}
  \item $m+1$ extends a segment $\segment{a}{b}$ serving as the new tail. Then, $\lambda$ takes the value
  \begin{equation*}
    \lambda_{\text{tail}}(b,\tau_{m+1},\bm{k}) =\ r_{b,\tau_{m+1},2} \cdot \left(r_{b,\tau_{m+1},1}\right)^{k_b-1} \cdot \prod_{\substack{s \in [p],\\ s \neq b}} \left(r_{s,\tau_{m+1},1}\right)^{k_s}.
  \end{equation*}
  \item $m+1$ creates a new segment containing itself only. In that scenario, $\lambda$ takes the value
  \begin{equation*}
    \lambda_{\text{only}}(\tau_{m+1},\bm{k}) = \prod_{s \in [p]} \left(r_{s,\tau_{m+1},1}\right)^{k_s}.
  \end{equation*}
\end{itemize}

The computation of $\lambda$ suggests that we no longer need to care about the concrete order of $S$ and the cell sequence $(\tau_1, \dots, \tau_{m+1})$, but rather $\bm{k}$, $\tau_{m+1}$, the behavior of $m+1$ with respect to $S$, and the number of cell pairs realized by heads and tails of segments.
%Define $\bm{k}$ similarly as $\bm{k}$ but within elements $1, \dots, m+1$, and
Let $\bm{\rho} = (\rho_{1,1}, \dots, \rho_{p,p})$ be the vector of length $p^2$ where $\rho_{s,t}$ is the number of segments $\segment{s}{t}$ within elements $1, \dots, m$.
Now we can adapt \Cref{eq:g-def} by grouping orders of $S$ and the possible cell sequences $(\tau_1, \dots, \tau_m)$ that have the same $\bm{k}$ and $\bm{\rho}$.
Define $h_m(\bm{k}, \bm{\rho})$ as
\begin{equation*}
  h_m( \bm{k}, \bm{\rho}) = \sum_{\substack{\tau_1, \dots, \tau_m \text{ satisfying } \bm{k},\\ \text{segments of } 1, \dots, m \\ \text{satisfying } \bm{\rho}}} \  \prod_{i=1}^m \wlo_{\tau_i} \prod_{1 \le i < j \le m} r_{\tau_i,\tau_j,\kappa_{i,j}}.
\end{equation*}
Note that by the analysis above, the value of $\prod_{1 \leq i < j \leq m} r_{\tau_i,\tau_j,\kappa_{i,j}}$ is also unique given $(\tau_1, \dots, \tau_m)$ and the segments of $1, \dots, m$.

%The value of $h$ for $m$ elements can be computed from $h$ for $m-1$ elements by enumerating the cell $C_{\tau}$ of $m$ and the behavior of $m$ with respect to $S$.
%Denote by $\bm{k}-\bm{\vecdelta}_l$ the vector of size $p$ such that
%\begin{equation*}
%  k^{-\tau}_{s} = \begin{cases}
%    k_{s}-1, & s=\tau, \\
%    k_{s}, & \mbox{otherwise.}
%  \end{cases}
%\end{equation*}
%
%Similarly, $\bm{\rho}^{-(a,b)}$ denotes the vector of size $p^2$ such that
%\begin{equation*}
%  \rho^{-(a,b)}_{s,t} = \begin{cases}
%    \rho_{s,t} - 1, & s=a,\ t=b, \\
%    \rho_{s,t}, & \mbox{otherwise,}
%  \end{cases}
%\end{equation*}
%and $\bm{\rho}^{+(a,b)}$ is defined similarly.
%The operation can be nested as well, e.g., $\bm{\rho}^{-(a,b)+(c,d)}$ refers to $\left(\bm{\rho}^{-(a,b)}\right)^{+(c,d)}$.

We compute $h_{m+1}(\bm{k},\bm{\rho})$ given the cell $C_{l}$ realized by the element $m+1$ and given the behavior of $m+1$, i.e., $\beta\in \{\text{merge1},\text{merge2},\text{head},\text{tail},\text{only}\}$, denoted by $h_{m+1}(\bm{k},\bm{\rho} | l,\beta)$:
\begin{itemize}
  \item If $\beta=\text{merge1}$, then
  $m+1$ merges two segments $\segment{a}{b}$ and $\segment{c}{d}$ by linking the tail of $\segment{a}{b}$ to $m+1$ and linking $m+1$ to the head of $\segment{c}{d}$ for some $a,b,c,d \in [p]$ and $b \neq c$.
  Elements $1, \dots, m$ have cells consistent with $\bm{k}-\bm{\vecdelta}_l$ and segments consistent with $\bm{\rho}^{\text{merge1}} = \bm{\rho}-\bm{\vecdelta}_{a,d}+\bm{\vecdelta}_{a,b}+\bm{\vecdelta}_{c,d}$.
  Then,
      \begin{equation*}
        \begin{aligned}
           h_{m+1}(\bm{k},\bm{\rho} | l,\text{merge1})
           = \sum_{\substack{a,b,c,d \in [p],\ b \neq c, \\ \bm{\rho}^{\text{merge1}} \ge 0}} \Bigl(h_{m}( \bm{k}-\bm{\vecdelta}_l, \bm{\rho}^{\text{merge1}}) \cdot \eta \cdot \wlo_l \cdot \lambda_{\text{merge1}}(b,c,l,\bm{k}-\bm{\vecdelta}_l)\Bigl),
        \end{aligned}
      \end{equation*}
      where
      \begin{equation*}
        \eta = \begin{cases}
          \rho^{\text{merge1}}_{a,b} \cdot \rho^{\text{merge1}}_{c,d}, & a \neq c \text{ or } b \neq d, \\
          \rho^{\text{merge1}}_{a,b}\left(\rho^{\text{merge1}}_{a,b}-1\right), & \mbox{otherwise.}
        \end{cases}
      \end{equation*}
      
  \item $\beta=\text{merge2}$:
  $m+1$ merges two segments $\segment{a}{b}$ and $\segment{b}{d}$ by linking the tail of $\segment{a}{b}$ to $m+1$ and linking $m+1$ to the head of $\segment{b}{d}$ for some $a,b,d \in [p]$.
  Elements $1, \dots, m$ have cells consistent with $\bm{k}-\bm{\vecdelta}_l$ and segments consistent with $\bm{\rho}^{\text{merge2}} = \bm{\rho}-\bm{\vecdelta}_{a,d}+\bm{\vecdelta}_{a,b}+\bm{\vecdelta}_{b,d}$.
  Then,
      \begin{equation*}
        \begin{aligned}
           h_{m+1}(\bm{k},\bm{\rho} | l,\text{merge2})
           = \sum_{a,b,d \in [p],\ \bm{\rho}^{\text{merge2}} \ge 0} \Bigl(h_{m}( \bm{k}-\bm{\vecdelta}_l, \bm{\rho}^{\text{merge2}}) \cdot \eta \cdot \wlo_l \cdot \lambda_{\text{merge2}}(b,l,\bm{k}-\bm{\vecdelta}_l)\Bigl),
        \end{aligned}
      \end{equation*}
      where
      \begin{equation*}
        \eta = \begin{cases}
          \rho^{\text{merge2}}_{a,b} \cdot \rho^{\text{merge2}}_{b,d}, & a \neq b \text{ or } b \neq d, \\
          \rho^{\text{merge2}}_{a,b}\left(\rho^{\text{merge2}}_{a,b}-1\right), & \mbox{otherwise.}
        \end{cases}
      \end{equation*}
      
  \item $\beta=\text{head}$:
  $m+1$ extends a segment $\segment{a}{b}$ serving as a new head for some $a,b \in [p]$.
  Elements $1, \dots, m$ have cells consistent with $\bm{k}-\bm{\vecdelta}_l$ and segments consistent with $\bm{\rho}^{\text{head}} = \bm{\rho}-\bm{\vecdelta}_{l,b}+\bm{\vecdelta}_{a,b}$. Then,
      \begin{equation*}
        \begin{aligned}
          h_{m+1}(\bm{k},\bm{\rho} | l,\text{head})
          =&  \sum_{a,b \in [p],\ \bm{\rho}^{\text{head}} \ge 0} \Bigl(h_{m}( \bm{k}-\bm{\vecdelta}_l, \bm{\rho}^{\text{head}}) \cdot \rho^{\text{head}}_{a,b} \cdot \wlo_l \cdot \lambda_{\text{head}}(a,l,\bm{k}-\bm{\vecdelta}_l\Bigl).
        \end{aligned}
      \end{equation*}
      
  \item $\beta=\text{tail}$:
  $m+1$ extends a segment $\segment{a}{b}$ serving as a new tail for some $a,b \in [p]$.
  Elements $1, \dots, m$ have cells consistent with $\bm{k}-\bm{\vecdelta}_l$ and segments consistent with $\bm{\rho}^{\text{tail}} = \bm{\rho}-\bm{\vecdelta}_{a,l}+\bm{\vecdelta}_{a,b}$.
  Then,
      \begin{equation*}
        \begin{aligned}
          h_{m+1}(\bm{k},\bm{\rho} | l,\text{tail})
          =&  \sum_{a,b \in [p],\ \bm{\rho}^{\text{tail}} \ge 0} \Bigl(h_{m}( \bm{k}-\bm{\vecdelta}_l, \bm{\rho}^{\text{tail}}) \cdot \rho^{\text{tail}}_{a,b} \cdot \wlo_l \cdot \lambda_{\text{tail}}(b,l,\bm{k}-\bm{\vecdelta}_l)\Bigl).
        \end{aligned}
      \end{equation*}
      
  \item $\beta=\text{only}$:
  $m+1$ creates a new segment containing itself only.
  Elements $1, \dots, m$ have cells consistent with $\bm{k}-\bm{\vecdelta}_l$ and segments consistent with $\bm{\rho}^{\text{only}} = \bm{\rho}-\bm{\vecdelta}_{l,l}$.
  Then,
      \begin{equation*}
        \begin{aligned}
          &h_{m+1}(\bm{k},\bm{\rho} | l,\text{only})
          =  h_{m}( \bm{k}-\bm{\vecdelta}_l, \bm{\rho}^{\text{only}}) \cdot \wlo_l \cdot \lambda_{\text{only}}(l,\bm{k}-\bm{\vecdelta}_l).
        \end{aligned}
      \end{equation*}
\end{itemize}

Summing the above terms up, we have
\begin{equation*}
  \begin{aligned}
    h_{m+1}( \bm{k}, \bm{\rho})
    = \sum_{\substack{l \in [p],\\ k_{l}>0}} \ \sum_{\substack{\beta\in \{\text{merge1},\text{merge2},\\ \text{head},\text{tail},\text{only}\}}} h_m(\bm{k},\bm{\rho} | l,\beta).
  \end{aligned}
\end{equation*}

Finally, the WFOMC for the fixed order of $\lopred$ can be obtained when each of the $n$ elements realizes some cell and there is only $1$ segment, i.e.,
\begin{equation*}
  \gamma = \sum_{\bm{k}\in\mathbb{N}^p: |\bm{k}|=n} \sum_{\bm{\rho}\in\mathbb{N}^{p^2}: |\bm{\rho}|=1} h_n(\bm{k}, \bm{\rho}).
\end{equation*}

As we have already mentioned above, every other order of $\lopred$ has the same weight as the one with the natural order of numbers.
Therefore,
\begin{equation*}
  \symbwfomc(\Psi \land \loaxiom(L) \land \succaxiom(S), n, w, \negw) = n! \cdot \gamma.
\end{equation*}
The procedure is summarized in \Cref{alg:lo+succ}.

\begin{algorithm}[t]
\caption{WFOMC for \fotwo+$\loaxiom$+$\succaxiom$}
\label{alg:lo+succ}
\KwIn{An \fotwo sentence $\Psi = \forall x \forall y: \psi(x, y)$ with $\{\lopred, S\} \subseteq \preds{\Psi}$, domain size $n$, weighting functions $(\weights)$}
\KwOut{$\symbwfomc(\Psi \land \loaxiom(\lopred) \land \succaxiom(S), n, \weights)$}

$h_1(\bm{\vecdelta}_l,\bm{\vecdelta}_{l,l}) \gets \wlo_l \text{ for each } l \in [p]$ \\
\For{$m \gets 2$ \KwTo $n$} {
    $h_m( \bm{k}, \bm{\rho}) \gets 0$ for every $\bm{k} \in \mathbb{N}^p$ and $\bm{\rho} \in \mathbb{N}^{p^2}$\\
    \ForEach{\upshape $\bm{k}, \bm{\rho}$ \textbf{ such that } $|\bm{k}|=m$} {
        \ForEach{\upshape $l \in [p]$ \textbf{ such that } $k_{l}>0$} {
            \ForEach{\upshape $\beta\in \{\text{merge1},\text{merge2},\text{head},\text{tail},\text{only}\}$} {
                compute $h_m(\bm{k},\bm{\rho} | l,\beta)$ \\
                $h_{m}( \bm{k}, \bm{\rho}) \gets h_m( \bm{k}, \bm{\rho}) + h_m(\bm{k},\bm{\rho} | l,\beta)$
            }
        }
    }
}
$\gamma \gets \sum_{\bm{k}\in\mathbb{N}^p: |\bm{k}|=n} \sum_{\bm{\rho}\in\mathbb{N}^{p^2}: |\bm{\rho}|=1} h_n(\bm{k},\bm{\rho})$ \\
\Return $n! \cdot \gamma$
\end{algorithm}

\begin{theorem}\label{thm:lo+succ}
An \fotwo sentence with a linear order relation and a successor relation of another linear order is domain-liftable.
\end{theorem}

\begin{proof}
    Existence of \cref{alg:lo+succ} proves the statement.
    
    The correctness would again be proved by induction, with the critical step outlined in the derivation above.
    As for the time complexity, the number of $h$-values we need to compute is 
    $$n \cdot \mathcal{O}(n^p) \cdot \mathcal{O}(n^{p^2}) \in \mathcal{O}(n^{p^2+p+1})$$ and each computation takes time polynomial in $p$.
    Therefore, the fragment is domain-liftable.
\end{proof}

The result can, once again, be extended to the fragment of \ctwo with cardinality constraints using  \cref{lemma:c2+cc}.
We can also support unary evidence, as we did for \cref{alg:iwfomc,alg:iwfomc2}.
Last but not least, using our \ctwo encodings of the $k$-th successor in \cref{app:successor-by-lo} with respect to a linear order relation, we may also extend the domain-liftability result from having a linear order axiom to having its extended version.

\begin{corollary}
\ctwo with an extended linear order axiom and a successor relation of another linear order (and with cardinality constraints) is domain-liftable.
\end{corollary}

% \subsection{Multiple Successor Relations}

% \kqp{To be decided if we add this section}

% This algorithm can be adapted to \fotwo with any constant number of successor relations by using a (much) more comprehensive configuration $\sigma$. 

\section{Experiments}
\label{sec:7-experiments}
We have implemented \Cref{alg:iwfomc,alg:iwfomc2,alg:lo+succ} in Python%
\footnote{The implementation as well as the code to recreate all the experiments is available at \url{https://github.com/jan-toth/wfomc-over-ordered-domains}.}
and we evaluate them on various tasks, including performing exact lifted inference in Markov Logic Networks \citep{richardson06:mlns} with structure similar to the random graphs of \citet{watts98:ws-model} or several problems from combinatorics involving counting permutations taken from the MATH dataset~\citep{hendrycks21:math-dataset}.
Solving similar combinatorics problems using lifted techniques has been motivated by the work of \citet{totis23:cola+coso}.

Although our implementation is not optimized for performance, it still provides valuable insight into how the algorithms scale with problem size, especially by demonstrating their polynomial runtime.
We also compare our performance to existing alternative software that could be used to solve the tasks at hand without our algorithms.
We mainly compare against propositional (weighted) model counters since existing lifted inference software usually does not support linearly ordered domains.
One exception is \emph{RecursiveWFOMC}, i.e., a scalability-oriented improvement of \cref{alg:iwfomc}, which was introduced by \citet{meng24:recursive-wfomc}.
Note, however, that same as \cref{alg:iwfomc}, RecursiveWFOMC supports domain-liftable computations with a linear order axiom, yet without implicit access to its successor relations.
Hence, we are often left with comparisons against \emph{GANAK} \citep{shubham19:ganak,soos25:ganak2}, which is arguably a state-of-the-art WMC solver, as reported by its authors, and the solver \emph{d4} \citep{lagniez17:d4}, which sometimes outperforms GANAK.%
\footnote{The WMC solvers are available at \url{https://github.com/meelgroup/ganak} and  \url{https://github.com/crillab/d4v2}. The implementation of RecursiveWFOMC is available within our code.}

When we utilize lifted algorithms that do not implicitly support successor relations for problems that require successor access, we encode the successors explicitly using \Cref{eq:app:old_succ1}, which offers a runtime improvement over the encoding proposed in \cref{sec:4-successor-relations} as we further discuss in \cref{app:successor-by-lo}.
For the propositional solvers, we only include the runtime of the algorithms themselves; we exclude the runtime of producing a ground CNF, which is expected by the solvers as input (generating the CNF files has never been a bottleneck in our experiments).
Note that when we ground the problems, we append the grounding of the properties enumerated in \cref{def:loaxiom} of the linear order, or even the properties from \Cref{eq:successor,eq:cyclic_pred} whenever necessary.
That might seem like we are artificially adding complexity, since we might, similarly to how \cref{alg:iwfomc,alg:iwfomc2,alg:lo+succ} work, fix one linear order and multiply by $n!$ afterwards.
While that strategy would improve the running times of the propositional solvers, even outperforming \cref{alg:iwfomc2} at times, it is not applicable in general.
Recall \cref{ex:basic-succ}, where we organized $k_e$ English books and $k_m$ math books on a shelf so that each book category would be in one cluster.
If we fix the linear ordering of the books here, we will have two distinct solutions (either English books go first or the math books do), and we must multiply by $k_e! \cdot k_m!$ to obtain the correct overall model count.
Hence, fixing the order for propositional solvers does not necessarily yield a correct solution; therefore, we do not consider this strategy in our experiments.

In all of our experiments, even when we represent the model considered as a Markov Logic Network (MLN), we eventually solve a WFOMC instance of a specific first-order sentence (essentially performing inference over the models it represents).
To convert an MLN to a WFOMC problem, we employ reduction from \citet{broecketal14:wfomc-skolem}.
We refer readers to \cref{app:mlns} for details on both MLNs and their reduction to WFOMC.

We conducted all the experiments on a machine with an Intel(R) Core(TM) i7-14700K CPU, 64GB of RAM, and the Ubuntu 24.04.3 LTS operating system.
We did not limit RAM usage, but we set a 3-hour timeout.
If any entries in the figures below are missing, it is because their computation exceeded one of the available resources.
Additionally, we use the labels \texttt{Incremental\_v1}, \texttt{Incremental\_v2} and \texttt{Incremental\_v3} to refer to \cref{alg:iwfomc,alg:iwfomc2,alg:lo+succ}, respectively.
All of our experiments show, as evidenced by the figures below, that \cref{alg:iwfomc2} offers much better performance than any other existing solvers, sometimes offering a speed-up of several orders of magnitude while also scaling to non-trivial domain sizes with dozens or even hundreds of domain elements.
A similar property is shown for \cref{alg:lo+succ}, although its scaling is, as one might expect, due to the quadratic function in the polynomial degree (cf. proof of \cref{thm:lo+succ}), somewhat worse.

\subsection{Illustrative Examples}
Before delving into comparisons on complex problems or larger datasets, let us recall \cref{ex:sequence-split,ex:ws-like} which we used to illustrate the modeling capabilities of the linear order as well as its cyclic successor relation:
\begin{equation*}
    \begin{aligned}
         \Phi_1 = \forall x \forall y: (T(x) \wedge (x \lopred y)) \to T(y)
    \end{aligned}
    \qquad
    \begin{aligned}
    \Phi_2 = &(\forall x : \neg Edge(x, x))\\
    \land\;&(\forall x \forall y : Edge(x,y) \to Edge(y,x))\\   
    \land\;&(\forall x \forall y : CySucc(x,y) \to Edge(x,y))\\
    \land\;&(|Edge|=2\cdot(n+m))
    \end{aligned}
\end{equation*}
% \begin{minipage}{.45\textwidth}
%   \begin{align*}
%     \Phi_1 = \forall x \forall y: (T(x) \wedge (x \lopred y)) \to T(y)
%   \end{align*}
% \end{minipage}%
% \hfill
% \begin{minipage}{.45\textwidth}
%   \begin{align*}
%     \Phi_2 = &(\forall x : \neg Edge(x, x))\\
%     \land\;&(\forall x \forall y : Edge(x,y) \to Edge(y,x))\\   
%     \land\;&(\forall x \forall y : CySucc(x,y) \to Edge(x,y))\\
%     \land\;&(|Edge|=2\cdot(n+m))
%   \end{align*}
% \end{minipage}

While $\Phi_1$ splits a sequence of elements into a head and a tail, $\Phi_2$ defines a graph with a single cyclic chain going through all the vertices and $m$ additional edges with no other restrictions.
\cref{fig:illustrative-examples} shows running times for computing $\symbfomc(\Phi_1,n)$ and $\symbfomc(\Phi_2,n)$ over various domain sizes.
On smaller domains, we compare against other solvers; on larger domains, we showcase the polynomial running time trend as the domain increases.
Since $\Phi_1$ is a very simple problem not requiring any of the improvements developed in \cref{sec:4-successor-relations}, we only utilize \cref{alg:iwfomc} from the lifted techniques, which already allows us to scale to hundreds of domain elements. At the same time, the propositional counters do not even manage a domain of size 15.
As for $\Phi_2$, \cref{alg:iwfomc2} outperforms all alternatives by a large margin, scaling up to hundreds of domain elements, while the best of others struggle on a domain with $n=17$.

\begin{figure}[t]
    \centering
    
    \begin{subfigure}[b]{0.23\textwidth}
        \centering
        \includegraphics[width=\textwidth]{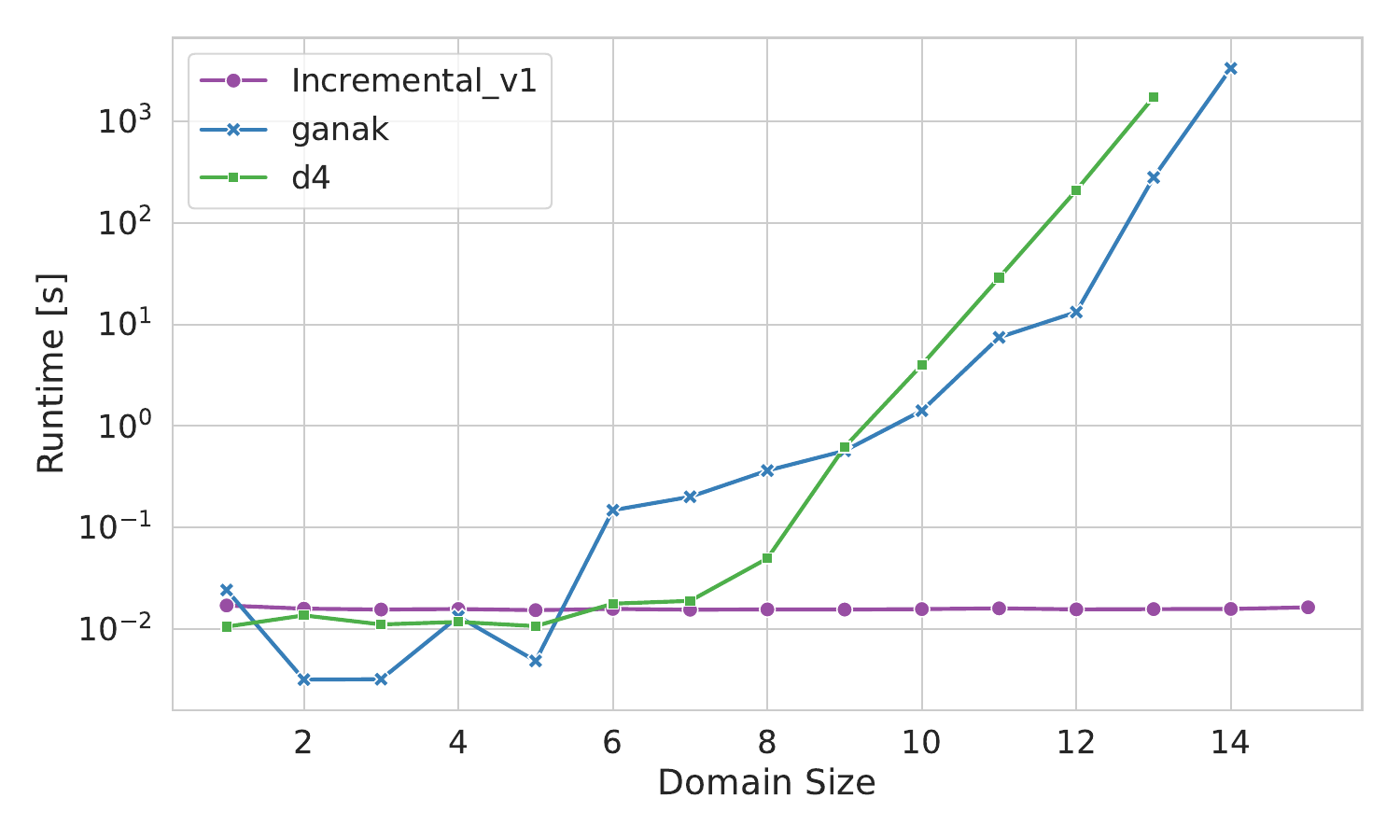}
        \caption{$\Phi_1$ for $n\le15$}
        % \label{fig:exp_a}
    \end{subfigure}
    \hfill
    \begin{subfigure}[b]{0.23\textwidth}
        \centering
        \includegraphics[width=\textwidth]{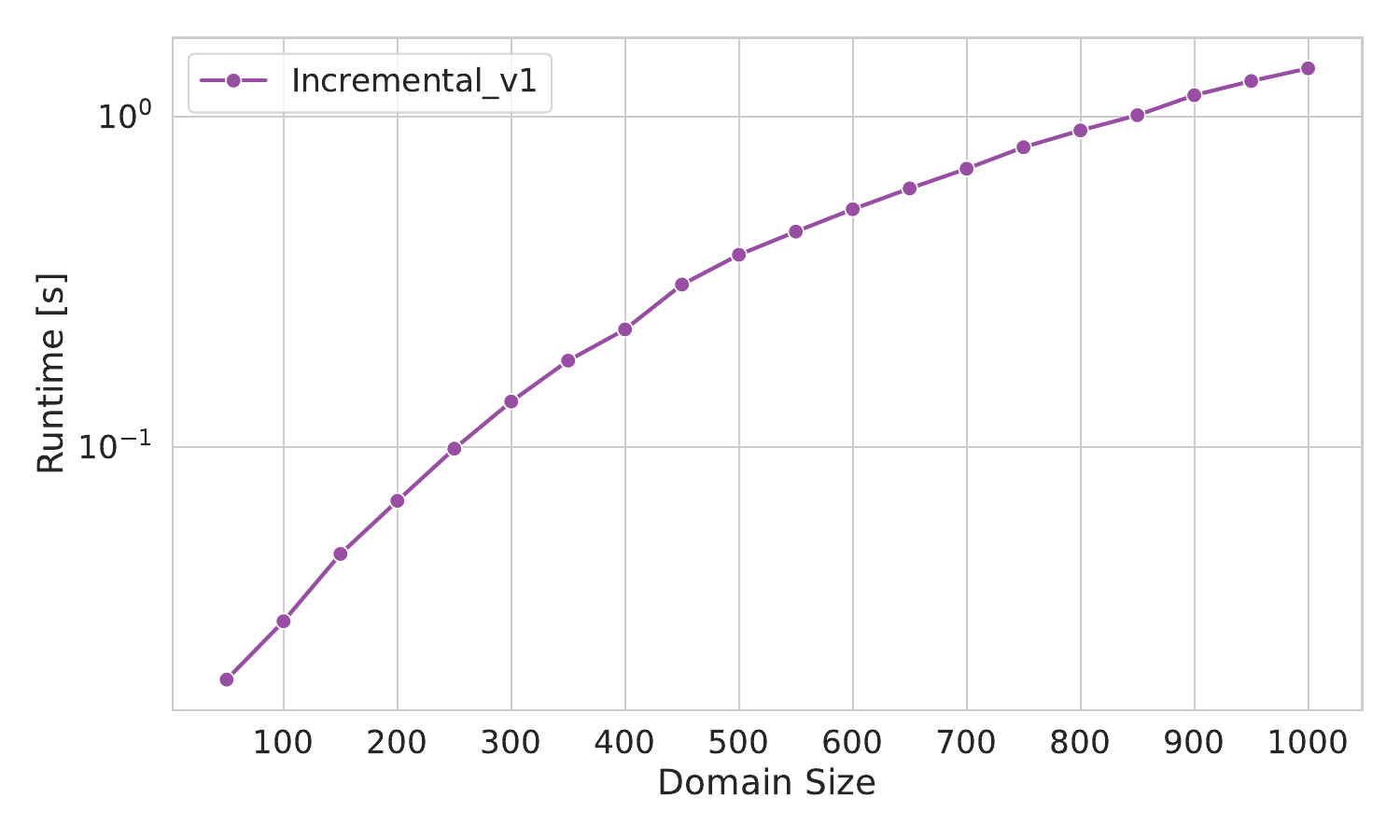}
        \caption{$\Phi_1$ for $n\le 1000$}
        % \label{fig:exp_b}
    \end{subfigure}
    \hfill
    \begin{subfigure}[b]{0.23\textwidth}
        \centering
        \includegraphics[width=\textwidth]{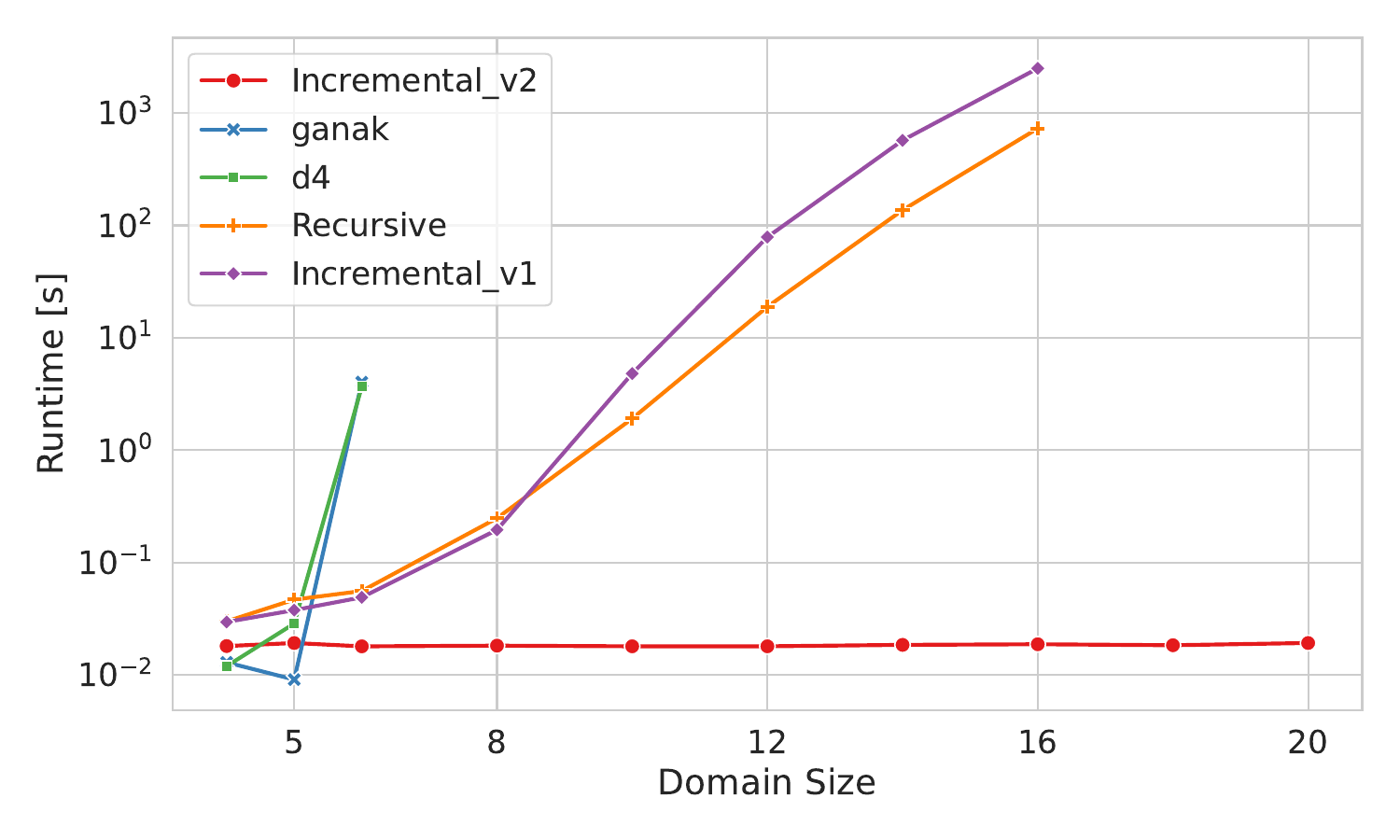}
        \caption{$\Phi_2$ for $n\le20$, $m=n$}
        % \label{fig:exp_c}
    \end{subfigure}
    \hfill
    \begin{subfigure}[b]{0.23\textwidth}
        \centering
        \includegraphics[width=\textwidth]{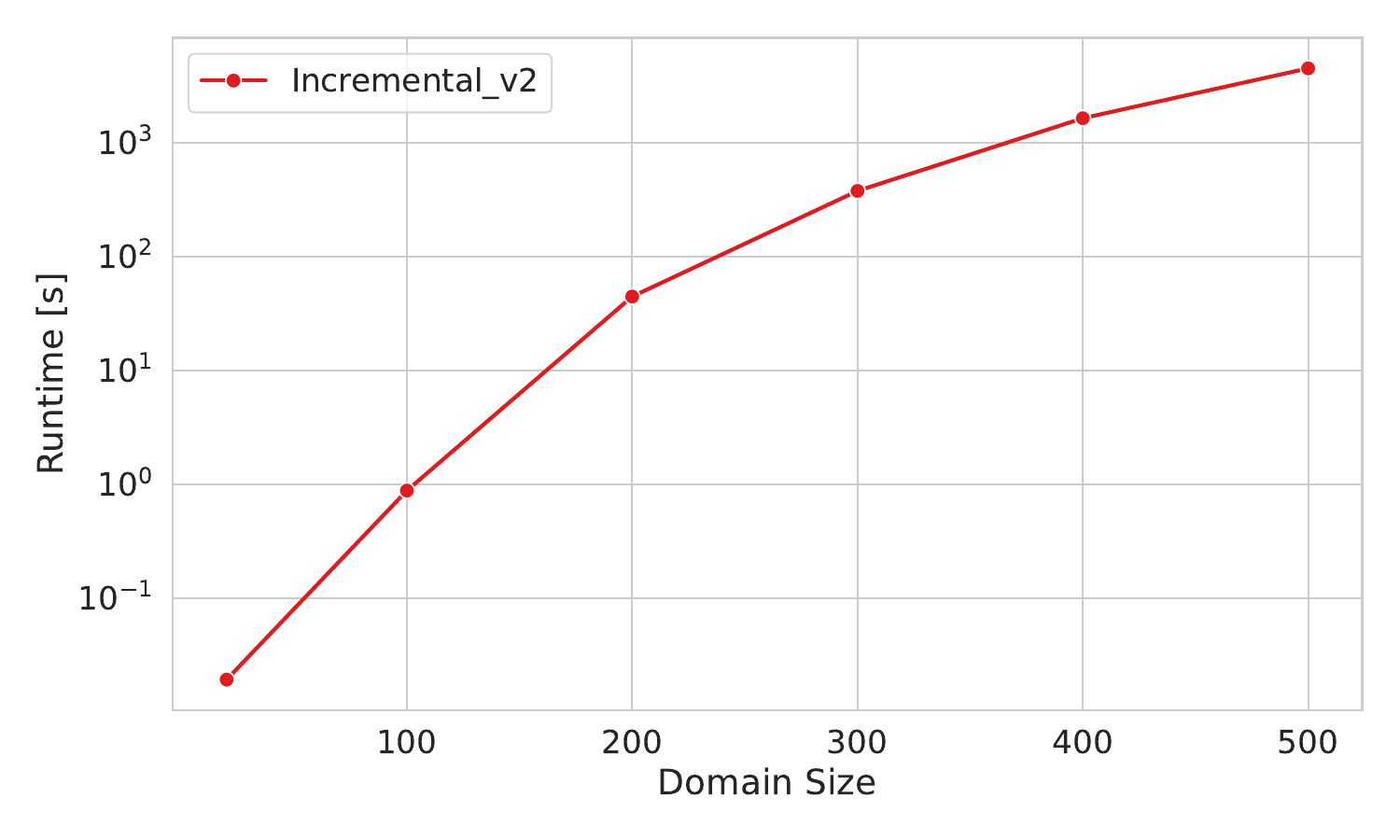}
        \caption{$\Phi_2$ for $n\le500$, $m=n$}
        % \label{fig:exp_d}
    \end{subfigure}
    
    \caption{Runtime comparisons and WFOMC scaling behavior for $\Phi_1$ and $\Phi_2$.}
    \label{fig:illustrative-examples}
\Description{Runtime measurements for the illustrative examples.
For the first problem, we compare Algorithm 1 with GANAK and d4. The propositional solvers can only solve instances with domain sizes of at most 14, whereas the lifted approach demonstrates constant scaling behavior on such small problems.
On problems scaling up a thousand domain elements, it demonstrates a clear polynomial tendency.

For the second problem, we compare all three lifted approaches and the propositional counters.
The propositional solvers only solve instances with domain sizes up to six, while Algorithm 1 and RecursiveWFOMC manage to get to problem sizes of 16.
Algorithm 2 demonstrates constant complexity for such small domain sizes; on larger domains, up to 500 domain elements, it shows a clear polynomial trend.
}    
\end{figure}

\subsection{Problems with the Linear Order and Its Successors}
Now, let us turn to more complex problems requiring the linear order and possibly some of its successor axioms.
We present 3 synthetic problems of varying complexity over arguably interesting structures requiring an ordered domain, as well as a set of 28 problems from the MATH dataset~\citep{hendrycks21:math-dataset} concerning permutation counting, which we managed to encode as first-order theories with the linear order.
When we talk about various problem complexity, we refer to the number of (valid) cells $p$, which, as evidenced by the proofs of \cref{th:fo2+lo_liftable,thm:general_linear_order}, determine the degree of our polynomial running times.%
\footnote{Note that using $p$ does not allow us to gauge the entire problem complexity when counting quantifiers or cardinality constraints are present due to \cref{lemma:c2+cc}, which relies on repeated calls to an oracle for \fotwo.
That is why, in \cref{fig:illustrative-examples}, $\Phi_1$ shows better scaling behavior than $\Phi_2$, even though $\Phi_1$ has two valid cells, while $\Phi_2$ only has one.
Nevertheless, with potential counting in mind, $p$ remains a useful (and accurate) complexity measure.}

\paragraph{Watts-Strogatz Model}
The WS model is a procedure for generating random graphs with properties such as short average paths and high clustering \citep{watts98:ws-model}.
Given $n$ ordered nodes, we first connect each node to its $K$ closest neighbors (assuming $K$ is an even integer) by undirected edges.
If the sequence reaches the end or beginning, we wrap around to the other end.
For every node, we \emph{rewire} each of its \emph{rightmost} edges with probability $\beta$.
Rewiring refers to switching the other endpoint to any other vertex.
We have already approximated the WS model for $K=2$ by the sentence $\Phi_2$, where we substituted the rewired edges with $m$ additional edges chosen at random.
Now, however, we design a Markov Logic Network $\Phi_{ws}$ that aims to capture the model exactly.
For brevity and easier exposition, we keep $K=2$.
\begin{align*}
    \Phi_{ws} = \{\;&(\infty,\; \forall x : \neg WiredEdge(x, x)),\\
    &(\infty,\; \forall x \forall y : \neg WiredEdge(x,y) \lor \neg WiredEdge(y,x)),\\   
    &(\infty,\; \forall x \exists^{=1} y : WiredEdge(x,y)),\\
    &(\infty,\; \forall x \forall y : Edge(x, y) \leftrightarrow (WiredEdge(x,y) \lor WiredEdge(y,x))),\\
    &(w_1,\; WiredEdge(x,y) \land \neg CySucc(x,y)),\\   
    &(w_2,\; WiredEdge(x,y) \land CySucc(x,y))\;\}
\end{align*}
The quantified sentences (the hard rules with infinite weight) define the undirected edges of the graph in terms of wired edges, which are directed, are not loops, and do not go both ways between any pair of vertices.
The soft rules represent two mutually exclusive events when there is a wired edge going from $x$ to $y$.
Each wired edge is either part of the cyclic chain or not, meaning it has been rewired or not, respectively.
The weights $w_1$ and $w_2$ are determined from the probability $\beta$.
Note that we could model a scenario for $K>2$, although we would require the generalized cyclic successors up to $K/2$.
After converting $\Phi_{ws}$ to a WFOMC instance, the number of valid cells for the problem is two.

The performance of each of the algorithms on $\Phi_{ws}$ is compared in \cref{fig:ws2_small}, while \cref{fig:ws2_large} showcases that \cref{alg:iwfomc2} can scale much further than any of the other tested approaches.
\cref{alg:iwfomc2} again demonstrates the best performance by a significant margin, solving WFOMC instances with up to 90 domain elements within the three-hour limit.
Note, however, that GANAK, albeit a propositional counter with exponential complexity, outperforms both \cref{alg:iwfomc} and RecursiveWFOMC, showing that a provably polynomial runtime of a lifted algorithm does not immediately necessitate better performance.

\paragraph{Hidden Markov Model with Higher-Order Dependencies}
Next, in order to evaluate performance for the general $k$-th successor relations, we consider a hidden Markov model with higher-order dependencies.
Specifically, we design an MLN $\Phi_{hmm}$ for simple weather predictions.
However, only \cref{alg:iwfomc2} implicitly supports the $k$-th successor relation, and the explicit encodings from \cref{app:successor-by-lo} scale poorly for both \cref{alg:iwfomc} and RecursiveWFOMC.
Therefore, we do not use them to compute the WFOMC of $\Phi_{hmm}$.
Nevertheless, to showcase their behavior, we also provide a standard hidden Markov Model $\Phi_{hmm_2}$ that drops the higher-order transitions.
\begin{equation*}
\begin{aligned}
\Phi_{hmm} = \{ \;
&(\infty, \forall x: \neg (Sn(x) \land Rn(x))), \\
&(\infty, \forall x: Sn(x) \lor Rn(x)), \\
&(\infty, (Sn(x) \land Succ_1(x,y)) \to Rn(x)), \\
&(0.5, S(x) \to Rn(x)), \\
&(1.0, (S(x) \land Succ_1(x,y))\to S(y)), \\
&(0.4, (S(x) \land Succ_2(x,y))\to S(y)), \\
&(0.1, (S(x) \land Succ_3(x,y))\to S(y))\; \}
\end{aligned}
\qquad
\begin{aligned}
    \Phi_{hmm_2} = \{\;
    &(\infty, \forall x: \neg (Sn(x) \land Rn(x))), \\
    &(\infty, \forall x: Sn(x) \lor Rn(x)), \\
    &(\infty, (Sn(x) \land Succ_1(x,y)) \to Rn(x)), \\
    &(0.5, S(x) \to Rn(x)), \\
    &(1.0, (S(x) \land Succ_1(x,y))\to S(y)), \\
\end{aligned}
\end{equation*}
Intuitively, the weather of someday is either rainy or sunny (the first two rules), and it is influenced by a hidden state $S(x)$ (the fourth rule).
The bottom soft rules model the transitions between hidden states.
Computing the probability of the third rule (i.e., ``If it is sunny today, it will rain tomorrow'') being true would give us our weather forecast.
The number of valid cells for both problems is now four.

The performance comparisons are now depicted in \cref{fig:higher_hmm_small,fig:hmm_small}.
One can observe the polynomial behavior of \cref{alg:iwfomc2} on larger domains in \cref{fig:higher_hmm_large,fig:hmm_large}.
With the higher number of cells, \cref{alg:iwfomc2} no longer scales to hundreds of domain elements. However, it still performs much better than any other algorithm considered, successfully computing WFOMC on domains more than four times larger than those the other solvers could handle.

\begin{figure}[bt]
    \centering

    \begin{subfigure}[b]{0.23\textwidth}
        \centering
        \includegraphics[width=\textwidth]{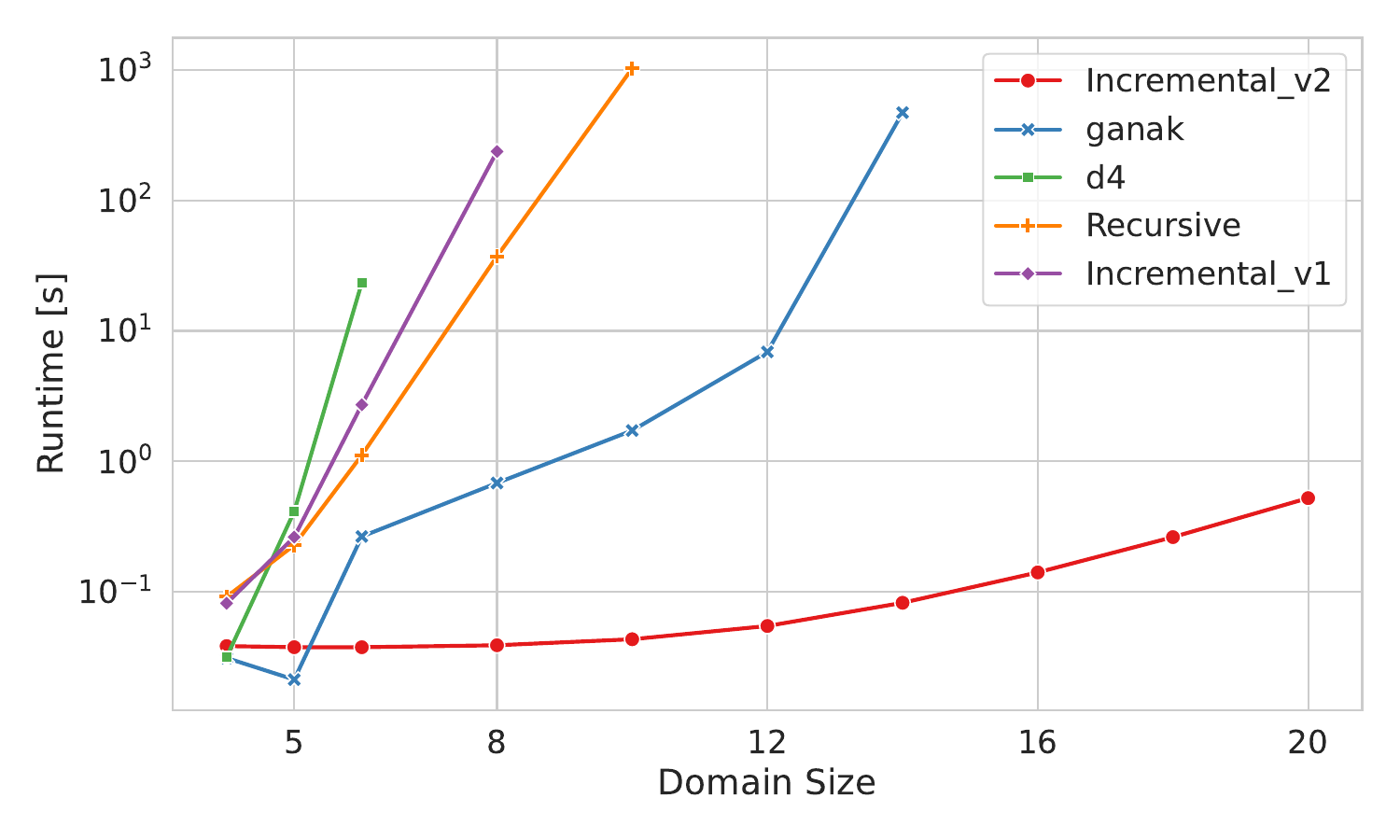}
        \caption{$\Phi_{ws}$ for $n\le20$}
        \label{fig:ws2_small}
    \end{subfigure}
    \hspace{6em}
    \begin{subfigure}[b]{0.23\textwidth}
        \centering
        \includegraphics[width=\textwidth]{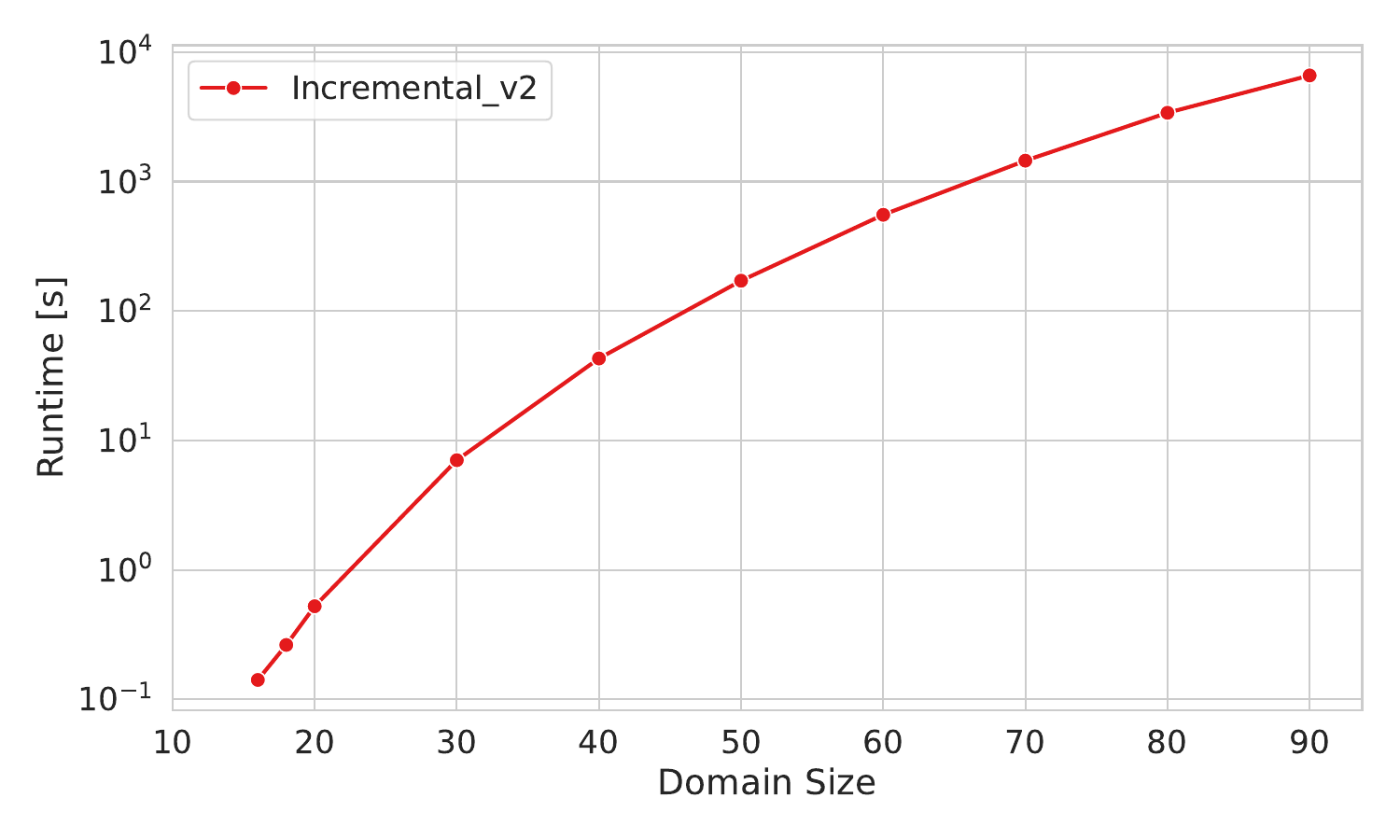}
        \caption{$\Phi_{ws}$ for $n\le 90$}
        \label{fig:ws2_large}
    \end{subfigure}

    \begin{subfigure}[b]{0.23\textwidth}
        \centering
        \includegraphics[width=\textwidth]{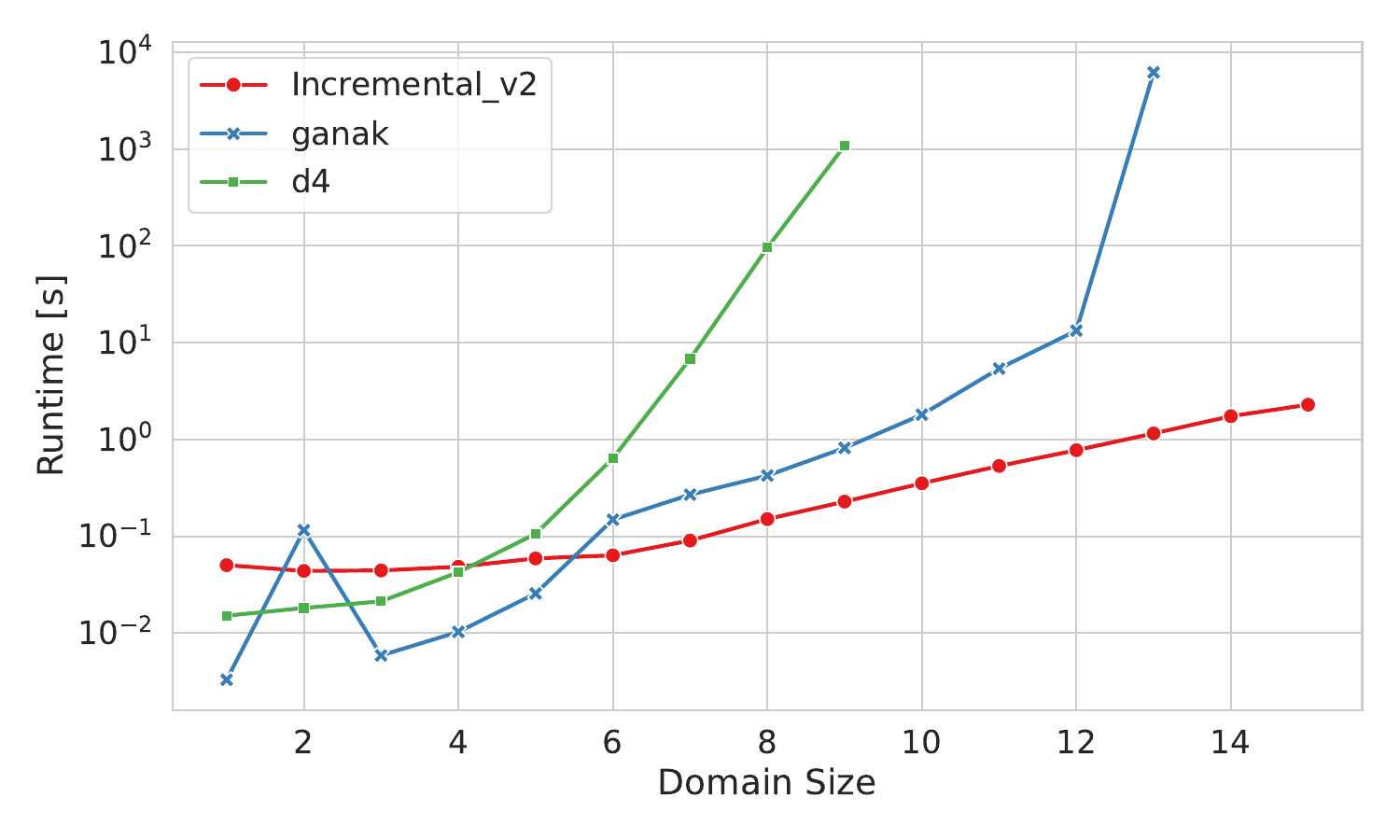}
        \caption{$\Phi_{hmm}$ for $n\le15$}
        \label{fig:higher_hmm_small}
    \end{subfigure}
    \hfill
    \begin{subfigure}[b]{0.23\textwidth}
        \centering
        \includegraphics[width=\textwidth]{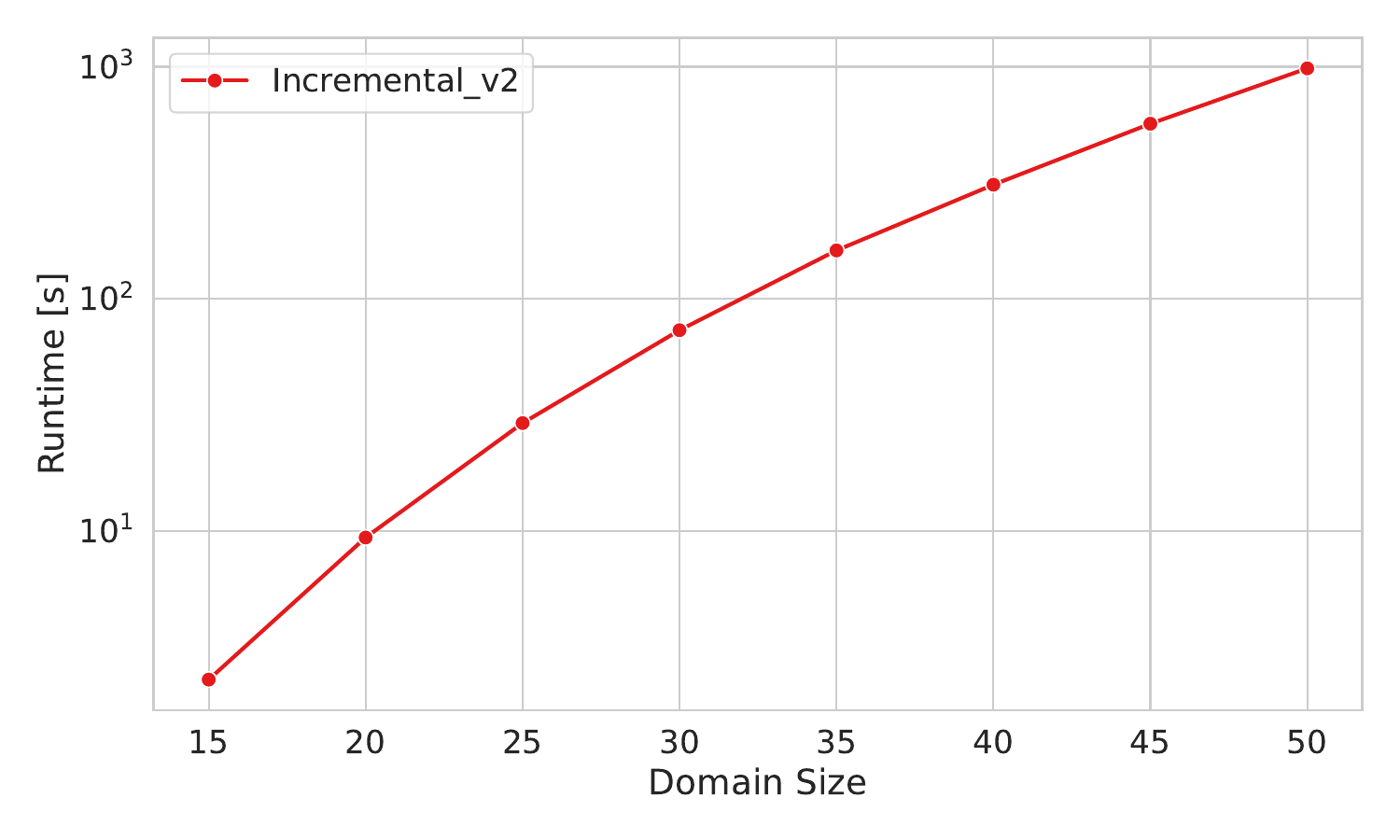}
        \caption{$\Phi_{hmm}$ for $n\le50$}
        \label{fig:higher_hmm_large}
    \end{subfigure}
    \hfill
    \begin{subfigure}[b]{0.23\textwidth}
        \centering
        \includegraphics[width=\textwidth]{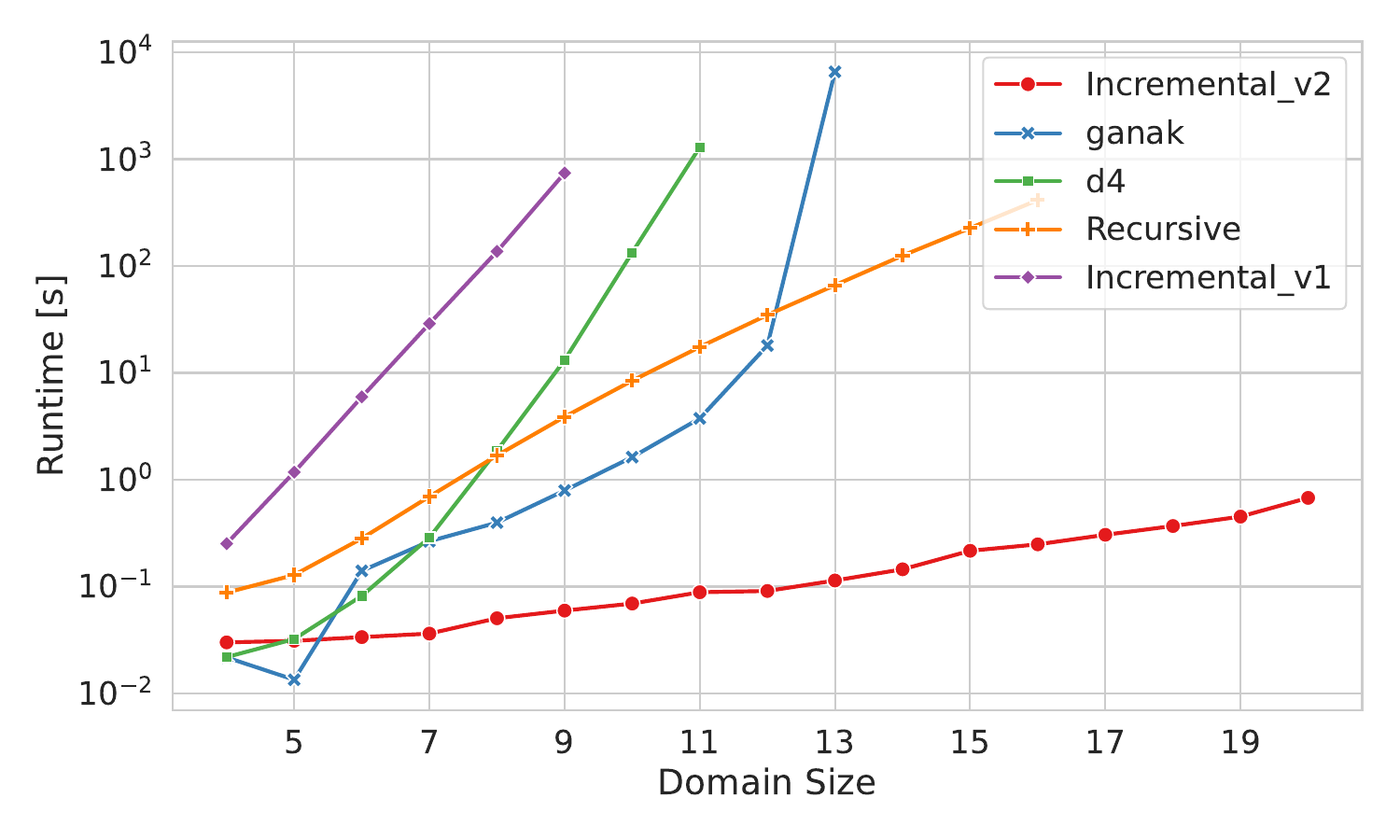}
        \caption{$\Phi_{hmm_2}$ for $n\le20$}
        \label{fig:hmm_small}
    \end{subfigure}
    \hfill
    \begin{subfigure}[b]{0.23\textwidth}
        \centering
        \includegraphics[width=\textwidth]{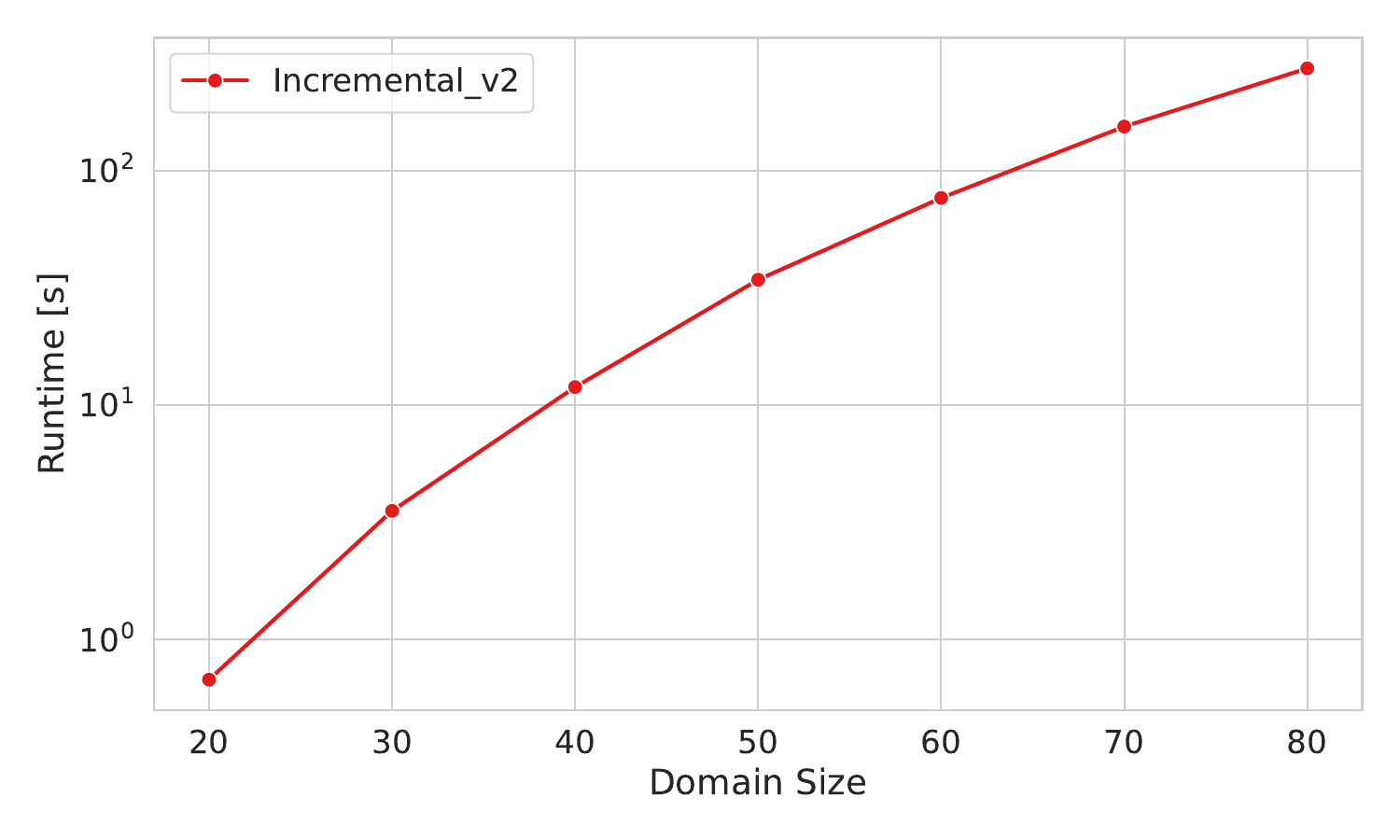}
        \caption{$\Phi_{hmm_2}$ for $n\le 80$}
        \label{fig:hmm_large}
    \end{subfigure}
    
    \caption{Runtime comparisons and WFOMC scaling behavior for $\Phi_{ws}$, $\Phi_{hmm}$ and $\Phi_{hmm_2}$.}
    \label{fig:lo-problems}
\Description{More runtime comparisons showing similar behavior of Algorithm 2 on the Watts-Strogatz and HMM problems.}
    
\end{figure}

\paragraph{Combinatorics Math Problems}
Finally, we evaluate performance on combinatorics problems involving counting permutations.
The problem set comes from the ``counting\_and\_statistics'' category of the MATH dataset~\citep{hendrycks21:math-dataset}.
We extracted a total of 28 combinatorics problems and encoded them using the linear order axiom along with its immediate and cyclic successor relations.
We refer readers to our source code, which includes a JSON file containing a natural-language description of each problem we selected.
For illustration, recall \cref{ex:basic-succ} asking, ``How many ways can we put 3 math books and 5 English books on a shelf if all the math books must stay together and all the English books must also stay together?'' which corresponds to problem 7.
Another example is problem 137, which asks ``In how many ways can we seat 8 people around a table if Alice and Bob will not sit next to each other?''

\cref{fig:comb_cactus_x1} shows the cactus plot for solving the combinatorics math problems.
While \cref{alg:iwfomc2} offers a substantial improvement in runtime compared to both \cref{alg:iwfomc} and RecursiveWFOMC, its advantage over the propositional counters is not as pronounced.
Therefore, we also provide \cref{fig:comb_cactus_x2,fig:comb_cactus_x3}, which show solving the same problems with domain sizes scaled up by a factor of two and three, respectively.
Apart from the cactus plots, concrete running time can be inspected for each problem in \cref{fig:comb_bar_x1,fig:comb_bar_x2,fig:comb_bar_x3}.

\begin{figure}[tb]
    \centering
    
    \begin{subfigure}[b]{0.23\textwidth}
        \centering
        \includegraphics[width=\textwidth]{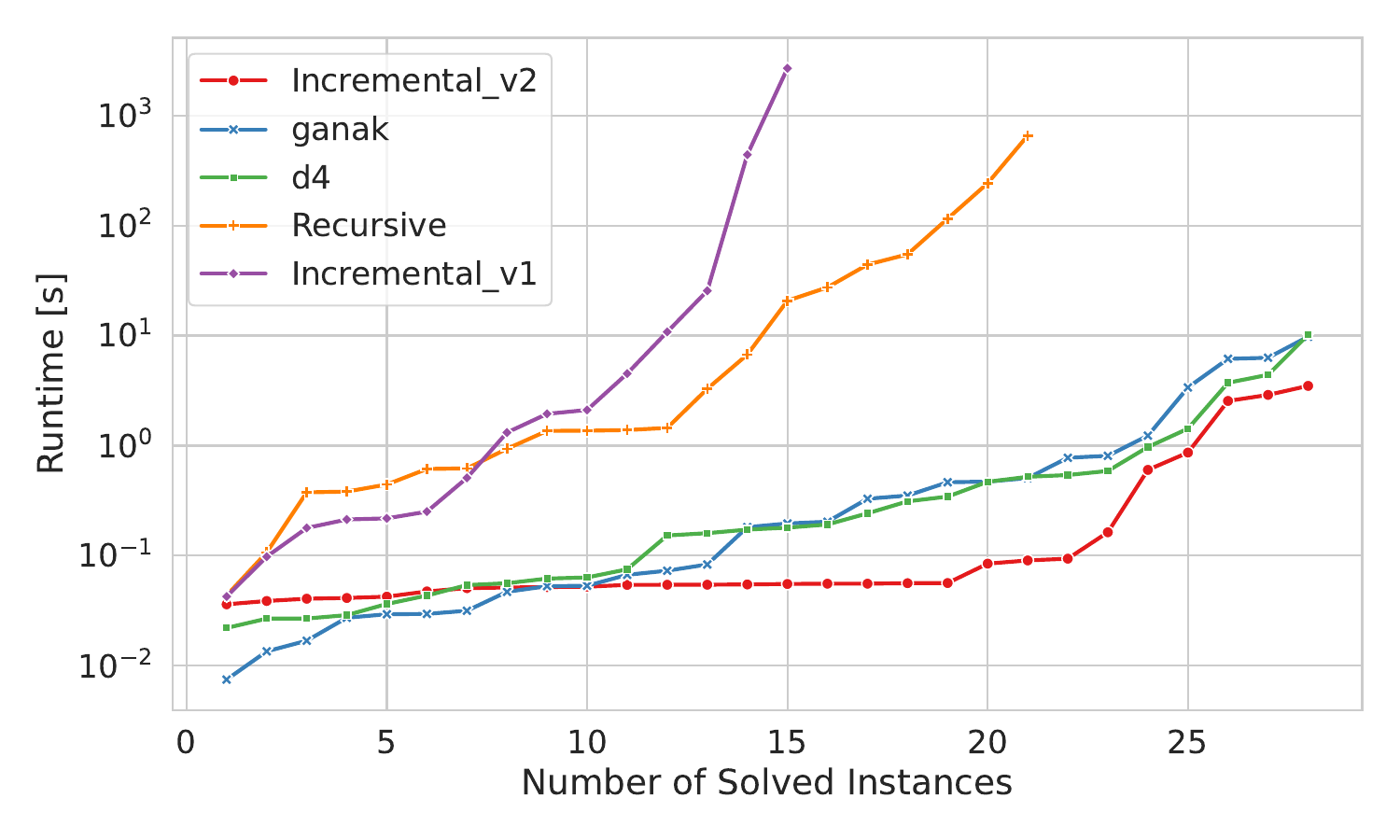}
        \caption{Original MATH problems}
        \label{fig:comb_cactus_x1}
    \end{subfigure}
    \hspace{3em}
    \begin{subfigure}[b]{0.23\textwidth}
        \centering
        \includegraphics[width=\textwidth]{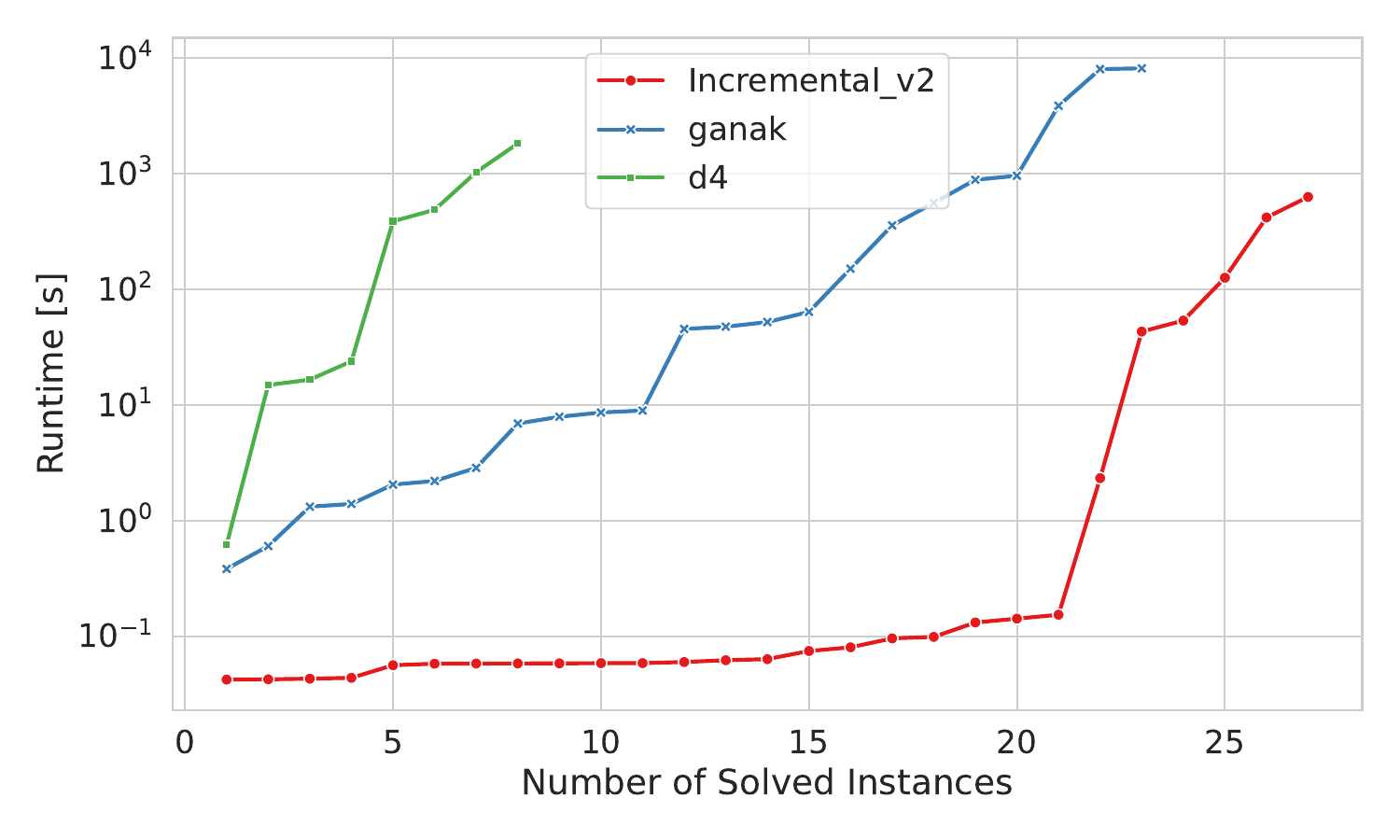}
        \caption{MATH problems with $2n$}
        \label{fig:comb_cactus_x2}
    \end{subfigure}
    \hspace{3em}
    \begin{subfigure}[b]{0.23\textwidth}
        \centering
        \includegraphics[width=\textwidth]{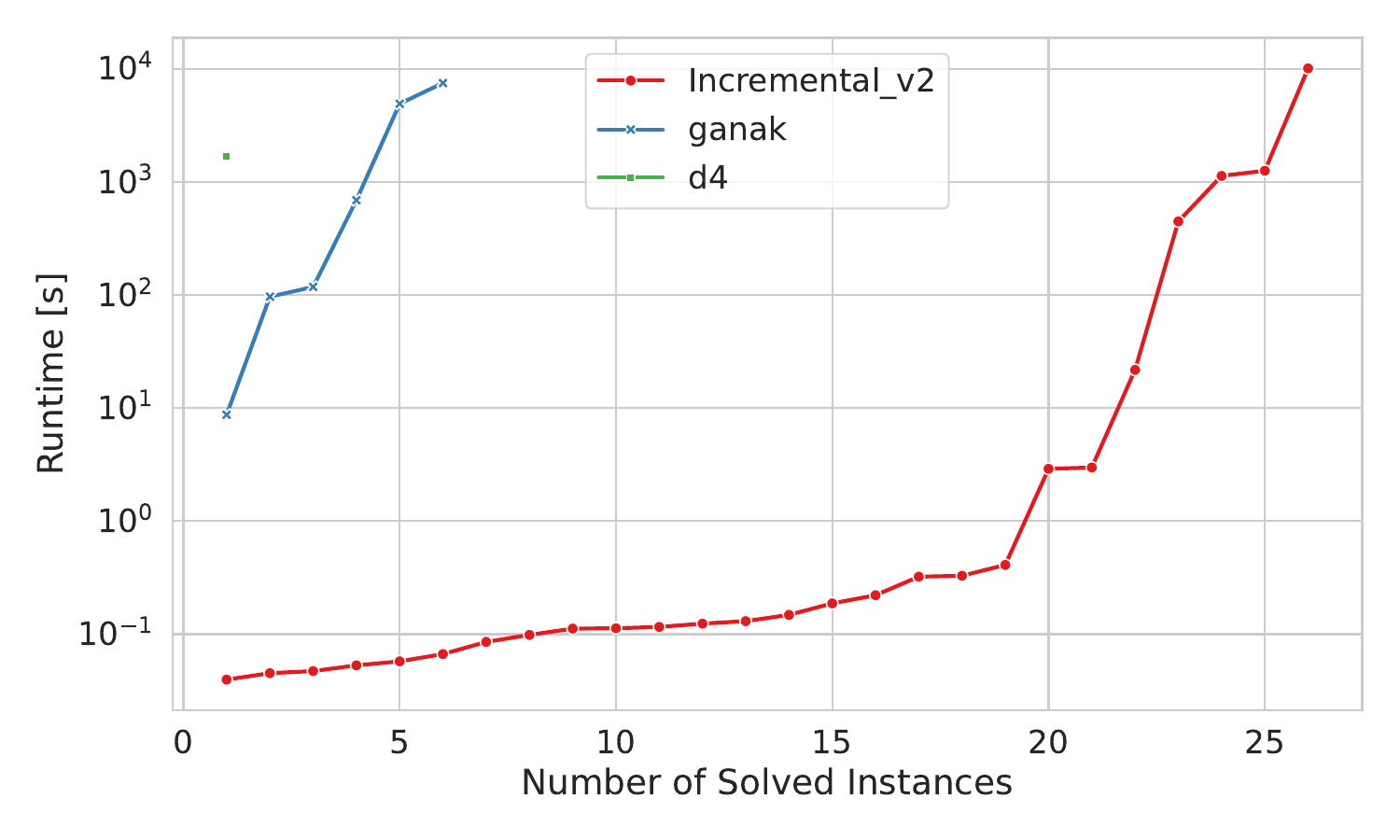}
        \caption{MATH problems with $3n$}
        \label{fig:comb_cactus_x3}
    \end{subfigure}

    \begin{subfigure}[b]{0.23\textwidth}
        \centering
        \includegraphics[width=\textwidth]{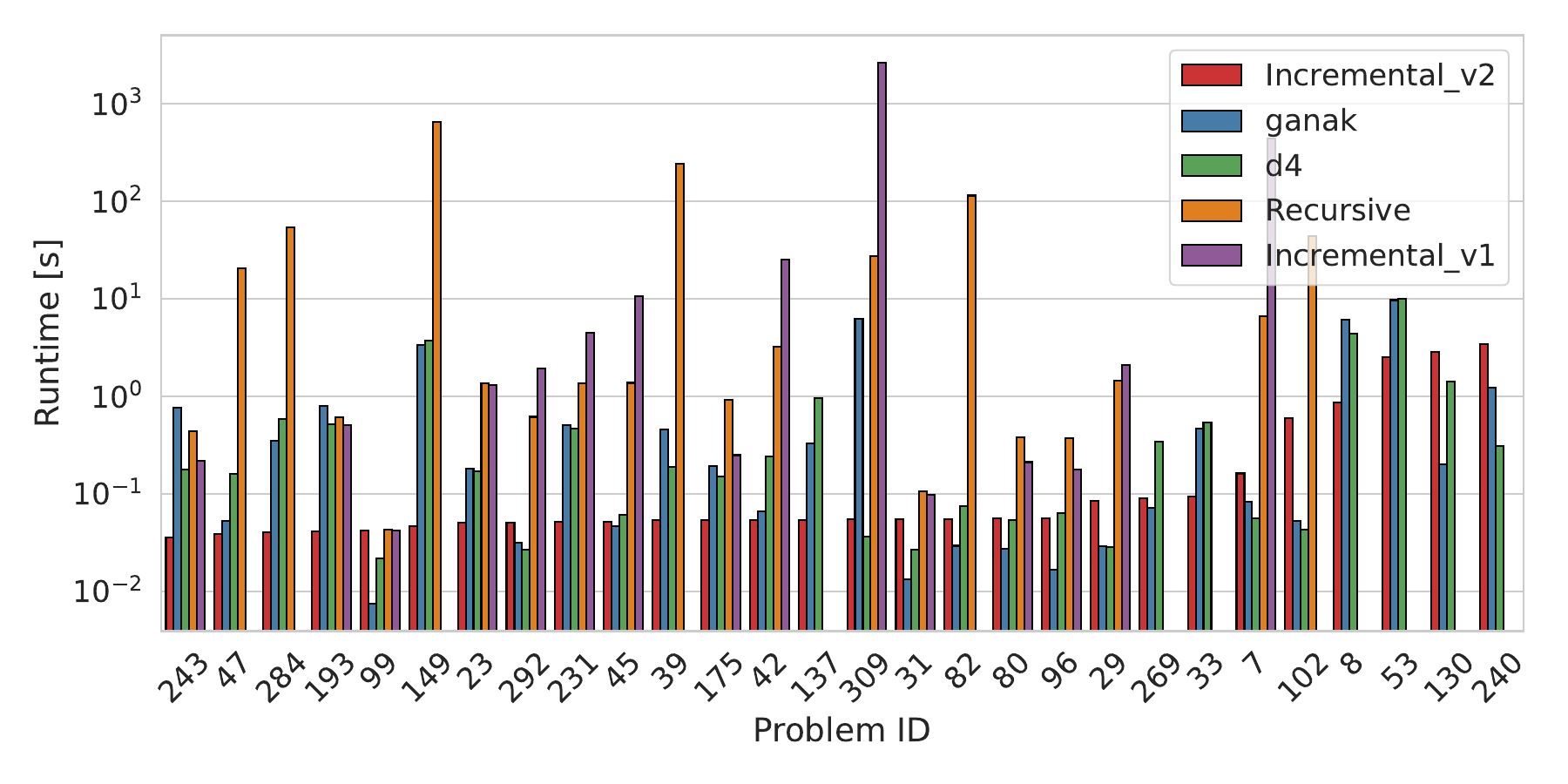}
        \caption{Runtimes on original MATH problems}
        \label{fig:comb_bar_x1}
    \end{subfigure}
    \hspace{3em}
    \begin{subfigure}[b]{0.23\textwidth}
        \centering
        \includegraphics[width=\textwidth]{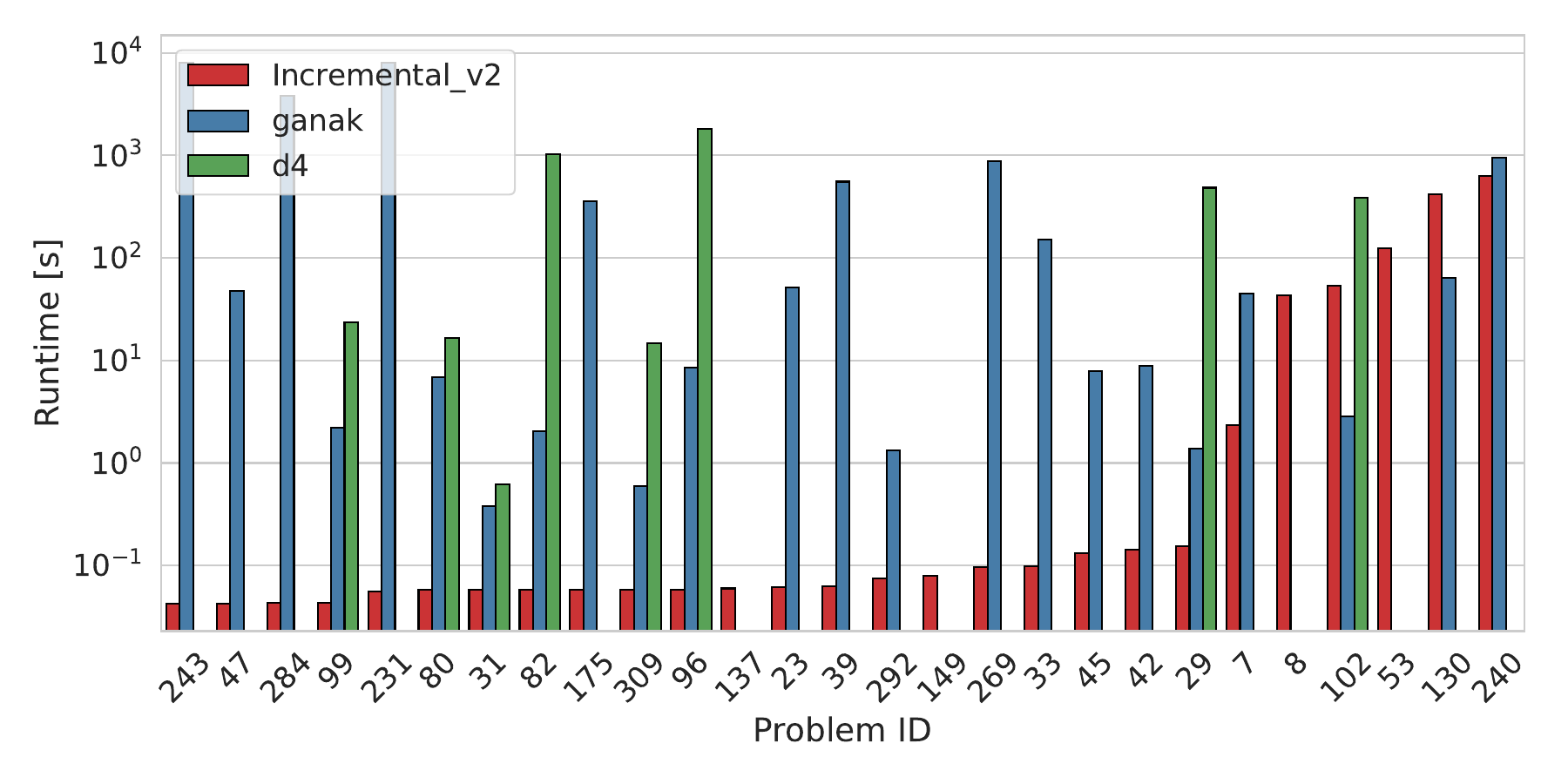}
        \caption{Runtimes on MATH problems with $2n$}
        \label{fig:comb_bar_x2}
    \end{subfigure}
    \hspace{3em}
    \begin{subfigure}[b]{0.23\textwidth}
        \centering
        \includegraphics[width=\textwidth]{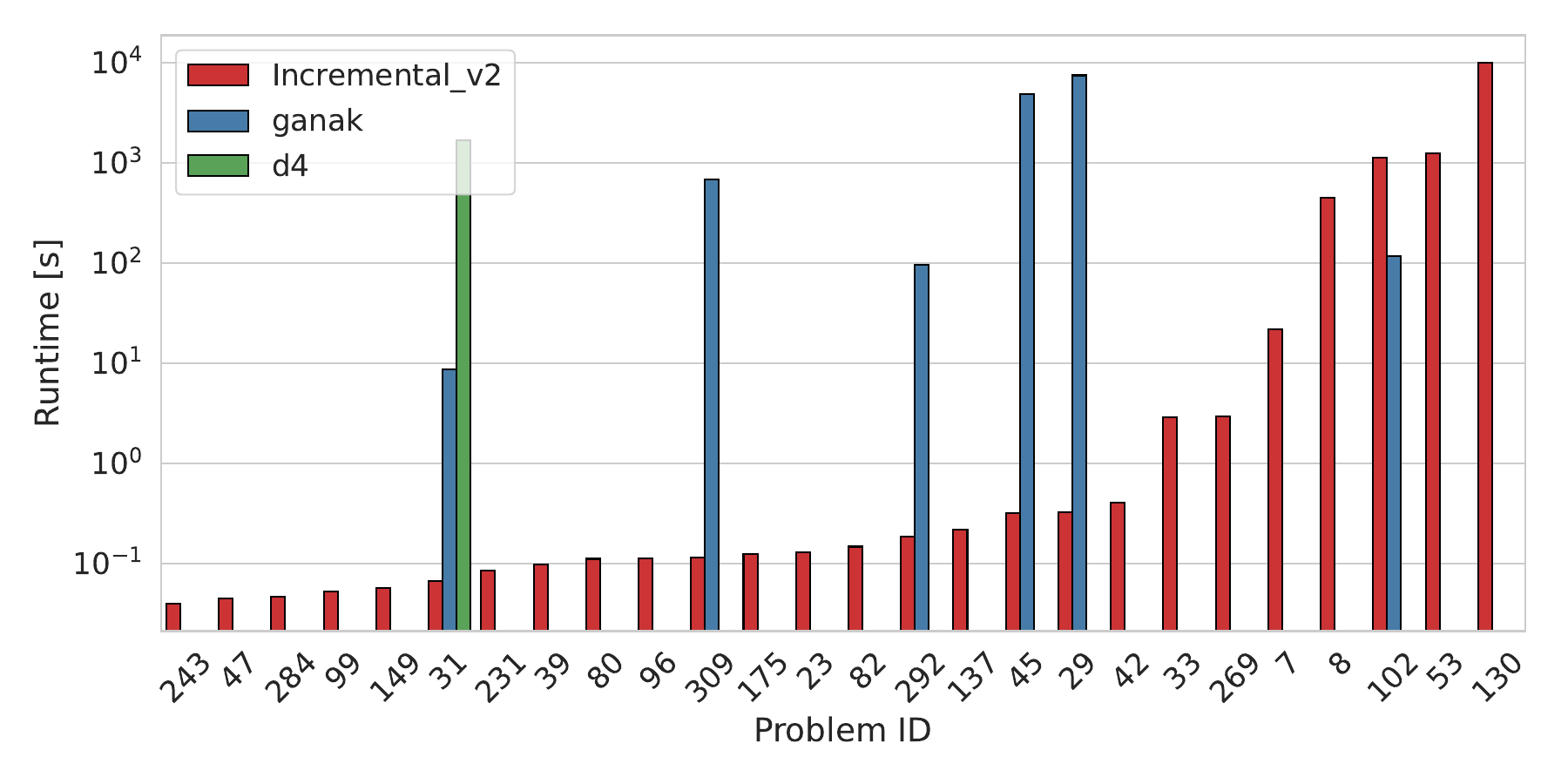}
        \caption{Runtimes on MATH problems with $3n$}
        \label{fig:comb_bar_x3}
    \end{subfigure}
    
    \caption{Cactus plots and runtime comparisons for problems from the MATH dataset for scaled-up domain sizes.}
    \label{fig:comb_cactus}

\Description{Runtime comparisons over the subset of the MATH dataset that we have extracted for counting purposes.
While the running times are similar for Algorithm 2 and the propositional counters on the original problems (the other lifted techniques are much worse already), the advantage of the lifted approach is quite pronounced when we enlarge the domain by a factor of two or even three.}
    
\end{figure}

\subsection{Problems with the Linear Order and Another Successor}
Unfortunately, we were unable to find any suitable problems for \cref{alg:lo+succ} (i.e., problems requiring both linear order and the successor relation of another linear order) in the MATH dataset; hence, we only construct our own simple benchmarking problems of increasing complexity.
We again choose the number of valid cells as our measure of problem complexity, since it still determines the degree of our polynomial runtime, as shown in the proof of \cref{thm:lo+succ}.
Whenever applicable for any problem, we match the sequence produced by computing WFOMC of that problem on domain sizes $n\in\{1,2,3,\ldots\}$ to the On-Line Encyclopedia of Integer Sequences \citep{oeis}, abbreviated as OEIS, to obtain a \emph{human-readable} description of the problem.
We consider five sentences, four of which are motivated solely by their varying numbers of valid cells, while the fifth is arguably more practical.
For each problem, it holds that $\le$ denotes the linear order relation as usual, and $S$ denotes the successor of another linear order, i.e., $\loaxiom(\le) \land \succaxiom(S)$ is implicitly conjoined to the problem.
The problems are as follows:
\begin{itemize}
    \item Sentence $\Phi_3$ has one valid cell, and if we fix the linear order, its sequence coincides with the OEIS entry A000670\footnote{\url{https://oeis.org/A000670}}, which asks for \emph{the number of ways to partition $n$ elements to disjoint subsets and arrange them into a sequence}.
    \begin{align*}
        \Phi_3 = \; &(\forall x \forall y: B(x, y) \to S(x, y))\\
        \land\; &(\forall x \forall y: (S(x, y) \wedge (x\le y)) \to B(x, y))
    \end{align*}

    \item $\Phi_4$ has two valid cells, and with fixed linear order, it corresponds to A000629\footnote{\url{https://oeis.org/A000629}}, i.e., \emph{the number of ways to partition $n+1$ elements to disjoint subsets and arrange them into a necklace}.
    \begin{align*}
        \Phi_4 = \forall x \forall y: (S(x, y) \land (x\le y)) \to (U(x) \leftrightarrow \lnot U(y))
    \end{align*}

    \item $\Phi_5$ has four valid cells, yet unfortunately, no matching entry in the OEIS. It first denotes the tail of the linear order by $U_1$. Second, if an element is in the tail, then its $S$-successor must be in the set $U_2$.
    \begin{align*}
        \Phi_5 = \; &(\forall x \forall y: ((U_1(x) \land (x\le y)) \to U_1(y))\\
        \land\; &(\forall x \forall y: ((U_1(x) \wedge S(x, y)) \to U_2(y))
    \end{align*}

    \item The sentence $\Phi_6$ is an augmentation of $\Phi_5$ with five valid cells. On top of already defined constraints, there is a binary relation $B(x,y)$ which, if true, requires $x$ to be in the set $U_1$ (the tail) and $y$ in $U_2$.
    \begin{align*}
    % \Phi_6 =  \;&\Phi_5 \land \forall x \forall y: B(x, y) \to (U_1(x) \wedge U_2(y))\\
        \Phi_6 = \; &(\forall x \forall y: ((U_1(x) \land (x\le y)) \to U_1(y))\\
        \land\; &(\forall x \forall y: ((U_1(x) \wedge S(x, y)) \to U_2(y))\\
        \land\; &(\forall x \forall y: B(x, y) \to (U_1(x) \wedge U_2(y)))
    \end{align*}

    \item Finally, $\Phi_{cards}$ counts \emph{the number of ways a deck of cards can be shuffled so that no two previously neighboring cards stay next to each other after the shuffling}. While $\Phi_{cards}$ has only one valid cell, it uses the $Succ_1$ relation rather than the linear order relation itself, and, from a practical standpoint, augmentations of the sentence could be used to answer various questions about the probability of card distributions after a deck has been shuffled.
    When fixing the linear order again, $\Phi_{cards}$ corresponds to the sequence A002464\footnote{\url{https://oeis.org/A002464}} in OEIS, which can be interpreted as Hertzsprung's problem, i.e., ways to arrange $n$ non-attacking kings on an $n \times n$ board, with one king in each row and column.
    \begin{align*}
        \Phi_{cards} = \; &(\forall x \forall y: S(x,y) \to \neg Succ_1(x,y))\\
        \land\; &(\forall x \forall y: S(x,y) \to \neg Succ_1(y,x))
    \end{align*}
\end{itemize}

Since there are no other lifted approaches that would support the axioms $\loaxiom(\le) \land \succaxiom(S)$, we compare our performance only to the propositional counters.
The comparison is shown in \cref{fig:lops}.
While \cref{alg:lo+succ} can easily scale to hundreds of domain elements on problems with a single valid cell, its scaling capabilities diminish significantly on more complex problems such as $\Phi_5$ and $\Phi_6$, barely scaling to twice the domain size of what WMC solvers can handle.
Nevertheless, our lifted algorithm still outperforms the propositional approaches on all tested instances and demonstrates a clear polynomial trend.

\begin{figure}[tbp]
    \centering

    \begin{subfigure}[b]{0.23\textwidth}
        \centering
        \includegraphics[width=\textwidth]{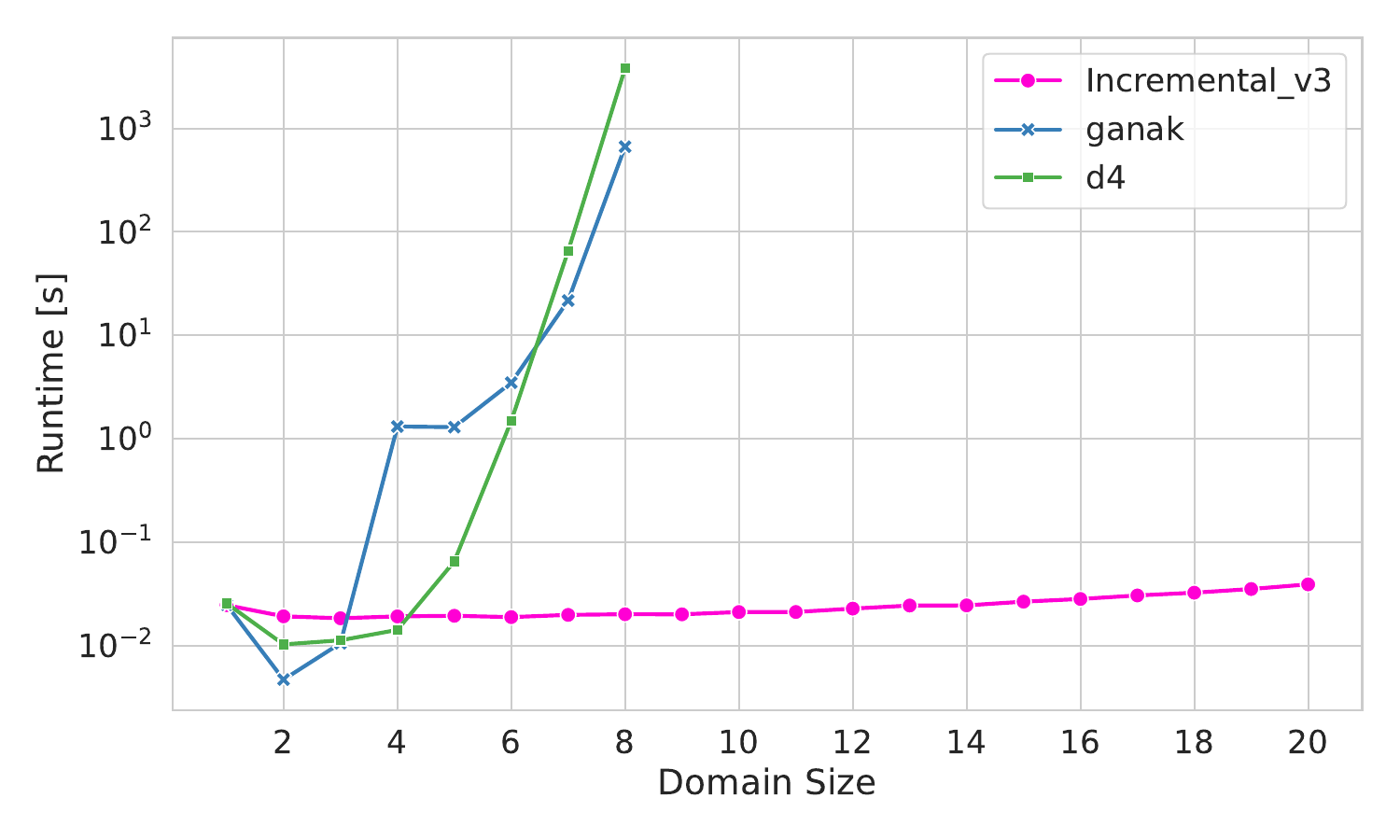}
        \caption{$\Phi_{cards}$ for $n\le20$}
        \label{fig:lops_cards}
    \end{subfigure}
    \hfill
    \begin{subfigure}[b]{0.23\textwidth}
        \centering
        \includegraphics[width=\textwidth]{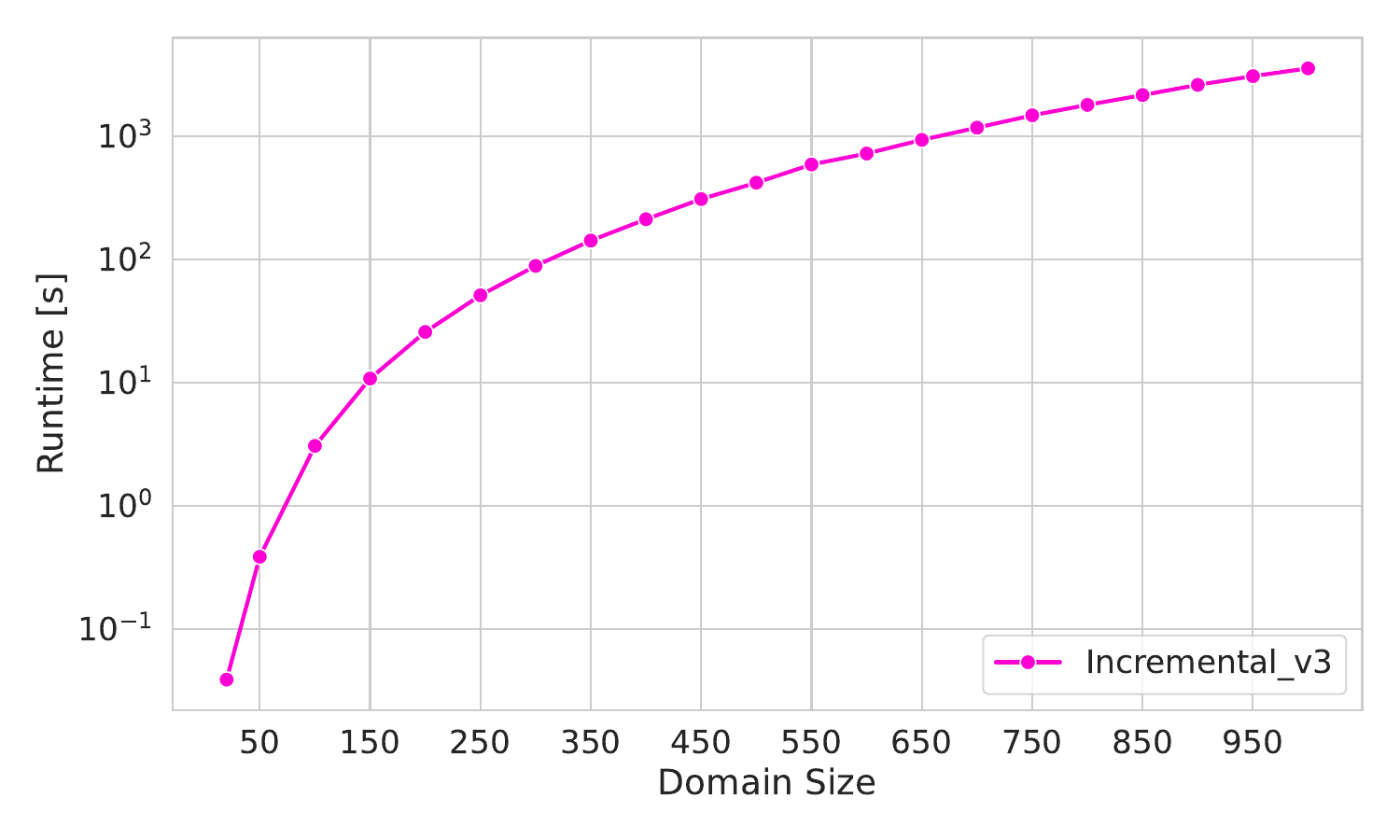}
        \caption{$\Phi_{cards}$ for $n\le1000$}
        \label{fig:lops_cards_large}
    \end{subfigure}
    \hfill
    \begin{subfigure}[b]{0.23\textwidth}
        \centering
        \includegraphics[width=\textwidth]{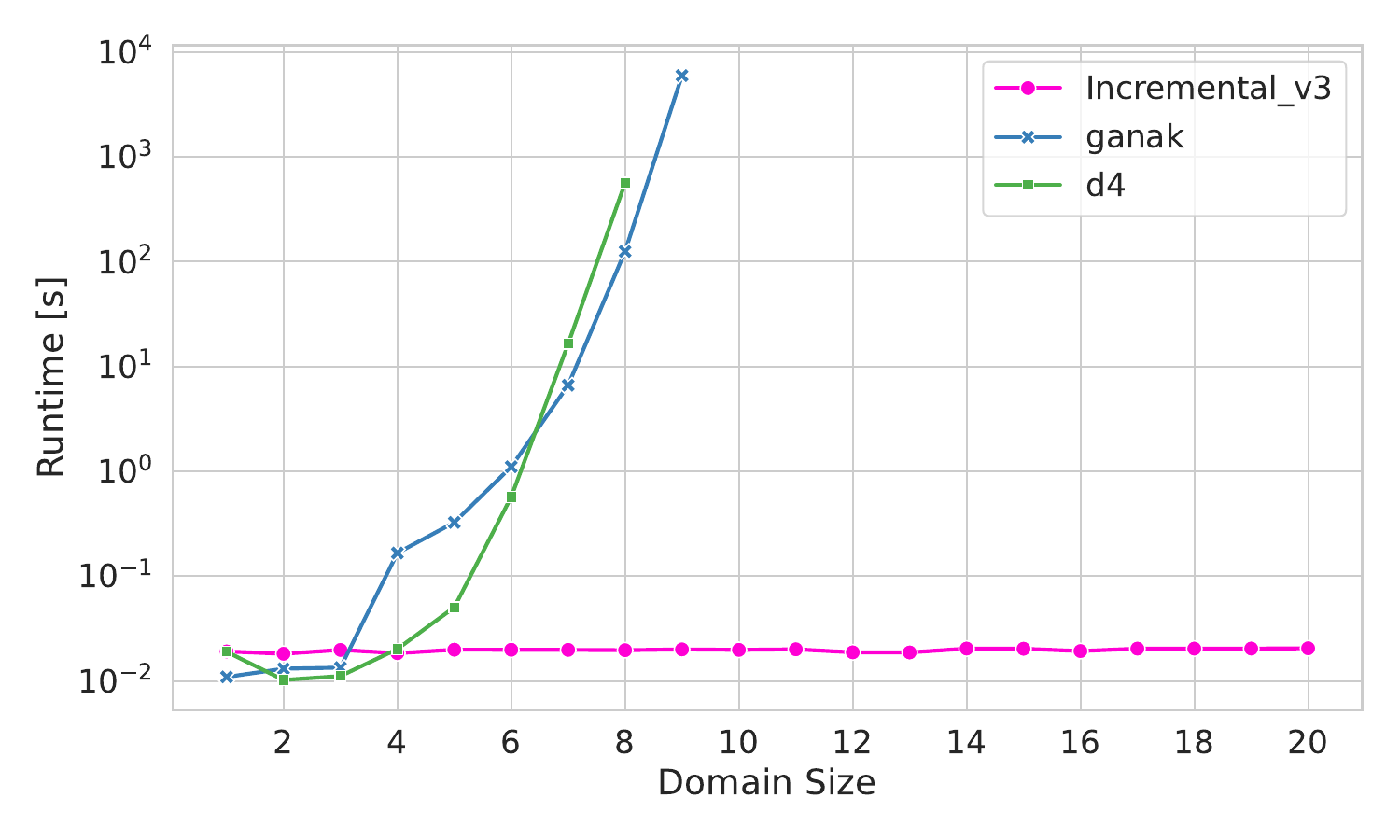}
        \caption{$\Phi_{3}$ for $n\le20$}
        \label{fig:lops_p1}
    \end{subfigure}
    \hfill
    \begin{subfigure}[b]{0.23\textwidth}
        \centering
        \includegraphics[width=\textwidth]{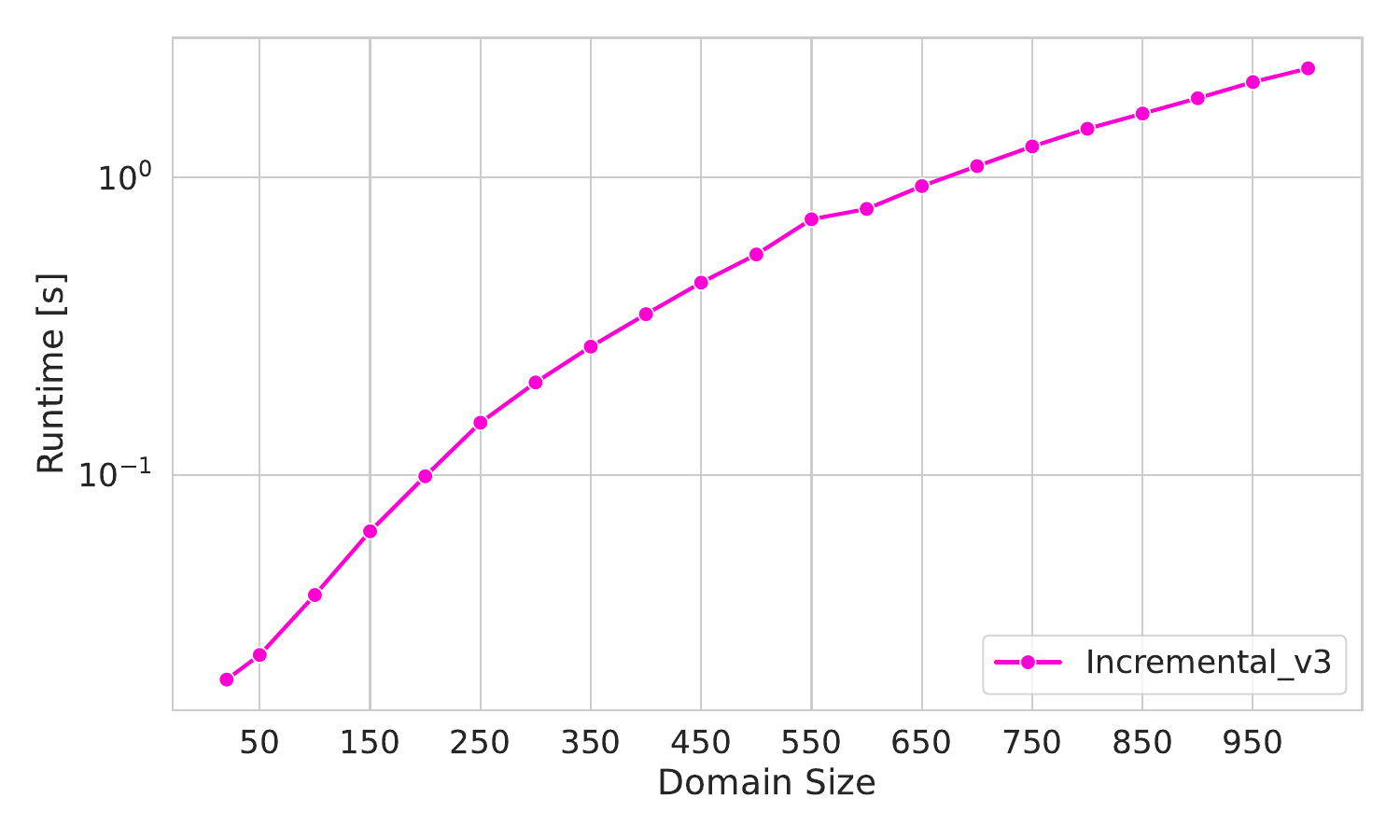}
        \caption{$\Phi_{3}$ for $n\le 1000$}
        \label{fig:lops_p1_large}
    \end{subfigure}

        \begin{subfigure}[b]{0.23\textwidth}
        \centering
        \includegraphics[width=\textwidth]{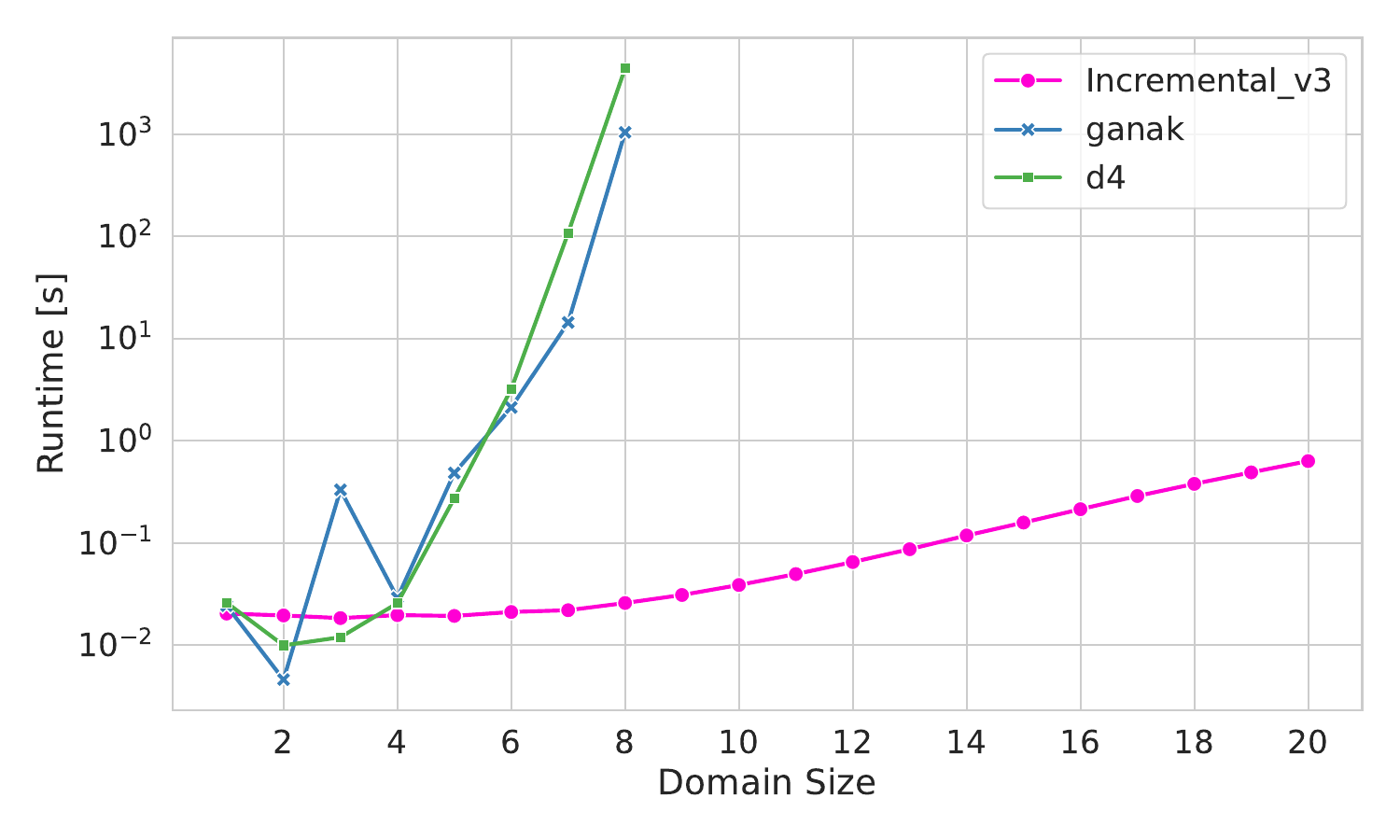}
        \caption{$\Phi_{4}$ for $n\le20$}
        \label{fig:lops_p2}
    \end{subfigure}
    \hfill
    \begin{subfigure}[b]{0.23\textwidth}
        \centering
        \includegraphics[width=\textwidth]{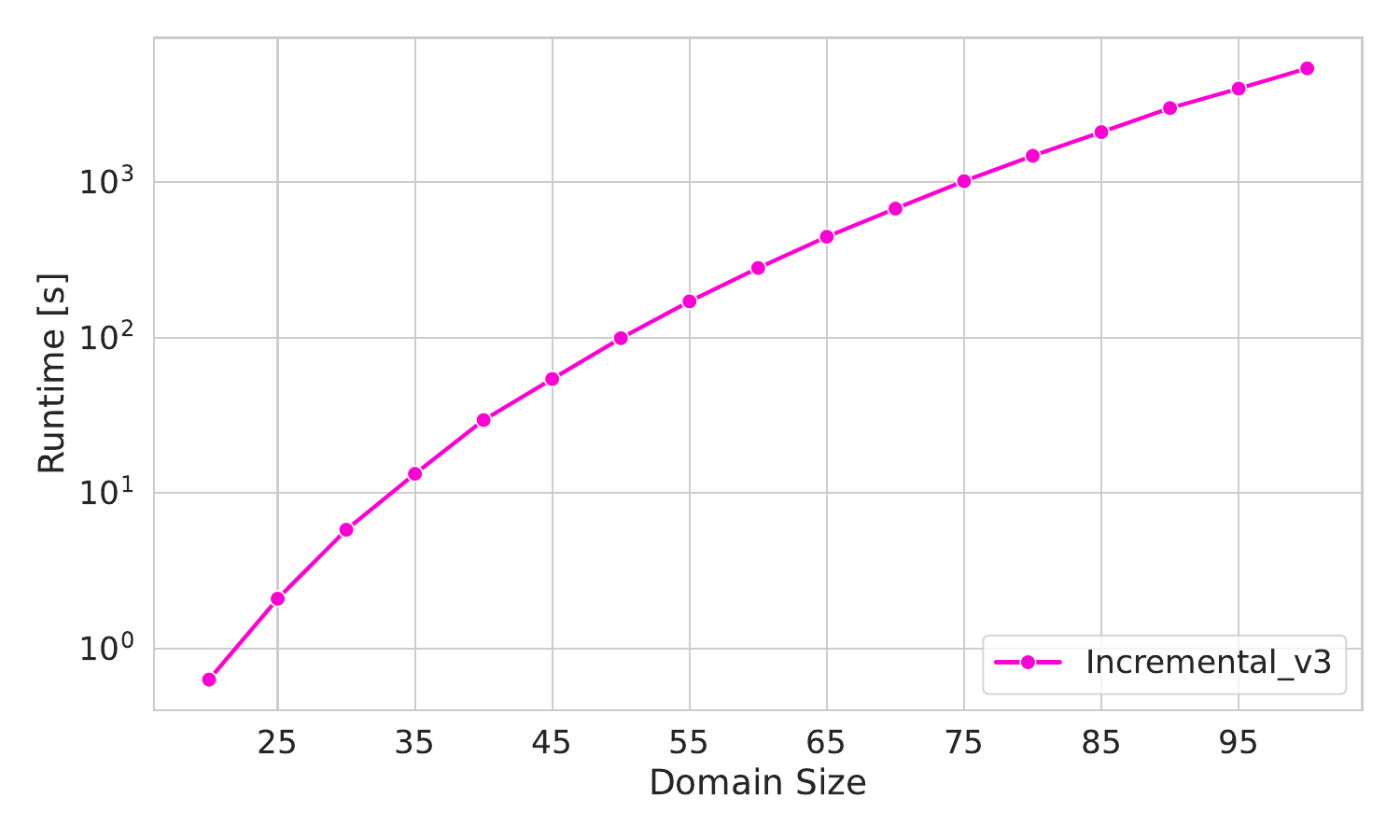}
        \caption{$\Phi_{4}$ for $n\le100$}
        \label{fig:lops_p2_large}
    \end{subfigure}
    \hfill
    \begin{subfigure}[b]{0.23\textwidth}
        \centering
        \includegraphics[width=\textwidth]{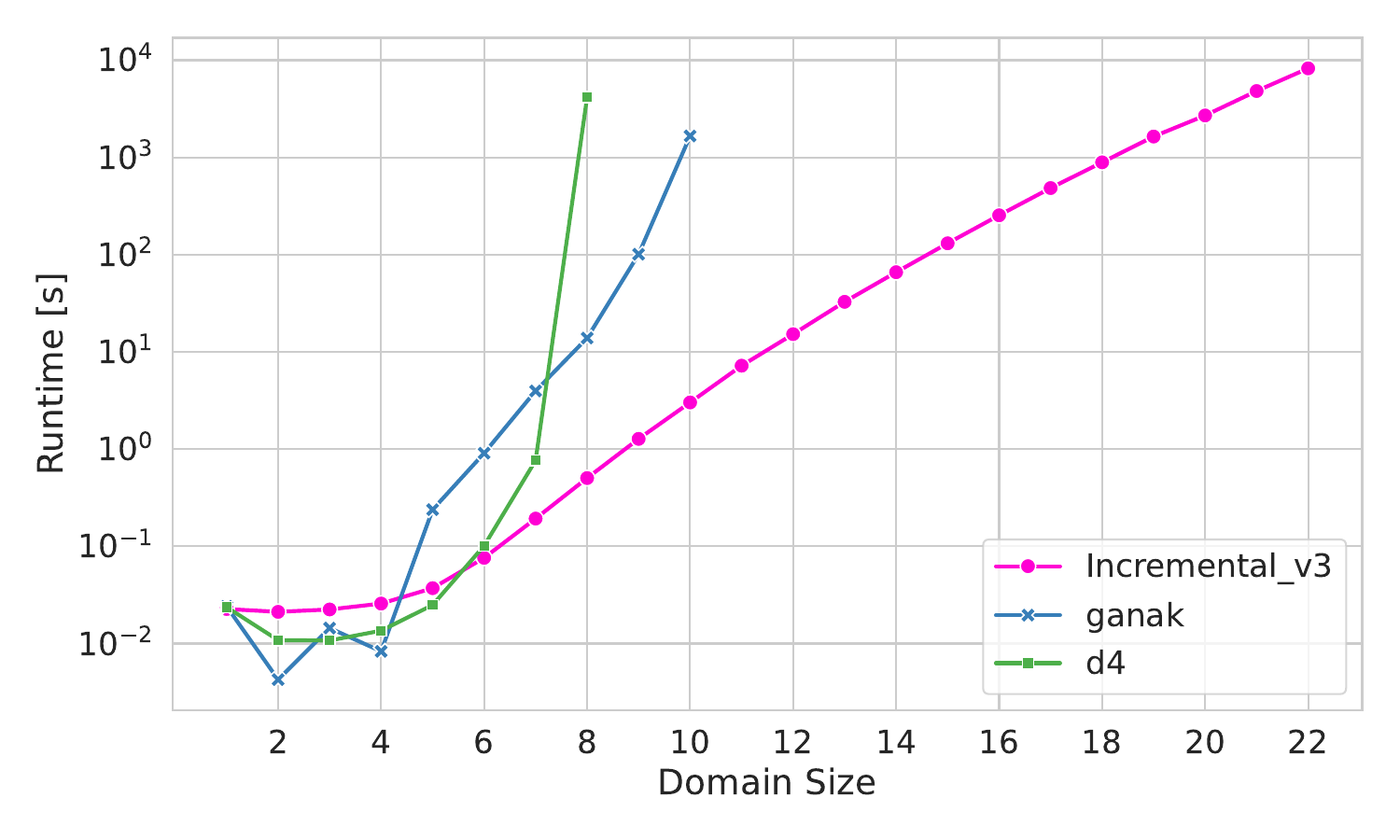}
        \caption{$\Phi_{5}$ for $n\le22$}
        \label{fig:lops_p4}
    \end{subfigure}
    \hfill
    \begin{subfigure}[b]{0.23\textwidth}
        \centering
        \includegraphics[width=\textwidth]{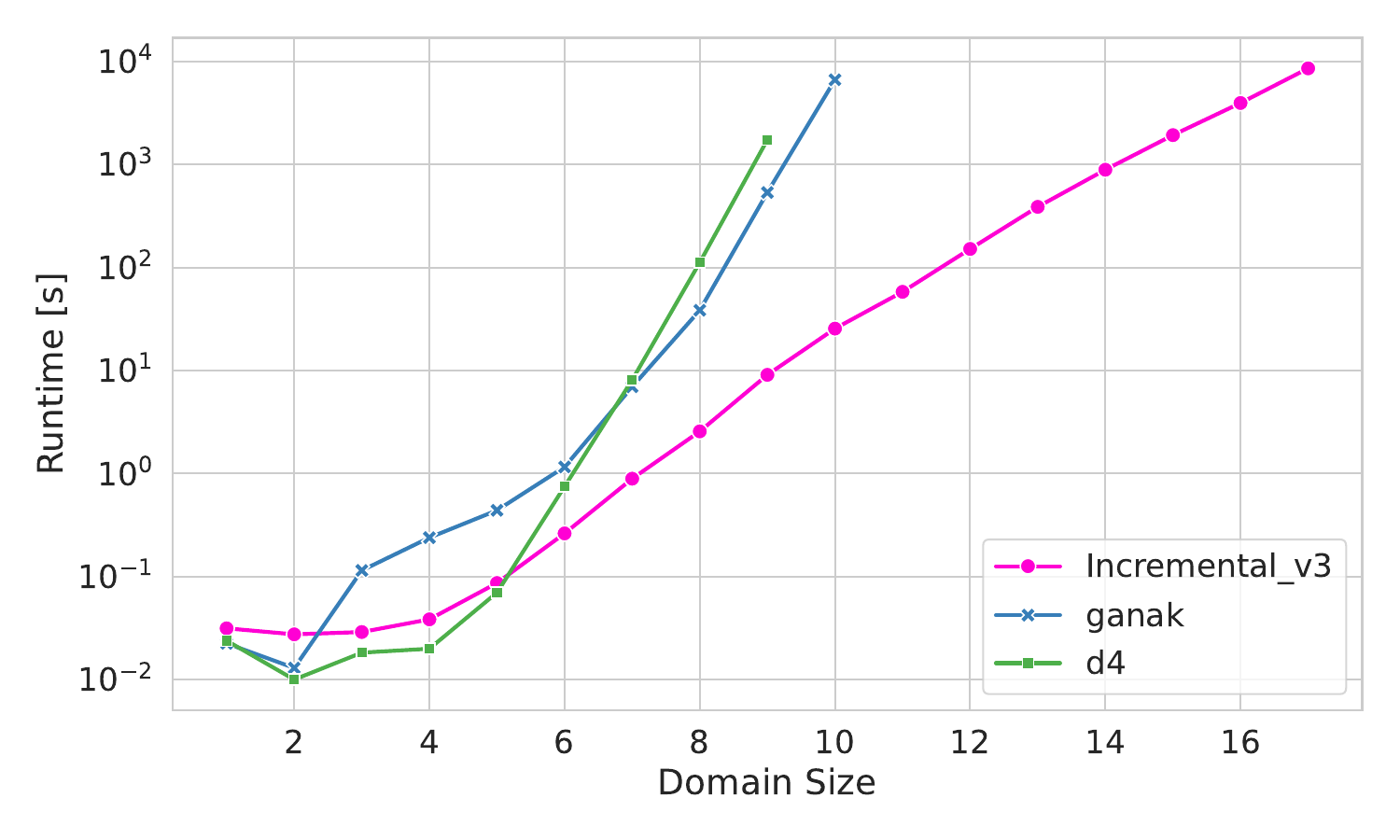}
        \caption{$\Phi_{6}$ for $n\le 17$}
        \label{fig:lops_p5}
    \end{subfigure}
    
    \caption{Runtime comparisons and scaling behavior for problems with $\loaxiom$ and $\succaxiom$}. %$\Phi_{3}$, $\Phi_{4}$, $\Phi_{5}$, $\Phi_{6}$ and $\Phi_{cards}$ }
    \label{fig:lops}

\Description{Runtime measurements for the problems with the linear order axiom and the successor axiom.
We compare Algorithm 3 to GANAK and d4.
The largest instances the propositional counters can solve are for domains of size 8-10. At the same time, the lifted approach goes much further for problems with at most two valid cells and still almost twice as far for problems with four or five valid cells.
The polynomial tendency of Algorithm 3 can also be observed.
}
    
\end{figure}

% The resulting runtimes can be found in Figure \ref{fig:execution}. In all cases, the logarithm of the running time exhibits a bend as the domain size $n$ grows. This behavior provides evidence that the runtime of our algorithm is indeed polynomial in $n$. Moreover, our algorithm scales better than weighted model counters when the domain size is large.

\section{Conclusion}
\label{sec:8-coclusion}
In this work, we investigated the boundary of domain-liftability for Weighted First-Order Model Counting over ordered domains.
We introduced a family of algorithms based on the domain recursion rule, which we dub \emph{IncrementalWFOMC}.

We first established that the two-variable fragment of first-order logic, extended with a single linear order axiom, is domain-liftable.
To support this, we presented the base version of IncrementalWFOMC (\cref{alg:iwfomc}), which computes the weighted model count over an ordered domain in time polynomial in the domain size.

Our experimental results, however, revealed that while domain-liftability guarantees polynomial asymptotic complexity, it does not strictly guarantee superior runtime in practice.
As observed in \cref{fig:ws2_small}, unoptimized lifted algorithms can sometimes be outperformed by state-of-the-art propositional weighted model counters on smaller domains.
Hence, while domain-liftability is essential for scaling to larger domains, practical usability often requires handling specific structural constraints, such as the successor relations, natively.

Addressing this, we extended the framework to support general successor relations, developing \cref{alg:iwfomc2}.
This version natively supports the linear order axiom along with its successor relations up to the $k$-th successor.
By avoiding the overhead of explicit successor encoding, the algorithm demonstrated significant practical superiority, scaling to domains orders of magnitude larger than those solvable by propositional approaches.

To explore the theoretical limits of tractability, we introduced a second linear order.
We proved that WFOMC for \fotwo with two linear order relations is \class{\#P_1}-hard.
This negative result also implies the intractability of other axioms that can encode linear orders, such as the acyclicity axiom on two distinguished relations.

To bridge the gap between these positive and negative results, we analyzed a \emph{hybrid} setting featuring one full linear order and access to the successor relation of a second, unknown linear order.
We proved that this setting is domain-liftable by providing \cref{alg:lo+succ} to solve it.
Therefore, the tractability boundary for ordered domains lies strictly beyond the setting where the immediate successor relation of the second linear order is known.

Notably, our findings align with established results on the decidability of finite satisfiability.
For instance, \citet{manuel10:unsat-fo2+2successor} proved that \fotwo with two successor relations is decidable, whereas
\fotwo with two linear order relations and their successor relations is undecidable.
Our results imply that WFOMC with only two successor relations is domain-liftable, while we establish that strengthening both successor axioms to linear order axioms leads to \class{\#P_1}-hardness.
Although the techniques for hardness proofs in the two problems are not identical, both essentially rely on encoding a grid and embedding a hard problem (e.g., the tiling problem) using first-order fragments.
A natural question arises: is there a unified way to transform the undecidability of finite satisfiability to the \class{\#P_1}-hardness of WFOMC, or conversely, mapping decidability to domain-liftablity?

Apart from the question above, future work should focus on refining the tractability boundary, either by characterizing the intractability threshold directly or, at least, by narrowing the identified gap.
Finally, given the practical utility of the IncrementalWFOMC family for combinatorial counting and probabilistic inference, further improvements, including an optimized implementation, should also be investigated.

\section*{Acknowledgments}
Jan T\'{o}th and Ond\v{r}ej Ku\v{z}elka were supported by the Czech Science Foundation project 23-07299S \emph{Statistical Relational Learning in Dynamic Domains}.
Kuncheng Zou and Yuanhong Wang were supported by National Natural Science Foundation of China (No.62506141).
V\'{a}clav K\r{u}la was supported by the Central Europe Leuven Strategic Alliance (CELSA) project \emph{Towards Scalable Algorithms for Neuro-Symbolic AI}.
Yuyi Wang was supported by the Natural Science Foundation of Hunan (Grant No. 2024JJ5128).

%%
%% The next line prints the references.
\bibliographystyle{plainnat} % Or another author-year BST like 'agsm' or 'apalike'
\bibliography{refs}    % Base name of your .bib file (without .bib)

\appendix

\newpage
\section{Handling Unary Evidence via Cardinality Constraints}
\label{app:evidence-by-ccs}
\rev{
To solve WFOMC involving unary evidence (in the presence of the linear order axiom), the given evidence can be logically transformed into cardinality constraints \citep{wang24:wfoms}.

Given a set of ground unary literals, we can categorize domain elements based on their \emph{evidence type}, i.e., a set of literals on a single variable consistent with evidence specified for the domain element.
Note that the total number of evidence types is bounded by a constant with respect to $n$.
Each relation may appear in an evidence type positively, negatively or be missing.

\begin{example}
Consider the domain $\Delta = [3] = \{1,2,3\}$ and unary evidence $\epsilon = \{P(1), \neg Q(1), R(2)\}$.

For this domain, the evidence $\epsilon$ yields three distinct evidence types, one for each domain element.
Since $\evidence{1} = P(1) \land \neg Q(1)$, the evidence type of the domain element 1 is $\sigma^1 = \{P(x), \neg Q(x)\}$, or logically, $P(x) \land \neg Q(x)$.
Analogously, we have $\sigma^2 = \{R(x)\}$ since $\evidence{2} = R(2)$.
Finally, $\evidence{3} = \top$ and $\sigma^3 = \emptyset$.

Assuming $\preds{\epsilon} = \{P,Q,R\}$ is the entire vocabulary, there is a total of $3^{|\preds{\epsilon}|}=3^3$ possible evidence types.
\end{example}

Assume the evidence yields $m$ distinct evidence types, denoted as $\sigma^1, \sigma^2, \dots, \sigma^m$, and let $n_i$ represent the exact number of domain elements realizing each evidence type $\sigma^i$.
Introduce a fresh, auxiliary unary predicate $\xi^i$ for each distinct evidence type $\sigma^i$.
The input sentence is then conjoined with
$$\left(\bigwedge_{i\in[m]} (\forall x: \xi^i(x) \to \sigma^i(x))\right) \land \left(\forall x: \bigvee_{i\in[m]}\xi^i(x)\right) \land \left(\bigwedge_{1\le i<j\le n} (\forall x:\neg \xi^i(x) \lor \neg \xi^j(x))\right).$$
Every element in the domain must satisfy exactly one $\xi^i$ predicate.
If $\xi^i(x)$ holds true for an element, the corresponding evidence formula $\sigma^i(x)$ must hold true as well.
Now, the specific unary evidence can be replaced by a global cardinality constraint of $(|\xi^i| = n_i)$ for each evidence type.

Denote $\Gamma$ the original sentence, $\epsilon$ the unary evidence and $\Gamma'$ the new sentence. 
Then it holds that
$$\symbwfomc(\Gamma\land\epsilon,n,\weights) = \frac{1}{\binom{n}{n_1.n_2,\ldots,n_m}}\cdot\symbwfomc(\Gamma'\land\bigwedge_{i\in[m]}(|\xi^i|=n_i),n,\weights).$$
The scaling factor accounts for all possible symmetric partitions of the domain, since the evidence $\epsilon$ specifies exact domain elements having the evidence type $\sigma^i$, not just that $n_i$ elements satisfy it.

For complete technical details, complexity analysis, and a formal proof of soundness regarding this reduction, we refer the readers to \citet[Appendix A]{wang24:wfoms}.
}

\section{Encoding Successor Relations using Counting Quantifiers and Linear Order Axiom}
\label{app:successor-by-lo}
% \begin{itemize}
%   \item The smallest element in the order
%   \begin{equation}
%     \forall x \left( First(x) \leftrightarrow \forall y \ L(x,y) \right).
%   \end{equation}
%   \item The largest element in the order
%   \begin{equation}
%     \forall x \left( Last(x) \leftrightarrow \forall y \ L(y,x) \right).
%   \end{equation}
%   \item The successor relation of the order
%   \begin{equation}
%     \begin{aligned}
%       & \forall x (\lnot Last(x) \to \exists_{=1} y \ Succ_1(x,y)) \\
%       \land & \forall x (\lnot First(x) \to \exists_{=1} y \ Succ_1(y,x)) \\
%       \land & \forall x \forall y (Succ_1(x,y) \to L(x,y)),
%     \end{aligned}
%   \end{equation}
% \end{itemize}

% \begin{equation}
%     \begin{aligned}
%         &\forall x: \neg Perm(x,x) \land \\
%         &\forall x\exists_{=1} y: Perm(x,y) \land \forall y\exists_{=1} x: Perm(x,y) \land\\
%         &\forall x\forall y: Pred_1(x,y)\Rightarrow Perm(x,y)\land \\
%         &\forall x\forall y: Pred_1(x,y)\Rightarrow (x \le y)\land\\
%         &|Pred_1| = n-1,
%     \end{aligned}
%     \label{eq:old_pred1}
% \end{equation}

Suppose we have a linear order relation $\lopred$.
Let us consider a possible world $\omega^\lopred$ containing only ground atoms that define the linear order.
We may think of any successor relation of $\lopred$ as a subset of $\omega^\lopred$, and we may use \ctwo sentences to capture that particular subset.

\subsection{The Immediate Successor Relation}
Let us start by capturing the immediate successor relation $Succ_1$.%
\footnote{While we have already demonstrated how to express such a subset using \Cref{eq:succ1:first-last,eq:succ1:bijection-like,eq:succ1:left2right}, we repeat the construction here to keep the section self-contained.}
For example, if we work with the natural order $1 \leq 2 \leq 3 \leq \ldots \leq n$, we will want to capture the set $\{Succ_1(1, 2), Succ_1(2, 3), \ldots, Succ_1(n-1,n)\}\subset \omega^\lopred$.
We may express such a subset (for an arbitrary domain ordering) in the following way:
\begin{itemize}
    \item Find the first and the last element.
    \begin{equation}
    \label{eq:app:succ1-first-last}
        \begin{aligned}
            &(\forall x : First(x) \leftrightarrow \forall y: \ (x\lopred y)) \\
            \land\; &(\forall x : Last(x) \leftrightarrow \forall y: \ (y\lopred x))
        \end{aligned}
    \end{equation}

    \item Each element except the last one has exactly one successor, and each element except the first one has exactly one predecessor.
    \begin{equation}
    \label{eq:app:succ1-bijection-like}
    \begin{aligned}
        &(\forall x : \lnot Last(x) \to \exists^{=1} y: \ Succ_1(x,y)) \\
        \land\; &(\forall x : \lnot First(x) \to \exists^{=1} y: \ Succ_1(y,x))
    \end{aligned}
    \end{equation}

    \item The successor relation only goes along the edges $G(\lopred)$, i.e., it only goes from smaller elements to greater ones.
    \begin{equation}
    \label{eq:app:succ1-left2right}
        \forall x \forall y : Succ_1(x,y) \to (x < y)
    \end{equation}
\end{itemize}

Denote the conjunction of \Cref{eq:app:succ1-first-last,eq:app:succ1-bijection-like,eq:app:succ1-left2right} as $\Psi_{Succ_1}$.

\begin{lemma}
    Sentence $\Psi_{Succ_1}$ defines the binary relation $Succ_1$ as the (immediate) successor relation with respect to the linear order relation $\lopred$.
\end{lemma}
\begin{proof}
    Without loss of generality, let us only consider the natural ordering $1 \leq 2 \leq \ldots \leq n$ for the domain $[n]$.
    When arguing about particular atoms of $Succ_1$ being true or false, let us think in terms of the graph $G(Succ_1)$ and talk about edges.
    Let us also think of the domain (the vertices) as a sequence ordered from left to right in ascending order.
    Due to \Cref{eq:app:succ1-left2right}, the edges of $G(Succ_1)$ can only go from left to right.
    
    We start with the edge cases.
    For the first element 1 (i.e., the atom $First(1)$ is true), there must be exactly one outgoing edge (due to \Cref{eq:app:succ1-bijection-like}).
    Suppose that edge leads to a vertex $j\ge3$.
    Then there cannot be any incoming edge to the vertex 2 (also due to \Cref{eq:app:succ1-bijection-like}), which is a contradiction.
    Hence, we must have $Succ_1(1,2)$.
    Similarly, for the last element $n$, as there must be exactly one outgoing edge from $n-1$ and $n$ is a necessary target, $Succ_1(n-1,n)$ holds.

    For a vertex at a general position in the sequence, suppose we have the edge $Succ_1(i,i+k)$ where $k\ge2$.
    Then there must be an edge outgoing from one of the $\{2,\ldots,i-1\}$ vertices into the vertex $i+1$.
    There must be another edge going from the same set $\{2,\ldots,i-1\}$ to the vertex $i$.
    Hence, there are $i-3$ edges left to connect the elements $\{2,\ldots,i-1\}$ in a manner consistent with \Cref{eq:app:succ1-bijection-like}.
    However, that is not possible.
    Therefore, we must have the edge $Succ_1(i,i+1)$.
\end{proof}

Note that there are other ways to capture the immediate successor, especially if we also allow cardinality constraints, e.g.,
\begin{equation}
\label{eq:app:old_succ1}
    \begin{aligned}
        &(\forall x: \neg Perm(x,x)) \\
        \land\;&(\forall x\exists^{=1} y: Perm(x,y))\\
        \land\; &(\forall y\exists^{=1} x: Perm(x,y)) \\
        \land\; &(\forall x\forall y: Succ_1(x,y)\to Perm(x,y)) \\
        \land\; &(\forall x\forall y: Succ_1(x,y)\to (x \lopred y))\\
        \land\; &(|Succ_1| = n-1).
    \end{aligned}
\end{equation}
The relation $Perm/2$ must be a derangement (a bijection without fixed points) with exactly $n-1$ \emph{left-to-right} transitions.
Hence, $Perm$ coincides with the previously defined cyclic successor relation $CySucc$, and since $Succ_1$ is merely a copy of $Perm$ excluding the edge $Perm(n,1)$, the successor relation is properly defined.
The encoding from \Cref{eq:app:old_succ1} may lead to more efficient computation since it only introduces one new predicate $Perm$ instead of two, i.e., $First$ and $Last$.
However, from a theoretical standpoint, both encodings increase runtime only by a constant factor with respect to the domain size.

\subsection{A General Successor Relation}
Now, let us try to generalize $\Psi_{Succ_1}$ above to the case of the $k$-th successor relation ($Succ_k$), i.e., for a constant $k$, we want to identify pairs of elements (of the linear order) that have exactly $(k-1)$ elements in between.

In this case, we will require more auxiliary relations than just two.
Let $Pos_j(x)$ denote that the element $x$ is at the $j$-th position in the sequence.
Also consider $k$ distinct colors denoted by unary relations $C_1, C_2,\ldots, C_k$.
Last but not least, we will leverage the relation $Succ_1$, which we already consider to be defined using one of the encodings above.
Now, we proceed as follows:
\begin{itemize}
    \item We start by identifying the first $k$ and the last $k$ elements of the sequence.
    \begin{equation}
    \label{eq:app:succk-pos}
        \begin{aligned}
            &\forall x : \left( Pos_1(x) \leftrightarrow \forall y \ (x\lopred y) \right)\\
            \land\;&\forall x : \left( Pos_n(x) \leftrightarrow \forall y \ (y\lopred x) \right)\\
            \land\;&\bigwedge_{i=1}^{k-1} \forall x \forall y : Pos_{i}(x) \land Succ_1(x, y) \to Pos_{i+1}(y)\\
            \land\;&\bigwedge_{n-k+2}^{n}\forall x \forall y : ( Pos_{i}(x) \land Succ_1(y, x)) \to Pos_{i-1}(y)
        \end{aligned}
    \end{equation}

    \item Next, we color each element such that the first element receives the first color, the second element the second color, and so on, until the $k$-th element with the $k$-th color. Then, the sequence of colors repeats. %it holds $C_1(1), C_2(2),\ldots,C_k(k),C_1(k+1),C_2(k+2),\ldots,C_k(2k),\ldots$
    \begin{equation}
    \label{eq:app:succk-colors}
    \begin{aligned}
        &\left(\forall x : \bigvee_{i=1}^{k} C_i(x)\right)\\
        \land\;&\bigwedge_{1\leq i < j \leq k}(\forall x : \lnot C_i(x) \lor \lnot C_j(x))\\
        \land\;&(\forall x : Pos_1(x) \to C_1(x))\\
        \land\;&\bigwedge_{i=1}^{k-1}(\forall x \forall y : (C_i(x) \land Succ_1(x, y)) \to C_{i+1}(y))\\
        \land\;&(\forall x \forall y : (C_k(x) \land Succ_1(x, y)) \to C_1(y))
    \end{aligned}
\end{equation}

    \item The successor relation must again only go from left to right, and only elements of the same color can be connected.
    \begin{equation}
    \label{eq:app:succk-left2right}
        \begin{aligned}
            &(\forall x \forall y : Succ_k(x, y) \to (x < y))\\
            \land\; &(\forall x \forall y : Succ_k(x, y) \to \bigvee_{i=1}^{k} \left( C_i(x) \land C_i(y) \right))
        \end{aligned}
    \end{equation}

    \item Finally, we again ensure the \emph{bijection-like} behavior excluding the first $k$ or the last $k$ elements when appropriate.
    \begin{equation}
    \label{eq:app:succk-bijection-like}
        \begin{aligned}
            &\left(\forall x : \left(\bigwedge_{i=n-k+1}^n \lnot Pos_{i}(x)\right) \to \exists^{=1} y \ Succ_k(x,y)\right)\\
            \land\;&\left(\forall x : \left(\bigwedge_{i=1}^k \lnot Pos_{i}(x)\right) \to \exists^{=1} y \ Succ_k(y,x)\right)
        \end{aligned}
    \end{equation}
\end{itemize}

Denote the conjunction of \Cref{eq:app:succk-pos,eq:app:succk-colors,eq:app:succk-left2right,eq:app:succk-bijection-like} as $\Psi_{Succ_k}$.
\begin{lemma}
    Sentence $\Psi_{Succ_k}$ defines the binary relation $Succ_k$ as the $k$-th successor relation with respect to the linear order relation $\lopred$.
\end{lemma}
\begin{proof}
    The key observation is that, due to \Cref{eq:app:succk-left2right}, we only connect domain elements of the same color that have clearly $k-1$ elements in between.
    Thus, we have $k$ identical subproblems, each of which can be proven by the same argument we used for $Succ_1$.
\end{proof}

\section{Encoding a Grid using General Successor Relations}
\label{app:grid-by-successors}

The basic idea of encoding a $k\times n$ grid is very similar to the encoding in the proof of \cref{lemma:grid-2lo}.
Let there be a linear order relation $\lopred$.
We will again use the immediate successor relation to represent the vertical adjacency $V/2$.
For the horizontal adjacency $H/2$, however, we will leverage the information of a bounded \emph{height} $k$ and use the $k$-th successor relation.
Both successor relations will be defined with respect to the same linear order as $\lopred$.
The corner cases will again be handled by introducing auxiliary predicates $Left/1$, $Right/1$, $Top/1$, and $Bottom/1$.

\begin{itemize}
    \item First, start by denoting the first and last elements, respectively.
    \begin{equation}
    \label{eq:app:grid-first-last}
        \begin{aligned}
            &(\forall x : First(x) \leftrightarrow \forall y: \ (x\lopred y)) \\
            \land\; &(\forall x : Last(x) \leftrightarrow \forall y: \ (y\lopred x))
        \end{aligned}
    \end{equation}

    \item Then, we can define the top and bottom elements in the grid using the $k$-th successor relation.
    \begin{equation}
    \label{eq:app:grid-top-bottom}
        \begin{aligned}
            & (\forall x: First(x) \to Top(x)) \\
            \land\; &(\forall x\forall y: (Top(x) \land Succ_k(x,y)) \to Top(y)) \\
            \land\; &(\forall x: Last(x) \to Bottom(x)) \\
            \land\; &(\forall x\forall y: (Bottom(x) \land Succ_k(y,x)) \to Bottom(y))
        \end{aligned}
    \end{equation}

    \item The cardinality constraints are used to ensure that there are exactly $n$ elements in the top and bottom rows.
    \begin{equation}
    \label{eq:app:grid-n-cols}
        (|Top| = n) \land (|Bottom| = n) 
        % \begin{aligned}
        %     &(|Top| = n) \\
        %     \land\; &(|Bottom| = n)      
        % \end{aligned}
    \end{equation}

    \item With the top and bottom elements defined, we can easily define the vertical adjacency relation $V/2$.
    \begin{equation}
    \label{eq:app:grid-vertical}
        \begin{aligned}
            \forall x\forall y: V(x,y) \leftrightarrow \left(\neg Top(y)\land \neg Bottom(x) \land Succ_1(x,y)\right).
        \end{aligned}
    \end{equation}

    \item The left and right boundaries of the grid can be defined similarly.
    \begin{equation}
    \label{eq:app:grid-left-right}
        \begin{aligned}
            &(\forall x: First(x) \to Left(x))\\
            \land\;&(\forall x\forall y: (Left(x) \land V(x,y)) \to Left(y))\\
            \land\;&(\forall x: Last(x) \to Right(x)) \\
            \land\;&(\forall x\forall y: (Right(x) \land V(y,x)) \to Right(y)) \\
            \land\;&(|Left| = k) \land (|Right| = k)\\
            % \land\;&(|Right| = k)
        \end{aligned}
    \end{equation}

    \item Finally, the horizontal adjacency relation $H/2$ can be defined as well.
    \begin{equation}
    \label{eq:app:grid-horizontal}
        \begin{aligned}
          \forall x\forall y: H(x,y) \leftrightarrow \left(\neg Left(y)\land \neg Right(x) \land Succ_k(x,y)\right).
            \end{aligned}
    \end{equation}
\end{itemize}

See \cref{fig:app:grid-kxn} for a demonstration of how the grid encoding works.
The blue edges represent the immediate successor relation, of which we omit the \emph{diagonal} edges when defining $V/2$ in \Cref{eq:app:grid-vertical}.
The red edges then show the $k$-th successor relation, which coincides with the horizontal adjacency $H/2$.

% Denote the conjunction of sentences in \Cref{eq:app:grid-first-last,eq:app:grid-top-bottom,eq:app:grid-n-cols,eq:app:grid-vertical,eq:app:grid-left-right,eq:app:grid-horizontal} by $\Psi_{k,n-grid}$.
% Note that $\Psi_{k,n-grid}$ is in the \fotwo{} fragment with cardinality constraints.
% By \cref{lemma:c2+cc}, we can reduce \wfomc over $\Psi_{k,n-grid}$ to \wfomc over \fotwo.
% By \Cref{thm:general_linear_order}, the \wfomc{} of the sentence with an additional grid constraint can be computed in time polynomial in the domain size $n$ if the height $k$ of the grid is bounded by a constant.

\begin{figure}[htb]
  \centering
        \begin{tikzpicture}
            \tikzstyle{roundnode}=[circle, draw, inner sep=0pt, minimum size=2mm]
            
            % Parameters
            \def\n{8}   % number of columns
            \def\k{4}   % number of rows
            \def\xstep{1.0}
            \def\ystep{1.0}
            
            % Draw nodes
            \foreach \i in {1,...,\n}{
                \foreach \j in {1,...,\k}{
                      \node[roundnode] (v\i\j) at ({\xstep*\i},{-\ystep*\j}) {};
                }
            }
            
            % Horizontal red arrows
            \foreach \j in {1,...,\k}{
                \foreach \i [evaluate=\i as \ip using int(\i+1)] in {1,...,\numexpr\n-1}{
                      \path[->, red] (v\i\j) edge (v\ip\j);
            }
            }
            
            % Vertical blue arrows (upward)
            \foreach \i in {1,...,\n}{
                \foreach \j [evaluate=\j as \jp using int(\j+1)] in {1,...,\numexpr\k-1}{
                      \path[<-, blue] (v\i\jp) edge (v\i\j);
            }
            }
            
            % Diagonal blue arrows (upright)
            \foreach \i [evaluate=\i as \j using int(\i+1)] in {1,...,\numexpr\n-1}{
                \path[->, blue] (v\i\k) edge (v\j1);
            }
            
            % Braces and labels
            \draw[decorate, decoration={brace, amplitude=6pt}] 
            ($(v11)+(-0.1,0.2)$) -- ($(v\n1)+(0.1,0.2)$)
            node[midway, yshift=10pt]{$n$};
            \draw[decorate, decoration={brace, amplitude=6pt}] 
            ($(v1\k)+(-0.2,-0.1)$) -- ($(v11)+(-0.2,-0.1)$)
            node[midway, xshift=-12pt]{$k$};
            
            % Border names
            \node[] (Top) at (4.5,0) {\small $Top$} ;
            \node[] (Bottom) at (4.5,-4.5) {\small $Bottom$} ;
            \node[] (Left) at (-0.3,-2.63) {\small $Left$};
            \node[] (Right) at (8.7,-2.63) {\small $Right$};
              
        \end{tikzpicture}
\caption{A $k\times n$ grid expressed using the immediate successor relation (blue) and the $k$-th successor relation (red).}
\label{fig:app:grid-kxn}
\Description{The figure shows domain elements represented as vertices formed into a grid with k rows and n columns. The grid is enforced by a vertical adjacency, representing the immediate successor relation, going downwards between lines and diagonally from the last line to the first, and a horizontal adjacency, going left-to-right between columns, representing the k-th successor relation.
The relations are used to define the Top and Bottom rows, as well as the Left and Right columns.}
\end{figure}
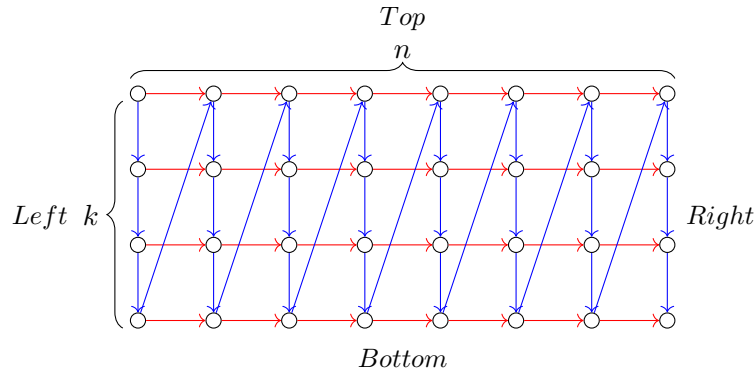

\section{Markov Logic Networks}
\label{app:mlns}

Markov Logic Networks \citep{richardson06:mlns}, often abbreviated as MLNs, are a popular model in statistical relational learning offering a straightforward integration of first-order logic and probability distributions.
An MLN $\Phi$ is a set of weighted first-order logic formulas (possibly with free variables) with weights taking on values from the real domain or infinity:
$$\Phi = \Set{(w_1, \alpha_1), (w_2, \alpha_2), \ldots, (w_k, \alpha_k)}$$
Given a domain $\Delta$, the MLN defines a probability distribution over possible worlds such as
\begin{align*}
    % Pr_{\Phi, \Delta}(\omega) = \frac{\llbracket \omega \models \Phi_{\infty} \rrbracket}{Z} \exp\left(\sum_{(w_i, \alpha_i) \in \Phi_{\real}} w_i \cdot N(\alpha_i, \omega) \right)
    Pr_{\Phi, \Delta}(\omega) = 
\begin{cases} 
    \frac{1}{Z} \exp\left(\sum_{(w_i, \alpha_i) \in \Phi_{\real}} w_i \cdot N(\alpha_i, \omega) \right) & \text{if } \omega \models \Phi_{\infty}, \\
    0 & \text{otherwise.}
\end{cases}
\end{align*}
where $\Phi_{\real}$ denotes formulas with real-valued weights (soft constraints), $\Phi_{\infty}$ denotes formulas with infinity-valued weights (hard constraints), $Z$ is the normalization constant ensuring valid probability values, and $N(\alpha_i, \omega)$ is the number of substitutions to free variables of $\alpha_i$ that produce a grounding of those free variables that is satisfied in $\omega$.
The distribution formula is equivalent to that of a Markov Random Field \citep{koller09:pgms}.
Hence, an MLN, along with a domain, defines a probabilistic graphical model, and inference in the MLN is thus inference over that model.

Inference (and also learning) in MLNs is reducible to \wfomc{} \citep{broecketal14:wfomc-skolem}.
For each $(w_i, \alpha_i(\bm{x}_i)) \in \Phi_\real$, introduce a new formula $\forall \bm{x}_i: \xi_i(\bm{x}_i) \leftrightarrow \alpha_i(\bm{x}_i)$, where $\xi_i$ is a fresh predicate, $w(\xi_i) = \exp(w_i), \overline{w}(\xi_i) = 1$ and $w(Q) =  \overline{w}(Q) = 1$ for all other predicates $Q$.
Hard formulas are added to the theory as additional conjuncts.
Denoting the new theory by $\Gamma$ and a query by $\phi$, we can perform the inference as
\begin{align*}
    Pr_{\Phi, \Delta}(\phi) = \frac{\wfomc(\Gamma \wedge \phi, |\Delta|, w, \negw)}{\wfomc(\Gamma, |\Delta|, w, \negw)}.
\end{align*}

\begin{example}
    Recall $\Phi_{hmm}$ from \cref{sec:7-experiments}, i.e.,
\begin{align*}
    \Phi_{hmm} = \{ \;
    &(\infty, \forall x: \neg (Sn(x) \land Rn(x))), \\
    &(\infty, \forall x: Sn(x) \lor Rn(x)), \\
    &(\infty, (Sn(x) \land Succ_1(x,y)) \to Rn(x)), \\
    &(0.5, S(x) \to Rn(x)), \\
    &(1.0, (S(x) \land Succ_1(x,y))\to S(y)), \\
    &(0.4, (S(x) \land Succ_2(x,y))\to S(y)), \\
    &(0.1, (S(x) \land Succ_3(x,y))\to S(y))\; \}.
\end{align*}

Transforming $\Phi_{hmm}$ into a theory processable by WFOMC gives us
\begin{align*}
\Gamma_{hmm} = \; 
    &(\forall x:\neg (Rn(x) \land Sn(x)) \land(Rn(x) \lor Sn(x))) \\
    \land\; &(\forall x \forall y: (Sn(x)\land Succ_1(x,y)) \to Rn(x)) \\
    \land\; &(\forall x \forall y: \xi_0(x) \leftrightarrow (S(x) \to Rn(x)))\\
    \land\; &(\forall x \forall y: \xi_1(x) \leftrightarrow (S(x) \land Succ_1(x,y)) \to S(y))\\
    \land\; &(\forall x \forall y: \xi_2(x) \leftrightarrow (S(x) \land Succ_2(x,y)) \to S(y))\\
    \land\; &(\forall x \forall y: \xi_3(x) \leftrightarrow (S(x) \land Succ_3(x,y)) \to S(y)),
\end{align*}
with weights set such that $w(\xi_i) = \exp(w_i)$ and $\negw(\xi_i)=1$ for all $i\in[4]$
and $w(P) = \negw(P) = 1$ for all other predicates $P\in\preds{\Gamma_{hmm}}$.
Then the partition function of $\Phi_{hmm}$ over the domain $[n]$ is equal to
$$\symbwfomc(\Gamma_{hmm}\land Linear(\le, Succ_1, Succ_2, Succ_3), n, \weights).$$
\end{example}

\section{Experiments Details}
In this section, we provide details about our experiments omitted in \cref{sec:7-experiments}.

\subsection{Combinatorics Problems}
When performing experiments on the subset of problems from the MATH dataset \citep{hendrycks21:math-dataset}, we have scaled the problems by factors of 2 and 3 to clearly demonstrate the superiority of \cref{alg:iwfomc2} over all other tested approaches.
We provide \cref{tab:math_domain_sizes} so that any interested reader may assess what problem sizes (denoted $n$), as well as how complex problems (in terms of the number of valid cells $p$), we have been solving.
A natural-language specification for each problem is available in our source code repository.%
\footnote{Available at \url{https://github.com/jan-toth/wfomc-over-ordered-domains}.}

Note problem 193, which we chose not to scale up for our experiments.
That is a choice based entirely on the problem semantics, since the problem asks ``How many nine-digit numbers can be made using each of the digits 1 through 9 exactly once, with the digits alternating between odd and even?''
In theory, we could increase the number of digits considered (i.e., change the base of the numbers constructed). However, we chose not to, and we still report performance on 27 combinatorics problems.

\begin{table}[htb]
    \centering
    \caption{Domain sizes and number of valid cells for each encoded MATH problem. Problem \textbf{193} (marked *) could not be scaled up.}
    \label{tab:math_domain_sizes}
    
    % Three groups of 3 columns (ID, Domain, Cells) separated by whitespace
    \begin{tabular}{ccc @{\hskip 0.4in} ccc @{\hskip 0.4in} ccc}
        \toprule
        \textbf{ID} & \textbf{$n$} & \textbf{$p$} & 
        \textbf{ID} & \textbf{$n$} & \textbf{$p$} & 
        \textbf{ID} & \textbf{$n$} & \textbf{$p$} \\
        \midrule
        
        7   & 8  & 6 & 53  & 11 & 7 & 175 & 7  & 3 \\
        8   & 10 & 6 & 80  & 6  & 3 & \textbf{193*} & \textbf{9} & \textbf{2} \\ % Emphasized
        23  & 8  & 4 & 82  & 6  & 3 & 231 & 8  & 4 \\
        29  & 7  & 5 & 96  & 6  & 3 & 240 & 9  & 9 \\
        31  & 4  & 3 & 99  & 6  & 3 & 243 & 7  & 2 \\
        33  & 8  & 4 & 102 & 8  & 9 & 269 & 8  & 4 \\
        39  & 8  & 4 & 130 & 11 & 9 & 284 & 8  & 2 \\
        42  & 9  & 4 & 137 & 8  & 3 & 292 & 7  & 4 \\
        45  & 8  & 4 & 149 & 10 & 2 & 309 & 5  & 3 \\
        47  & 7  & 2 &     &    &   &     &    &   \\ % Last row for the first group only
        
        \bottomrule
    \end{tabular}
\end{table}

\subsection{Proofs of Matchings to OEIS}
\begin{lemma}
Given a fixed linear ordering, $\symbfomc(\Phi_3, n)$ equals the $n$-th term of A000670, which is the number of ways to partition $n$ elements into disjoint subsets and arrange them into a sequence.
\end{lemma}

\begin{proof}
The sequence can be expanded to a permutation in the way that for each subset, we write down the numbers in descending order.
For $\Phi_3$ when fixing the linear order $\lopred$, the successor relation $S$ represents a permutation of domain elements.
The binary predicate $B(x,y)$ indicates whether two adjacent elements $x$ and $y$ belong to the same subset.
According to the expansion, if $(x\lopred y)$, they must belong to different subsets.
However, if $(y\lopred x)$, both cases are allowed.
Therefore, there is a bijective mapping from the models of $\Phi_3$ to arrangements of partitions.
\end{proof}

\begin{lemma}
Given a fixed linear ordering, $\symbfomc(\Phi_4, n)$ equals the $n$-th term of A000629, which is the number of ways to partition $n+1$ elements to disjoint subsets and arrange them into a necklace.
\end{lemma}

\begin{proof}
One can always break the necklace into a sequence by picking the subset containing element $n+1$ and concatenating the subsets in a clockwise direction.
In $\Phi_4$, we color each element by two colors: $U(x)$ being true and $U(x)$ being false.
One can construct the bijective mapping from the models of $\Phi_4$ to arrangements of partitions by requiring that two adjacent elements are in the same subset if and only if they have the same color.
Since $n+1$ is always the first element in the permutation, we fix its color as $U(x)=\top$, omit it, and count the colorings in the remaining $n$ elements.
\end{proof} 

\begin{lemma}
    Given a fixed linear ordering, $\symbfomc(\Phi_{cards}, n)$ equals the $n$-th  term of A002464, which is the number of permutations of length $n$ without rising or falling successions (also known as Hertzsprung's problem, i.e., ways to arrange $n$ non-attacking kings on an $n \times n$ chessboard, with one king in each row and column).
\end{lemma}

\begin{proof}
When  fixing the linear order $\lopred$ for $\Phi_{cards}$, the successor relation $S$ represents a permutation of domain elements.
When we have a pair of neighboring elements in the permutation, i.e., $S(x,y)$, then the sentence $\Phi_{cards}$ explicitly prohibits that $x$ was either the immediate predecessor of the immediate successor of $y$, which directly translates to no rising and no falling successions in the permutation.
\end{proof}

\end{document}